\documentclass[reqno]{amsart}

\usepackage[foot]{amsaddr}

\usepackage[UKenglish]{babel}

\usepackage[utf8]{inputenc}

\usepackage{microtype}
\usepackage{relsize}
\usepackage{hyperref}
\usepackage{amssymb}

\usepackage[capitalise,noabbrev,nameinlink]{cleveref}

\newenvironment{sendappendix}{}{}
\newenvironment{ifappendix}{\killcontents}{\endkillcontents}
\newenvironment{ifnotappendix}{}{}
\newif\ifsubmission
\submissiontrue
\newcommand\appref[1]{the full paper}

\usepackage[capitalise,noabbrev]{cleveref}
\usepackage{thmtools,thm-restate}
\usepackage{enumerate}
\theoremstyle{plain}
\newtheorem{theorem}{Theorem}[section]
\newtheorem{lemma}[theorem]{Lemma}

\newtheorem{corollary}[theorem]{Corollary}
\theoremstyle{definition}

\newtheorem{example}[theorem]{Example}
\theoremstyle{remark}

\newtheorem{fact}[theorem]{Fact}
\crefname{lemma}{Lemma}{lemmas}
\crefname{definition}{Definition}{definitions}
\crefname{fact}{Fact}{facts}
\Crefname{fact}{Fact}{Facts}
\crefname{enumi}{item}{items}
\Crefname{enumi}{Item}{Items}
\crefname{claim}{Claim}{claims}
\crefname{equation}{Equation}{equations}

\usepackage{natbib} 
\providecommand{\doi}[1]{\href{https://doi.org/#1}{\nolinkurl{doi:#1}}}

\title{When Locality Meets Preservation}
\author{Aliaume Lopez}
\address{Universit\'e Paris-Saclay, ENS Paris-Saclay, CNRS, LMF, France}
\address{Universit\'e Paris Cit\'e, CNRS, IRIF, France}
\email{aliaume.lopez@ens-paris-saclay.fr}

\usepackage{import}
\usepackage{pdfpages}
\usepackage{transparent}
\usepackage{xcolor}

\usepackage{amsthm}
\usepackage{amsmath}
\usepackage{booktabs}
\usepackage{stmaryrd}
\usepackage{faktor}
\usepackage{paralist}

\usepackage{upgreek}

\usepackage{mathtools}
\usepackage{thmtools}
\usepackage{thm-restate}

\usepackage{tikz}
\usetikzlibrary{positioning}
\usetikzlibrary{patterns}
\usepackage{tikz-cd}
\usetikzlibrary{backgrounds}

\newcommand\defined{\triangleq}

\newcommand\ltrsk{\L{}o\'s-Tarski Theorem}

\newcommand{\puext}{preservation under extensions}

\newcommand\FO{\mathsf{FO}}
\newcommand\MSO{\mathsf{MSO}}

\newcommand\Mod{\operatorname{Struct}}
\newcommand\ModF{\operatorname{Fin}}

\newcommand\td{\operatorname{td}}

\newcommand{\modset}[1]{\left\llbracket #1 \right\rrbracket}

\newcommand{\Balls}[3]{\operatorname{Balls}(#1,#2,#3)}

\newcommand{\Neighb}[3]{\mathcal{N}_{#1}(#2,#3)}

\newcommand{\Specter}[3]{\operatorname{Types}^{#2,#3}_{#1}}
\newcommand{\SpecterM}[3]{\operatorname{MTypes}^{#2,#3}_{#1}}

\newcommand{\rk}{\operatorname{rk}}
\newcommand{\NeighbT}[4]{\operatorname{tp}_{#1}^{#4}\left(#2,#3\right)}
\newcommand{\NeighbM}[4]{\operatorname{mtp}_{#1}^{#4}\left(#2,#3\right)}

\newcommand{\setof}[2]{\left\{ #1 \mid #2 \right\}}

\renewcommand{\vec}[1]{\mathbf{#1}}

\newcommand{\locleq}[3]{\mathbin{{\Rrightarrow}^{{#1},{#3}}_{#2}}}

\usepackage{environ}
\NewEnviron{killcontents}{}

\usepackage{tcolorbox}
\newtcolorbox{questions}{colback=green!5!white,colframe=green!75!black,fonttitle=\bfseries,title=Questions}
\newtcolorbox{warnings}{colback=orange!5!white,colframe=orange!75!black,fonttitle=\bfseries,title=Warning}
\newtcolorbox{insight}{colback=blue!5!white,colframe=blue!75!black,fonttitle=\bfseries,title=Insight}

\definecolor{myBlue}{HTML}{88C0D0}
\definecolor{myDarkBlue}{HTML}{5E81AC}
\definecolor{myTeal}{HTML}{8FBCBB}
\definecolor{myOrange}{HTML}{D08770}
\definecolor{myGreen}{HTML}{A3BE8C}
\definecolor{myRed}{HTML}{BF616A}
\definecolor{myYellow}{HTML}{EBCB8B}
\definecolor{myPurple}{HTML}{B48EAD}
\definecolor{myBlack}{HTML}{3B4252}
\definecolor{Prune}{RGB}{99,0,60}

\usepackage{varwidth}

\usepackage[ruled]{algorithm2e} 

\newcommand{\BadAxiom}{\mathcal{A}_{\textnormal{ord}}}
\newcommand{\BadOrder}{\modset{\mathcal{A}_{\textnormal{ord}}}_{\ModF(\upsigma)}}
\newcommand{\BadOrderS}{\mathcal{C}_{\textnormal{ord}}}

\newcommand{\XClass}{\mathcal{C}}

\newcommand{\CrefAbbrevNames}
  {%
    \crefname{figure}{Fig.}{Figs.}%
    \crefname{theorem}{Thm.}{Thms.}%
    \crefname{lemma}{Lem.}{Lemms.}%
    \crefname{corollary}{Cor.}{Cors.}%
  }
  \newcommand{\shortcref}[1]
  {{%
    \CrefAbbrevNames
    \cref{#1}%
  }}

\newcommand{\badphi}{\varphi_{B}}

\usepackage{pifont}

\newcommand{\localisable}{\ding{33}}
\newcommand{\locreduce}{\ding{91}}

\makeatletter
\let\oldparagraph\paragraph
\newcommand{\@paragraphstar}[1]{\oldparagraph*{\textbf{#1}}}
\newcommand{\@paragraphnostar}[1]{\oldparagraph{\textbf{#1}}}
\renewcommand{\paragraph}{\@ifstar{\@paragraphstar}{\@paragraphnostar}}
\makeatother

\begin{document}

\begin{abstract}
 This paper investigates the expressiveness of a fragment of
first-order sentences in Gaifman normal form, namely the positive
Boolean combinations of basic local sentences.  We show that they
match exactly the first-order sentences preserved under local
elementary embeddings, thus providing a new general preservation
theorem and extending the \L{}\'os-Tarski Theorem.

This full preservation result fails as usual in the finite, and we
show furthermore that the naturally related decision problems are
undecidable.  In the more restricted case of preservation under
extensions, it nevertheless yields new well-behaved classes of finite
structures: we show that preservation under extensions holds if and only if
it holds locally.

\end{abstract}


\clearpage

\maketitle

\section{Introduction}
\label{sec:intro}
\paragraph*{Preservation theorems.}
In classical model
theory, preservation theorems characterise
first-order definable sets enjoying some semantic property as those
definable in a suitable syntactic fragment
\cite[e.g.,][Section~5.2]{chang1990model}.
A well-known instance of a preservation theorem is
the \L{}o\'s-Tarski Theorem~\cite{tarski54,los55}: a first-order
sentence~$\varphi$ is preserved under extensions over all structures---i.e.,
$A\models\varphi$ and $A$ is an induced substructure of~$B$ imply
$B\models\varphi$---if and only if it is equivalent to an existential
sentence.
Similarly, the Lyndon Positivity Theorem applied to a unary predicate $X$~\cite{lyndon59}
connects surjective homomorphisms that are \emph{strong}
in every predicate except $X$
to sentences that are positive in $X$.
These two preservation theorems can be seen as
bridges between syntactic and semantic fragments in~\cref{fig:intro:summary}.

As most of classical model theory,
preservation theorems typically rely on compactness,
which is known to fail in the finite case.
Consequently,
preservation theorems generally do
not relativise to classes of structures, and in particular to the
class $\ModF(\upsigma)$ of all finite structures~[see the discussions
in \citenum{rosen02}, Section~2 and \citenum{kolaitis07},
Section~3.4]. For example, Lyndon's Positivity Theorem fails even
over the class of finite words~\cite{Kuperberg21}.
Nevertheless,
there are a few known instances of classes of finite structures where
some preservation theorems
hold~\cite{atserias2006preservation,Rossman08,Rossman16,
	harwath2014preservation}, and this type of question is still actively
investigated~\cite[e.g.][]{flum2021,Kuperberg21,dawar20}.

In the case of `order-based' preservation theorems like
the \L{}o\'s-Tarski Theorem, that hinge on the interplay between the
properties of the ordering and those of first-order logic, one can
typically focus on either of these two aspects.

Focusing on the
ordering, one can ask for the class $\XClass$ of
considered structures to
be \emph{well-quasi-ordered} with respect to
the induced substructure ordering (hereafter written
$\subseteq_i$).
In order-theoretic terms,
the set $\modset{\varphi}_{\XClass}$
of models
of a sentence $\varphi$ in $\XClass$
is \emph{upwards-closed} whenever
$\varphi$ is preserved under extensions.
The assumption that $\XClass$ is well-quasi-ordered implies that
$\modset{\varphi}_\XClass$ has finitely many minimal elements
for~$\subseteq_i$ and this is well known to imply {\puext}.  Practical
instances where~$\subseteq_i$ gives rise to a well-quasi-order are
scarce; it is the case for graphs of bounded tree-depth \citep*{Ding92}, but
to our knowledge the characterisation
by \citet*{Daligault2010} of which classes of bounded clique-width are
well-quasi-ordered yields the broadest known such class
(see the left column in \cref{fig:considered_classes}).

Focusing on first-order logic, one can leverage finite model theory results over classes of
structures provided they are somewhat locally well-behaved; this can be ensured
for instance by some flavour of \emph{sparsity}.
Two instances of
this are classes of bounded degree and the class of all graphs of
treewidth less than~$k$~\cite{atserias2008preservation}.  This is the
direction we take in this paper.

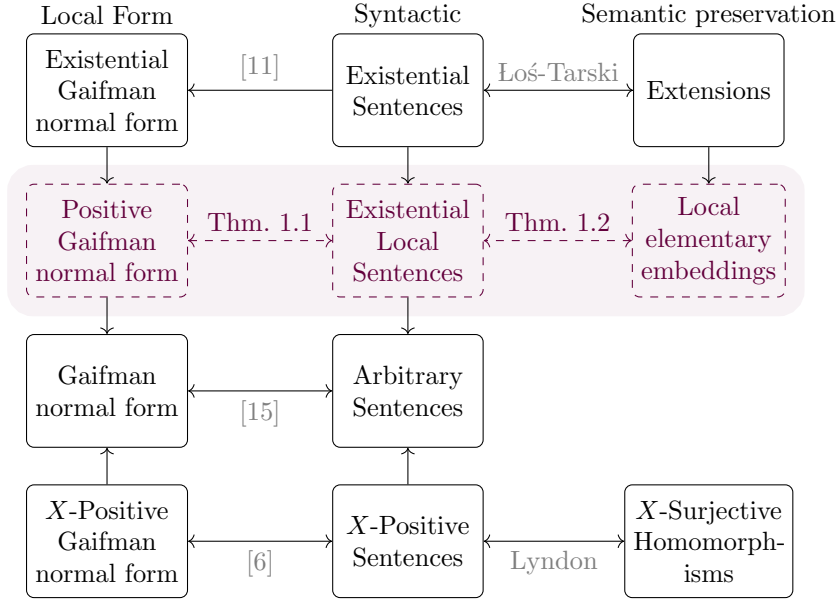
\begin{figure}[tbh]
	\centering
	\begin{tikzpicture}[xscale=2,yscale=1,
			mynode/.style={
					draw,
					minimum height=1.5cm,
					minimum width=2cm,
					execute at begin node={\begin{varwidth}{2cm}\centering},
								execute at end node={\end{varwidth}},
					rounded corners=1mm,
				},
			newresult/.style={
					Prune,
					dashed
				}
		]
		\begin{scope}
			\node[mynode] (GNF) at (0,2) {
				\centering Gaifman normal form};
			\node[mynode] (PxGNF) at (0,0)
			{$X$-Positive
				Gaifman normal form};
			\node[mynode,newresult] (PGNF) at (0,4) {Positive Gaifman normal form};
			\node[mynode] (ENF) at (0,6) {
				Existential Gaifman normal form};
			\draw[->] (ENF) -- (PGNF);
			\draw[->] (PGNF) -- (GNF);
			\draw[->] (PxGNF) -- (GNF);
			\node (C) at (0,7) {Local Form};
		\end{scope}
		\begin{scope}[xshift=4cm]
			\node[mynode] (Px) at (0,0) {
				$X$-Surjective Homomorphisms};
			\node[mynode,newresult] (Up) at (0,4) {
				Local elementary embeddings};
			\node[mynode] (Ext) at (0,6) {

				Extensions};
			\draw[->] (Ext) -- (Up);
			\node (B) at (0,7) { Semantic
				preservation};
		\end{scope}
		\begin{scope}[xshift=2cm]
			\node[mynode] (FO) at (0,2) {

				Arbitrary Sentences};
			\node[mynode] (PxFO) at (0,0) {

				$X$-Positive Sentences};
			\node[mynode,newresult] (LFO) at (0,4) {

				Existential Local Sentences};
			\node[mynode] (EFO) at (0,6) {

				Existential Sentences};
			\draw[->] (EFO) -- (LFO);
			\draw[->] (LFO) -- (FO);
			\draw[->] (PxFO) -- (FO);
			\node (A) at (0,7) {Syntactic};
		\end{scope}
		\draw[<->] (GNF) -- node[gray,midway,
			below]{[\citenum{libkin2012elements}]} (FO);

		\draw[<->] (PxGNF) --
		node[gray,midway,below]{[\citenum{dawar2006approximation}]} (PxFO);
		\draw[<->] (Px) --
		node[gray,midway,below]{Lyndon}
		(PxFO);

		\draw[->] (EFO) --
		node[gray,midway,above]{[\citenum{grohe2004existential}]}
		(ENF);

		\draw[<->] (Ext) --
		node[gray,midway,above]{\L{}o\'s-Tarski}
		(EFO);

		\draw[<->,newresult] (Up) --
		node[gray,midway,Prune,above]{
			\shortcref{lem:ccl:existlocalnfinfinite}
		} (LFO);

		\draw[<->,newresult] (PGNF) -- node[gray,Prune,midway,above]
		{\shortcref{thm:ccl:mongaifman}} (LFO);

		\begin{scope}[on background layer]
			\draw[line width=0.5cm, Prune!5,
				fill=Prune!5,
				rounded corners=1mm,
				line cap=round]
			(PGNF.south west) -- (Up.south east) --
			(Up.north east) -- (PGNF.north west) -- cycle;
		\end{scope}

	\end{tikzpicture}
	\caption{Comparison of the expressiveness of different fragments
		of {\boldmath $\FO[\upsigma]$} over \emph{general structures}.}
	\label{fig:intro:summary}
\end{figure}

\paragraph*{Locality.}
Given a structure $A$ over a finite relational signature $\sigma$,
its \emph{Gaifman
	graph} has the elements of $A$ as vertices and an edge $(a, b)$
whenever both $a$ and $b$ are in relation in~$A$. The distance $d_A(a,
	b)$ between two elements of~$A$ is their distance in the Gaifman graph
of~$A$.  For a tuple $\vec{a} \in A$, and a radius $r \in \mathbb{N}$
one can consider the $r$-neighborhood around $\vec{a}$ in $A$ defined
as $\Neighb{A}{\vec{a}}{r}
	\defined
	\setof{a' \in A}{\exists a \in \vec{a}, d_A(a,a') \leq r}
	=
	\bigcup_{a \in \vec{a}} \Neighb{A}{a}{r}
$; we emphasise that this union is not required to be disjoint.
Slightly abusing notations, we identify the set
$\Neighb{A}{\vec{a}}{r}\subseteq_i A$ with the corresponding induced
substructure of~$A$.

A first-order formula $\varphi(\vec{x})$ is said to
be \emph{$r$-local} if its evaluation over a structure $A$ and a tuple
$\vec{a}$ from $A$ only depends on the $r$-neighborhood of $\vec{a}$
in $A$, i.e.  $A, \vec{a} \models \varphi$ if and only if
$\Neighb{A}{\vec{a}}{r}, \vec{a} \models \varphi$.  In the particular
case of $r = 0$, local sentences are equivalent to quantifier free
sentences.  Since $\upsigma$ is a finite relational signature, every
formula $\varphi$ can be \emph{relativised} to a $r$-local formula,
hereafter denoted by $\varphi_{\leq r}$, which coincides with
$\varphi$ over neighborhoods of size $r$, although $\varphi_{\leq r}$
might have a higher quantifier rank.

It is more delicate to craft a notion of locality
for first order sentences
since they have no free variables.
The classical approach is to consider
\emph{basic local sentences} of the form
$\exists \vec{x}. \bigwedge_{i \neq j} d(x_i, x_j) > 2r \wedge
	\bigwedge_{i} \psi_{\leq r}(x_i)$
where $\psi_{\leq r}$
is a formula relativised to the $r$ neighborhood
of its single free variable.
Simply
put, the evaluation of a local $r$-basic sentence
is determined solely by the evaluation of
$\psi_{\leq r}$ over \emph{disjoint}
neighborhoods of radius $r$.
Note that the predicates $d(x, y) \leq r$
and
$y \in \Neighb{}{\vec{x}}{r}$
are definable for each fixed $r \in \mathbb{N}$.

The Gaifman Locality Theorem~\cite{gaifman1982local}
states
that every first-order sentence is equivalent to a Boolean combination of
basic local sentences.
This can be thought of as
$\FO[\upsigma]$
being limited to describing the
\emph{local behaviour} of structures.
In the study of preservation theorems
such as the {\ltrsk} over finite structures,
a first step is often to use Gaifman's normal form
and rely on the structural properties of
models
as a substitute for compactness~\cite{atserias2008preservation,atserias2006preservation,harwath2014preservation,dawar20}.

In~\cref{fig:intro:summary} several variants
of the Gaifman normal form are listed
along with their correspondence to
syntactic fragments of first-order logic.
For instance, in order to provide a faster
model checking algorithm,
\citeauthor{grohe2004existential}
crafted an \emph{existential Gaifman normal form}
as follows:
every existential first order sentence is equivalent to a \emph{positive}
Boolean combination of
\emph{existential} basic local sentences~\cite{grohe2004existential},
where an existential
basic local sentence is of the form
$\exists \vec{x}. \bigwedge_{i \neq j} d(x_i, x_j) > 2r \wedge
	\bigwedge_{i} \psi_{\leq r}(x_i)$
and $\psi_{\leq r}$
is an \emph{existential} formula relativised to the $r$-neighborhood
of its one free variable.
Note that this not an equivalence:
$\exists x. \exists y. d(x,y) > 2$ is in existential Gaifman normal form
but not preserved under extensions nor equivalent
to an existential sentence.
Another variant
of the Gaifman normal form,
tailored to the study of Lyndon's Positivity Theorem
over one unary predicate $X$
is
defined by \citet*[Theorem 2]{dawar2006approximation}.
It allows the authors to provide approximation schemes
for evaluating sentences positive in one unary predicate.

\subsection{Contributions}

Our first main contribution is a new line of equivalent characterisations
through local normal forms, syntactic restrictions, and preservation under
some ordering, as embodied in \cref{fig:intro:summary}.
This generalises the
correspondences between
the
existential Gaifman normal forms
defined by \citeauthor*{grohe2004existential},
existential sentences, and preservation under extensions.

\subsubsection{Existential Local Sentences and Positive Gaifman Normal Forms}

This paper studies a positive variant
of locality through the prism of
existential closures of $r$-local formulas,
abbreviated here as \emph{existential local sentences}. Those
are of the form
$\exists \vec{x}. \tau(\vec{x})$
where $\tau$ is an $r$-local formula.

As opposed to basic local sentences,
\emph{existential local sentences}
allow
interaction between the existentially quantified variables,
which increases their expressiveness.
Existential local sentences
also generalise existential sentences,
as quantifier-free formulas are $0$-local
(conversely, $0$-local formulas
can be rewritten as quantifier-free formulas).
Thus, by allowing formulas with a non-zero
locality radius, we provide a middle ground
between existential sentences and arbitrary sentences.

We prove that existential local sentences
and \emph{positive} Boolean combinations of basic local sentences
are equally expressive,
regardless of the class of structures considered.
This theorem, proven in \cref{sec:pgaif},
is related to the existential Gaifman normal form
of~\citet{grohe2004existential}
and relies, in part, on similar combinatorial arguments.

\begin{restatable}[Positive Locality]{theorem}{positivelocality}
	\label{thm:ccl:mongaifman}
	Let $\XClass \subseteq \Mod(\upsigma)$
	be a class of structures
	and
	$\varphi \in \FO[\upsigma]$ be a first-order sentence.
	The sentence $\varphi$ is equivalent over $\XClass$
	to an existential local sentence if and only if
	it is equivalent over $\XClass$ to a positive Boolean
	combination of basic local sentences.
\end{restatable}

\subsubsection{A Semantic Characterisation of Existential Local Sentences}
Recall that a map $h \colon A \to B$
is an \emph{elementary embedding} whenever
for every first order sentence $\varphi(\vec{x})$
and every $\vec{a} \in A^k$
$A, \vec{a} \models \varphi$ if and only if
$A, h(\vec{a}) \models \varphi$.
A localised notion of elementary embeddings is obtained
as follows:
$h \colon A \to B$
is a \emph{local elementary embedding}
when for every $k\geq 1$, $r\geq 0$, $\vec{a} \in A^k$,
and $r$-local formula $\varphi$,
$A, \vec{a} \models \varphi$
if and only if $B, h(\vec{a}) \models \varphi$.
This definition is a strengthening of the \emph{induced substructure}
ordering,
defined by $A \subseteq_i B$ whenever there exists
an injective morphism $h \colon A \to B$
such that for all relations $R \in \upsigma$
and elements $\vec{a} \in A^k$,
$A, \vec{a} \models R(\vec{x})$
if and only if
$B, h(\vec{a}) \models R(\vec{x})$.

We define another natural
ordering on structures by writing
$A \locleq{r}{q}{k} B$
whenever,
for all $\varphi = \exists x_1, \dots x_k. \tau(\vec{x})$
where $\tau$ is an $r$-local formula of quantifier rank at most $q$,
$A \models \varphi$ implies $B \models \varphi$.
An existential local sentence
is naturally preserved under $\locleq{r}{q}{k}$
for some $r,q \geq 0$ and $k \geq 1$.
We define the limit of those preorders as
\begin{equation}
	\locleq{\infty}{\infty}{\infty}
	\defined
	\bigcap_{r \geq 0}
	\bigcap_{q \geq 0}
	\bigcap_{k \geq 1}
	\locleq{r}{q}{k} \, .
\end{equation}

\begin{restatable}[Local preservation]{theorem}{existlocalinfinite}
	\label{lem:ccl:existlocalnfinfinite}
	Let $\varphi$ be a sentence in $\FO[\upsigma]$.
	The following properties are equivalent
	over the class of all structures $\Mod(\upsigma)$.
	\begin{enumerate}[(a)]
		\item The sentence $\varphi$ is equivalent to an existential local
		      sentence.
		\item There exist $r,q,k \in \mathbb{N}$ such
		      that $\varphi$ is preserved under $\locleq{r}{q}{k}$.
		\item The sentence $\varphi$ is
		      preserved under $\locleq{\infty}{\infty}{\infty}$.
		\item The sentence $\varphi$
		      is preserved under local elementary embeddings.
	\end{enumerate}
\end{restatable}

\subsubsection{Non-relativisation in the Finite}

The proof of \cref{lem:ccl:existlocalnfinfinite} does not relativise
to classes of finite structures, except for the equivalence
$(a) \Leftrightarrow (b)$.  Our second main contribution
in \cref{sec:ffc} is to show that \cref{lem:ccl:existlocalnfinfinite}
fails over $\ModF(\upsigma)$ the class of finite
structures, and to characterise for which parameters
$(r,k,q)$ preservation under $\locleq{r}{q}{k}$ leads to an
existential local form over $\ModF(\upsigma)$.

The picture is even bleaker when applying the
methodology of \citet*{flum2021} for the \L{}o\'s-Tarski Theorem
or \citet*{Kuperberg21} for Lyndon's Positivity Theorem: we show that most decision problems
ensuing this failure are undecidable.
Namely, we show in \cref{sec:ff:undecidable}
that
it is not possible to decide whether
a sentence is preserved under local elementary embedding
in the finite setting, nor is it possible
to decide whether a sentence preserved under elementary embedding
is equivalent to an existential local one,
and even under the promise that
the sentence is equivalent to an existential local one
such an equivalent sentence is not computable.

\subsubsection{Application to Preservation Under Extensions}
We use our understanding of existential local sentences
to split the proof of preservation under extensions
over a class $\XClass$
in two distinct steps:
\begin{inparaenum}
	\item[(\localisable)]
        $\XClass$ is \emph{localisable}, i.e.
	sentences preserved under extensions over $\XClass$
	are equivalent to existential local sentences $\XClass$,
	\item[(\locreduce)]
        $\XClass$ satisfies \emph{existential local preservation under extensions}, i.e.
        existential local sentences preserved under extensions
        over $\XClass$
        are equivalent to existential sentences over $\XClass$.
\end{inparaenum}
A class $\XClass$ that satisfies both (\localisable)
and (\locreduce) satisfies preservation under extensions, hence
this proof scheme is correct.
Moreover, existential sentences are existential local therefore
this proof scheme is complete.
\Cref{sec:presthm} is devoted to exploring the new classes of structures
where preservation under extensions holds that are gained
through this proof scheme. This is done by providing a finer
understanding of the interaction between locality and
combinatorial arguments in preservation under extensions.

(\localisable) Let us say that a class $\XClass \subseteq \ModF(\upsigma)$ is closed under
disjoint unions when the disjoint union of one structure
from~$\XClass$ and a finite structure from~$\ModF(\upsigma)$ remains in~$\XClass$.
When the class $\XClass$ is stable under induced substructures, we say
that $\XClass$ is \emph{hereditary}.
We prove in \cref{sec:presthm:inducedsub} that
hereditary classes of finite structures
closed under disjoint unions
are \emph{localisable} (\cref{lem:failure:inducedpres}). 
This is the diagonal edge labelled (\localisable) in~\cref{fig:intro:diagonal},
that this bypasses the non-relativisation of
\cref{lem:ccl:existlocalnfinfinite} in the finite.
As a consequence,
the class $\ModF(\upsigma)$ of finite structures
satisfies (\localisable).
Moreover, we prove in \cref{sec:locqo} that
closure under local elementary embeddings
coincides with closure under disjoint unions in the case of finite
structures.
This highlights the crucial use of
local elementary embeddings
in the literature~\cite{atserias2006preservation,
	atserias2008preservation, Rossman08},
in the guise of disjoint unions.

(\locreduce) 
We prove in \cref{sec:presthm:localdecstruct} that
a hereditary class $\XClass$ satisfies (\locreduce)
if and only if
its ``local neighbourhoods'' $\Balls{\XClass}{r}{k}$ satisfy preservation under extensions
for $r,k \geq 0$ (\cref{lem:pres:locwbpres}).
We construct the local neighbourhoods of a class
$\XClass$ by collecting the neighbourhoods around $k$ points
in structures of $\XClass$ as follows:
\begin{equation}
	\Balls{\XClass}{r}{k}
	\defined \setof{ \Neighb{A}{\vec{a}}{r}}{
		A \in \XClass
		\wedge
		\vec{a} \in A^{\leq k}
	} \, .
\end{equation}
This allows to bridge the gap between
existential sentences and existential local sentences
preserved under extensions in \cref{fig:intro:diagonal}
via the edge labelled (\locreduce).

\bigskip

The remaining of \cref{sec:presthm:localwb}
is devoted to proving
that the combination of
(\localisable) \cref{lem:failure:inducedpres}
and
(\locreduce) \cref{lem:pres:locwbpres}
strictly generalise previously known properties
that imply preservation under extensions.
To that end, we study classes $\XClass$ that are hereditary,
closed under disjoint unions, and
``locally well-behaved'', i.e. such that
$\Balls{X}{r}{k}$ is ``well-behaved'' for all $r,k \geq 0$.
Instances of ``well-behaved'' are
finite classes, classes of bounded tree-depth,
or more generally classes that are well-quasi-ordered
with respect to $\subseteq_i$, as depicted in the left
column of \cref{fig:considered_classes}.
By localising these properties, we obtain the
right column of \cref{fig:considered_classes},
which still implies preservation under extensions, but
strictly improves previously known results, with the exception
of ``locally finite classes'' that coincides
with the ``wide classes`` of \citet[Theorem 4.3]{atserias2008preservation}.
This validates our proof scheme as we effectively decoupled
the locality of first-order logic (\localisable)
from the combinatorial behaviour considered
(\locreduce)
in our proofs of preservation under extensions.

\begin{figure}[ht]
	\centering
	\begin{tikzpicture}[
			xscale=1.5,
			yscale=1.5,
			mynode/.style={
					draw,
					minimum height=1.5cm,
					minimum width=2.5cm,
					execute at begin node={\begin{varwidth}{2.5cm}\centering},
								execute at end node={\end{varwidth}},
					rounded corners=1mm,
				},
			newresult/.style={
					Prune,
					dashed
				}
		]
		\begin{scope}[xshift=4cm]
			\node[mynode,newresult] (Up) at (0,0) {
				Local elementary embeddings};
			\node[mynode] (Ext) at (0,2) {

				Extensions};
			\draw[->] (Ext) -- (Up);
			\node (B) at (0,3) { Semantic
				preservation};
		\end{scope}
		\begin{scope}[xshift=0cm]
			\node[mynode,newresult] (LFO) at (0,0) {

				Existential Local Sentences};
			\node[mynode] (EFO) at (0,2) {

				Existential Sentences};
			\node (A) at (0,3) {Syntactic};
		\end{scope}

		\draw[->] (EFO) --
		(Ext);

		\draw[<->,thick,red] (LFO) --
		node[midway, left, gray] {
			(\locreduce)
			\shortcref{lem:pres:locwbpres}
			+ $\subseteq_i$
		}
		(EFO);

		\draw[->,thick,red] (LFO) --
		node[gray,midway,below]{
			\shortcref{cor:ff:counterex}
		} (Up);

		\draw[->,thick, red]
		(Ext) --
		node[gray,midway, xshift=-2mm, yshift=1mm]{
			\rotatebox{23}{
				(\localisable)
				\shortcref{cor:pthm:localisable}
			}
		}
		(LFO);

		\begin{scope}[on background layer]
			\draw[line width=0.5cm, Prune!5,
				fill=Prune!5,
				rounded corners=1mm,
				line cap=round]
			(Ext) -- (Ext.south west) -- (Ext.north west) -- (Ext.north
			east) -- (Ext.south east) --
			(Ext.south west) --
			(Ext) --  (LFO) --
			(LFO.north east) -- (LFO.south east) -- (LFO.south west) --
			(LFO.north west) -- (LFO.north east) -- (LFO)
			-- (EFO)
			-- (EFO.north west) -- (EFO.south west) -- (EFO.south east)
			-- (EFO.north east) -- (EFO.north west)
			;
		\end{scope}
	\end{tikzpicture}
	\caption{Comparison of the expressiveness of fragments
		of {\boldmath $\FO[\upsigma]$} over hereditary classes of finite structures
		stable under disjoint unions. Single headed arrows represent strict inclusions.}
	\label{fig:intro:diagonal}
\end{figure}
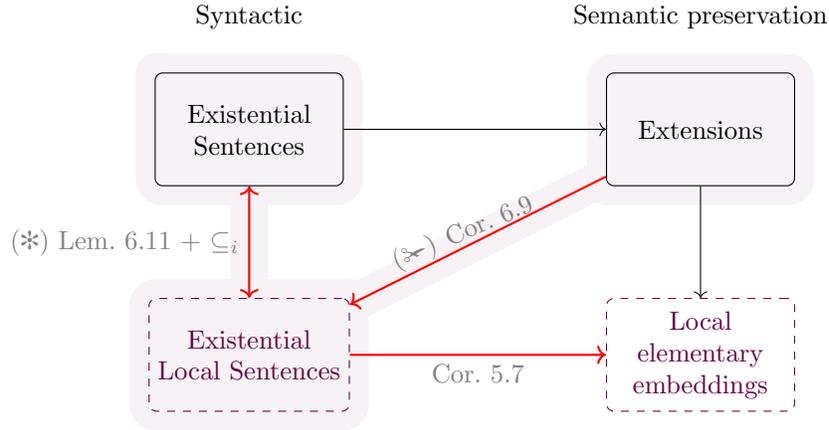

\begin{figure}[ht]
	\centering
	\begin{tikzpicture}[
			rotate=90,
			xscale=0.6,
			yscale=-3.5,
			mynode/.style={
					draw,
					minimum height=1cm,
					minimum width=3cm,
					execute at begin node={\begin{varwidth}{3cm}\centering},
								execute at end node={\end{varwidth}},
					rounded corners=1mm,
				},
			newresult/.style={
					Prune,
					dashed
				}
		]
		\node[mynode] (F) at (0,0) {finite};
		\node[mynode] (LF) at (0,2) {locally\\finite};

		\node[mynode] (T) at (3,0) {bounded treedepth};
		\node[mynode,newresult] (LT) at (3,2) {locally\\ bounded treedepth};

		\node[mynode] (W) at (6,0) {wqo};
		\node[mynode,newresult] (LW) at (6,2) {locally\\ wqo};

		\node[mynode] (P) at (9,0) {{pr.\ under extensions}};
		\node[mynode,newresult] (LP) at (9,2) {locally\\ {pr.\ under extensions}};

		\node[mynode] (WI) at (-3, 2) {wide};

		\draw[double equal sign distance,shorten >=0.1cm,shorten <=0.1cm]
		(LF) --
		node[midway, right] {\tiny (\labelcref{prop:presthm:finiteballs})}
		(WI);
		\draw[thick, Prune, double equal sign distance,shorten >=0.5cm,shorten <=0.5cm]
		(LP) --
		node[midway, above] {\tiny (\labelcref{prop:pres:localtrsk})}
		(P);
		\draw[->, thick, Prune] (F) -> (LF);
		\draw[->, thick, Prune] (LF) --
		node[midway, right] {\tiny (\labelcref{ex:diamond})}
		(LT);
		\draw[->, thick, Prune] (LT) --
		node[midway, right] {\tiny (\labelcref{ex:cliques})}
		(LW);
		\draw[->, thick, Prune] (LW) --
		node[midway, right] {\tiny (\labelcref{ex:cpp})}
		(LP);

		\draw[->] (F) --
		(T);
		\draw[->] (T) --
		node[midway, left] {\tiny \cite{Ding92}}
		(W);
		\draw[->] (W) -> (P);

		\draw[->, thick, Prune] (T) --
		node[midway, above] {\tiny (\labelcref{ex:diamond})}
		(LT);
		\draw[->, thick, Prune] (W) --
		node[midway, below] {\tiny (\labelcref{ex:diamond})}
		(LW);

		\begin{scope}[on background layer]
			\draw[line width=0.5cm, Prune!5,
				fill=Prune!5,
				rounded corners=1mm,
				line cap=round]
			(LT.south west) -- (LT.south east) --
			(LP.north east) -- (LP.north west) -- cycle;
		\end{scope}

	\end{tikzpicture}
	\caption{
		Implications of properties over \emph{hereditary} classes of
		\emph{finite} structures
		stable under disjoint unions.
	}
	\label{fig:considered_classes}
\end{figure}
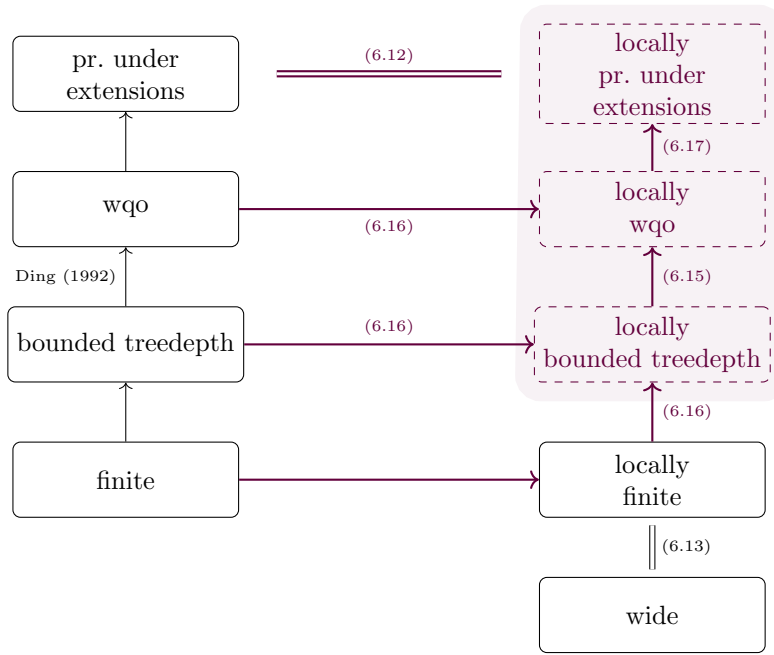

\section{Locality preorders}
\label{sec:locqo}As most of this paper is centered
around preorders of the form $\locleq{r}{q}{k}$,
we start by illustrating them
over several examples and
by relating them to well-known preorders
for specific values of $r,q$ and $k$.

Given a structure $A \in \Mod(\upsigma)$
and a tuple $\vec{a}$ of elements from $A$,
let us write $\NeighbT{A}{\vec{a}}{r}{q}$ for
the $(q,r)$-local type of $\vec{a}$,
i.e.,
the set of all formulas
of quantifier rank at most $q$ with $|\vec{a}|$ free variables
that are $r$-local such that $(A, \vec{a}) \models \varphi(\vec{x})$.
Note that there are only finitely many possible local types
for a given pair $(q,r) \in \mathbb{N}^2$ and a given number of variables~$|\vec{a}|$.
Those local types can be collected to fully describe the local
behaviour of~$A$ in
\begin{equation}
	\Specter{r}{q}{k}(A) \defined \setof{\NeighbT{A}{\vec{a}}{r}{q}}{ \vec{a} \in A^{\leq
				k}}\;.
\end{equation}
We shall often use this collection of types to
reason with our preorders.

\begin{fact}[Type collection]
	\label{fact:qo:specters}
	Let $r,q,k \in \mathbb{N}$ with $k \geq 1$.
	For all $A,B \in \Mod(\upsigma)$,
	$A \locleq{r}{q}{k} B$
	if and only if
	$\Specter{r}{q}{k}(A) \subseteq
		\Specter{r}{q}{k}(B)$.
\end{fact}
\begin{proof}
	Assume that $A \locleq{r}{q}{k} B$.  Let us write $T$ for the
	finite set (up to logical equivalence) of possible $r$-local
	formulas at radius $r$, quantifier rank $q$ and $k$ variables.
	Given a vector $\vec{a} \in A^k$, define
	$T_+ \defined \NeighbT{A}{\vec a}{r}{q}$ and $T_- \defined
		T \setminus T_+$. Those are finite collections of $r$-local
	formulas of quantifier rank at most $q$ with $k$ free variables
	$x_1, \dots x_k$.  Let us write
	$\psi(\vec{x}) \defined \bigwedge_{t \in T_+}
		t(\vec{x}) \wedge\bigwedge_{t \in T_-} \neg t(\vec{x})$, which is
	$r$-local and of quantifier rank at most $q$.  The structure $A$
	satisfies $\varphi \defined \exists \vec{x}. \psi(\vec{x})$ through
	the choice of the vector $\vec a$.  Because $A \locleq{r}{q}{k} B$,
	$B \models \varphi$ and this provides a vector $\vec b \in B^k$
	such that $B, \vec b \models \psi(\vec{x})$.  In turn, this proves that
	the sentences of quantifier rank at most $q$ and locality radius $r$
	that hold over $\vec b$ are exactly those in $T_+$.  Finally,
	$\NeighbT{B}{\vec b}{r}{q} = \NeighbT{A}{\vec a}{r}{q}$.

	Conversely, assume that
	$\Specter{r}{q}{k}(A) \subseteq \Specter{r}{q}{k}(B)$.  Let
	$\varphi$ be a sentence of the shape
	$\exists \vec{x}. \tau(\vec{x})$, where $|\vec{x}| = k$ and
	$\tau(\vec{x})$ is an $r$-local formula of quantifier rank at most
	$q$, that is true in $A$.  There exists a vector $\vec{a} \in A^k$
	such that $A, \vec a \models \tau(\vec{x})$, in particular
	$\tau(\vec{x}) \in \NeighbT{A}{\vec a}{r}{q}$.  Since
	$\Specter{r}{q}{k}(A) \subseteq \Specter{r}{q}{k}(B)$, there
	exists $\vec b \in B^k$ such that $\NeighbT{A}{\vec a}{r}{q}
		= \NeighbT{B}{\vec b}{r}{q}$.  Thus, $B, \vec
		b \models \tau(\vec{x})$ and therefore $B \models \varphi$.
\end{proof}

\begin{fact}[Refinement]
	\label{fact:qo:refinement}
	If $(q,r,k) \leq (q',r',k')$ component-wise,
	then
	$\locleq{r'}{q'}{k'} \subseteq \locleq{r}{q}{k}$.
\end{fact}

\begin{example}[Inequalities between two structures]
	We consider undirected graphs as structures in $\Mod(\upsigma)$
	where $\upsigma \defined\{ (E,2) \}$.
	Let $P_n$ be a finite path of size $n$
	and $C_m$ be a finite cycle of size $m$.

	We can prove that $C_m \locleq{r}{q}{k} P_n$
	whenever $n \geq m > k (2r + 1)$. Indeed,
	consider $k$ points of $C_m$ and their balls of radius $r$.
	As $m > k (2r +1) $, the unions of these balls
	exclude at least one point of $C_m$, and are thus
	a finite union of paths. As $n \geq m$, it is possible to
	find points in $P_n$ behaving similarly at radius $r$.

	However, as soon as $k, r$ and $q \geq 1$,
	$P_n$ is not below $C_m$
	for $\locleq{r}{q}{k}$. Indeed, it suffices to
	select an endpoint of the path $P_n$, and to assert that
	it is of degree one using a sentence of quantifier rank $1$
	evaluated at radius $1$. One cannot find a similar
	point in $C_m$ as all nodes are of degree two.
\end{example}

\begin{ifappendix}
	The connection between disjoint unions, local elementary embeddings
	and induced substructures is provided in the following
	lemmas and examples that are detailed in \cref{app:sec:locqo}.
\end{ifappendix}

\begin{restatable}{lemma}{lemmaelemloc}
	\label{lem:qo:elemloc}
	Let $A, B$ be two structures in $\Mod(\upsigma)$.
	If $h \colon A \to B$ is a local elementary embedding
	then $A \locleq{\infty}{\infty}{\infty} B$.
\end{restatable}
\begin{sendappendix}
	\begin{ifappendix}
		\lemmaelemloc*
	\end{ifappendix}
	\begin{proof}
		Fix $q,r,k \in \mathbb{N}$,
		and consider a tuple $\vec{a} \in A^k$.
		By construction, the tuple $h(\vec{a}) \in B^k$
		satisfies the same local $\FO$ formulas,
		and in particular, $\NeighbT{A}{\vec{a}}{r}{q}
			= \NeighbT{B}{h(\vec{a})}{r}{q}$, thus $A \locleq{r}{q}{k} B$.
	\end{proof}
\end{sendappendix}

\begin{restatable}[Preorders in the finite]{lemma}{preservationinthefininte}
	\label{lem:qo:locqofin}
	Let $A, B$ be two finite structures in $\ModF(\upsigma)$. The
	following statements are equivalent.
	\begin{enumerate}
		\item There exists $C$ such that $A \uplus C = B$.
		\item $A \locleq{\infty}{\infty}{\infty} B$.
		\item $A \locleq{r}{q}{\infty} B$ for some $r, q \geq 1$.
		\item There exists a local elementary embedding from $A$ to $B$.
	\end{enumerate}
\end{restatable}
\begin{proof}[Proof of $(1) \Rightarrow (2)$]
	Given a vector $\vec{a} \in A^k$, a radius $r$, and
	quantifier rank $q$,
	one notices that
	$\Neighb{A}{\vec a}{r}
		= \Neighb{A \uplus C}{\vec a}{r}$, hence
	$\NeighbT{A}{\vec{a}}{r}{q} = \NeighbT{B}{\vec{a}}{r}{q}$.
	We have proven that $A \locleq{\infty}{\infty}{\infty} B$.
\end{proof}
\begin{proof}[Proof of $(2) \Rightarrow (3)$] By definition.
\end{proof}
\begin{proof}[Proof of $(3) \Rightarrow (4)$]
	Let us consider $\vec{a} \in A^{|A|}$
	a vector containing all the points of $A$
	exactly once.
	There exists a vector $\vec{b} \in B^{|A|}$
	such that
	$\NeighbT{A}{\vec{a}}{1}{1} = \NeighbT{B}{\vec{b}}{1}{1}$.
	As a consequence, $\Neighb{B}{\vec{b}}{1} = \Neighb{B}{\vec{b}}{0}$
	since this equation holds for $\vec a$ and is expressible
	using one universal quantifier.
	The mapping $h \colon \vec{a} \mapsto \vec{b}$ is
	a local elementary embedding. Indeed,
	$r$-neighborhoods around $\vec a$ (resp. $\vec b$)
	are subsets of $\vec a$ (resp. $\vec b$). This
	proves that $r$-local formulas around $\vec a$ (resp. $\vec b$)
	can be rewritten as $0$-local, and we conclude using the equality
	of their $(1,1)$-local types.
\end{proof}
\begin{proof}[Proof of $(4) \Rightarrow (1)$]
	Let $h \colon A \to B$ be a local elementary embedding.
	Let $D$ be the substructure induced by $h(A)$ in~$B$; since
	$h$ is a local elementary embedding, $A$ and $D$ are isomorphic.
	Let us consider a point $c \in B \setminus D$.
	Assume by contradiction that
	there exists a relation containing both $c$
	and an element of $d\in D$. Let $\varphi(\vec x)$ be
	a $1$-local formula of quantifier rank $1$
	with $|A|$ free variables
	stating that there exists a point not in $\vec{x}$
	connected to some element of $\vec{x}$.
	Since $h$ is a local elementary embedding and $d \in h(A)$,
	$B, h(A) \models \varphi(\vec x)$
	and $A, A \models \varphi(\vec x)$. This is absurd,
	hence $B = D \uplus (B \setminus D) = A \uplus (B \setminus D)$.
\end{proof}

However, the existence of a
local elementary embedding is in general not equivalent
to $\locleq{\infty}{\infty}{\infty}$, as shown next.

\begin{restatable}[Preorder difference]{example}{exampleinfinitegrid}
	Over the signature of graphs, let $G$ be an infinite grid,
	and $G' \defined G \uplus G$.  There exists no local elementary
	embedding from $G'$ to $G$ but $G' \locleq{\infty}{\infty}{\infty}
		G$.
\end{restatable}
\begin{sendappendix}
	\begin{ifappendix}
		\exampleinfinitegrid*
	\end{ifappendix}
	\begin{proof}
		Assume by contradiction that there exists an elementary embedding $h$ from
		$G'$ to $G$. There exist two points $u, v$ in $G'$
		that are at an infinite distance,
		and their images $h(u), h(v)$ are at a finite distance $r$ in $G$.
		The sentence $\psi(x,y)$ stating that $x$ and $y$ are at distance
		less than $r$ is $r$-local around $x$ and $y$, hence
		is satisfied for $h(u)$ and $h(v)$ if and only if
		it is satisfied in $G'$ for $u$ and $v$; this is absurd.

		However, $G' \locleq{\infty}{\infty}{\infty} G$
		since given $r,q \geq 0$, $k \geq 1$ and a vector $\vec{u} \in (G')^k$
		one can find points in $G$ at
		long enough distances such that the neighborhoods coincide. \qedhere
	\end{proof}
\end{sendappendix}

Over finite structures, when the parameter values of $r,q$ or $k$ are
too small, one ends up with a preorder that is trivial, except for one
specific combination where we obtain the \emph{extension} preorder
$\subseteq_i$.

\begin{fact}[Trivial orders]
	Over finite structures, the following preorders are trivial, i.e.\
	every pair of structures is related: $\locleq{r}{q}{0}$ whenever
	$r, q \in \mathbb{N} \cup \{ \infty \}$, and $\locleq{\infty}{0}{1}$.
\end{fact}

\begin{restatable}[Extension preorder]{lemma}{extensionpreorder}
	Over finite structures and for all
	$q,r \in \mathbb{N} \cup \{ \infty \}$,
	${\locleq{0}{q}{\infty}}={\locleq{r}{0}{\infty}}={\subseteq_i}$.
\end{restatable}
\begin{sendappendix}
	\begin{ifappendix}
		\extensionpreorder*
	\end{ifappendix}
	\begin{proof}
		We first notice that $\locleq{0}{q}{\infty}$
		and $\locleq{r}{0}{\infty}$ cannot quantify
		over the neighborhoods of a vector $\vec{a}$
		since they either contain only $\vec{a}$
		or cannot quantify at all. This explains the first equality.

		\begin{enumerate}
			\item Assume $A \locleq{0}{q}{\infty} B$ for some $q \in \mathbb{N} \cup
				      \{ \infty \}$, consider
			      $k \defined |A|$
			      and let $\vec{a} \in A^k$ be the list of all elements in $A$.
			      Assume without loss of generality that $q$ is finite.
			      The inequality provides a vector $\vec{b} \in B^k$
			      such that $\NeighbT{A}{\vec{a}}{0}{q} = \NeighbT{B}{\vec{b}}{0}{q}$.
			      In particular, $\Neighb{A}{\vec{a}}{0} = A$ and
			      $\Neighb{B}{\vec{b}}{0} = \vec{b}$.
			      Since it is possible to express using quantifier free formulas
			      all the relations between elements in $A$,
			      the substructure of $B$ induced by $\vec{b}$
			      is isomorphic to $A$.
			      In particular, $A \subseteq_i B$.

			\item Assume that $A \subseteq_i B$ though a map $h$,
			      let $k \in \mathbb{N}$
			      and $\vec{a} \in A^k$.
			      Let us define $\vec{b} \defined h(\vec{a})$;
			      in particular, $\Neighb{A}{\vec{a}}{0} = A$ and
			      $\Neighb{B}{\vec{b}}{0} = h(\vec{a})$.
			      By definition of the map $h$,
			      $h(\vec{a})$ induces a substructure of $B$
			      that is isomorphic to the substructure induced by $\vec{a}$ in $A$.
			      This entails that the two structures are elementary equivalent,
			      hence that
			      $\NeighbT{A}{\vec{a}}{0}{q} = \NeighbT{B}{\vec{b}}{0}{q}$
			      for all $q \geq 0$.
			      We have proven that $A \locleq{0}{q}{k}$ for all $q,k \in
				      \mathbb{N}$, hence also for $q,k \in \mathbb{N} \cup \{ \infty \}$.
			      \qedhere
		\end{enumerate}
	\end{proof}
\end{sendappendix}

\section{Positive Gaifman Normal Form}
\label{sec:pgaif}The aim of this section is to provide
a connection between the positive variant of Gaifman normal forms
and existential local sentences.
The theorem and its proof are heavily inspired by the combinatorics
behind \citeauthor*{grohe2004existential}'s proof
of the existential Gaifman normal form. In particular,
the main issue arises from
finding points with disjoint neighborhoods.

As basic local sentences are existential local,
the only difficulty in~\cref{thm:ccl:mongaifman}
is converting an existential local sentence
into a positive Boolean combination of basic local sentences.
We split this transformation into intermediate syntactic steps
\begin{description}
    \item[Existential local] 
        $\exists \vec{x}. \psi(\vec{x})$
        where $\psi$ is an $r$-local formula.
    \item[Almost basic local]
        $\exists \vec{x}.
        \bigwedge_{i \neq j} d(x_i,x_j) > 2r
        \wedge
        \psi(\vec{x})$
        where $\psi$ is an $r$-local formula.
    \item[Asymmetric basic local]
        $\exists \vec{x}.
        \bigwedge_{i \neq j} d(x_i,x_j) > 2r
        \wedge
        \bigwedge_{i} \psi_i(x_i)$
        where $(\psi_i)_i$ is a family
        of $r$-local formulas with exactly one free variable.
    \item[Basic local]
        $\exists \vec{x}.
        \bigwedge_{i \neq j} d(x_i,x_j) > 2r
        \wedge
        \bigwedge_{i} \psi(x_i)$
        where $\psi$ is an
        $r$-local formula with exactly one free variable.
\end{description}
\begin{ifnotappendix}
    Recall that a formula $\varphi\left(\vec{x}\right)$ can be made
    $r$-local around $\vec{x}$ by relativising quantifications
    to $\Neighb{}{\vec{x}}{r}$. As a consequence, the classes
    above can equivalently be defined syntactically.

    To provide some intuition about the behaviour of
    theses syntactic restrictions interpolating between
    existential local sentences and basic local sentences,
    let us use examples in plain English using two colours ``green''
    and ``blue'', writing increasingly complex sentences.
    A basic local sentence can ask ``whether there exists at least
    $2$ disjoint neighbourhoods containing only green points'',
    which can be rewritten as
        $\exists x_1, x_2.
        d(x_1,x_2) > 2r
        \wedge \text{green}_r(x_1)
        \wedge \text{green}_r(x_2)$.
    An asymmetric basic local sentence
    can ask ``whether there exists $2$ disjoint neighbourhoods,
    one green and one blue'', 
    which can be rewritten as
        $\exists x_1, x_2.
        d(x_1,x_2) > 2r
        \wedge \text{green}_r(x_1)
        \wedge \text{blue}_r(x_2)$.
    The asymmetry comes from the ability to select a different
    property for each disjoint neighbourhood.
    An almost basic local sentence
    can ask ``whether there exists $2$ disjoint neighbourhoods
    that have the same colour'',
    which can be rewritten as
        $\exists x_1, x_2.
        d(x_1,x_2) > 2r
        \wedge (\text{green}_r(x_1)
        \wedge \text{green}_r(x_2))
        \vee (\text{blue}_r(x_1)
        \wedge \text{blue}_r(x_2))
        $.
    An existential local sentence
    can ask ``whether there exists $2$ points such that
    every green point connected to one is connected to the other'',
    which can be rewritten as
        $\exists x_1, x_2.
        \forall y \in \Neighb{}{x_1x_2}{r}.
        (\text{green}(y) \wedge
        E(x_1,y))
        \iff
        (\text{green}(y) \wedge
        E(x_2,y))
        $.
\end{ifnotappendix}

Asymmetric basic local sentences already appear as an intermediate
steps towards basic local sentences in the constructions
of~\citet{grohe2004existential} and \citet{dawar2006approximation}.  Most of the
transformations will rely on the description of the `spatial'
repartition of elements in a given structure $A$. We handle this
description through the following lemma%
\begin{ifappendix}
proven in \cref{app:sec:pgaif}%
\end{ifappendix}.%
\begin{ifnotappendix}
The intended purpose of \cref{lem:ccl:extentedcover}
is to build, from a neighbourhood $\Neighb{A}{\vec{a}}{r}$, a
disjoint union of neighbourhoods containing the former one
using a radius and a number of points controlled independently of the structure $A$.
\end{ifnotappendix}

\begin{restatable}{lemma}{extendedcover}
    \label{lem:ccl:extentedcover}
    For every $k, r \geq 0$,
    for every structure $A \in \Mod(\upsigma)$
    and vector $\vec{a} \in A^{\leq k}$
    there exists
    a vector $\vec{b} \in A^{\leq k}$
    and a radius $r \leq R \leq 4^k r$
    such that
    $\Neighb{A}{\vec{a}}{r} \subseteq 
    \Neighb{A}{\vec{b}}{R}$
    and $\forall b \neq b' \in \vec{b},
    \Neighb{A}{b}{3R} \cap \Neighb{A}{b'}{3R} = \emptyset$.
\end{restatable}
\begin{sendappendix}
    \begin{ifappendix}
        \extendedcover*
    \end{ifappendix}
\begin{proof}
    We proceed by induction over $k$.
    \begin{itemize}
        \item When $k \leq 1$ it suffices to take
            $R \defined r$ and for every vector $\vec{a} \in A^{\leq k}$
            build
            $\vec{b} \defined \vec{a}$
            and notice that
            $\Neighb{A}{\vec{a}}{r}
            = \Neighb{A}{\vec{b}}{R}$.

        \item When $k \geq 2$, we proceed by a simple case analysis
            \begin{enumerate}
                \item Either the balls $\Neighb{A}{a}{3r}$ are pairwise disjoint
                    when $a$ ranges over $\vec{a}$; 
                    in which case it suffices to consider $\vec{b} \defined
                    \vec{a}$ and $R \defined r$ to conclude.
                \item Or at least two of the balls $\Neighb{A}{a}{3r}$
                    intersect 
                    when $a$ ranges over $\vec{a}$,
                    and we can assume without loss of generality
                    that the neighborhoods of $a_1$ and $a_2$ at
                    radius $3r$ intersect.

                    Let us consider $c \in \Neighb{A}{a_1}{3r} \cap
                    \Neighb{A}{a_2}{3r}$.
                    Define $\vec{c} \defined (c, a_3, \ldots, a_{k})$,
                    this vector is of size $1$ when
                    $k = 2$.

                    Because $d(a_1, c) \leq 3r$ and $d(a_2, c) \leq 3r$,
                    $\Neighb{A}{a_1 a_2}{r} \subseteq
                    \Neighb{A}{c}{4r}$.

                    By induction hypothesis,
                    there exists a radius $4r \leq R \leq 4^{k-1} (4r)$
                    and a vector $\vec{b} \in A^{\leq k - 1}$
                    such that 
                    $\Neighb{A}{\vec{c}}{4r} \subseteq 
                    \Neighb{A}{\vec{b}}{R}$
                    and
                    $\forall b\neq b' \in \vec{b},
                    \Neighb{A}{b}{3R} \cap \Neighb{A}{b'}{3R} = \emptyset$.

                    Since
                    $\Neighb{A}{\vec{a}}{r} \subseteq 
                    \Neighb{A}{\vec{c}}{4r}$
                    and $r \leq 4r \leq R \leq 4^k r$
                    we concluded.
                    \qedhere
            \end{enumerate}
    \end{itemize}
\end{proof}
\end{sendappendix}

\subsection{From Existential Local Sentences to Asymmetric Basic Local Sentences}
\label{sec:pgaif:elstoabls}

Using~\cref{lem:ccl:extentedcover}
it is already possible
to transform an existential local sentence
into a positive Boolean combination
of \emph{almost} basic local sentences.

\begin{lemma}[From Existential local to Almost basic local]
    \label{lem:ccl:existlocal}
    Let $\varphi(\vec{x})$ be an $r$-local formula.
    There exist $1 \leq n \leq 4^{|\vec{x}|} r$ and
    $\psi_1, \dots, \psi_n$
    almost basic local sentences
    such that
    $\exists \vec{x}. \varphi(\vec{x})$
    is equivalent to
    $\bigvee_{1 \leq i \leq k} \psi_i$
    over $\Mod(\upsigma)$.
\end{lemma}
\begin{proof}
    Let us define $\Delta \defined \setof{(k,R)}{0 \leq k \leq |\vec{x}|
    \wedge r \leq R \leq 4^{|\vec{x}|} r}$
    and 
    \begin{equation}
        \begin{array}{rcl}
            \psi_{(k,R)} & \defined &
            \exists b_1, \dots, b_k.  \\
                                     & & \bigwedge_{1 \leq i < j \leq k} d(b_i, b_j) > 2R
            \\ \addlinespace
                                     & & \wedge\:\exists \vec{x} \in \Neighb{}{\vec{b}}{R}. \varphi(\vec{x})
        \end{array}
    \end{equation}

    To conclude it suffices to prove that
    $\exists \vec{x}. \varphi(\vec{x})$
    is equivalent to
    $\bigvee_{(k,R) \in \Delta}
    \psi_{(k,R)}
    \defined \Psi$.

    \begin{itemize}
        \item Assume that $A \in \Mod(\upsigma)$
            satisfies $\exists \vec{x}. \varphi(\vec{x})$.
            Then there exists $\vec{a} \in A^{|\vec{x}|}$
            such that $\Neighb{A}{\vec{a}}{r} \models \exists \vec{x}.
            \varphi(\vec{x})$ since $\varphi$ is $r$-local around $\vec{x}$.
            Using~\cref{lem:ccl:extentedcover},
            there exists a size $0\leq k \leq |\vec{a}|$
            a radius
            $r \leq R \leq 4^{|\vec{a}|} r$,
            and a vector $\vec{b} \in A^k$ such that
                $\Neighb{A}{\vec{a}}{r}
                \subseteq 
                \Neighb{A}{\vec{b}}{R}$
            and the balls of radius $3R$ around the
            points of $\vec{b}$ do not intersect.

            In particular,
            $\biguplus_{b \in \vec{b}} \Neighb{A}{b}{R} \models
            \exists \vec{x}. \varphi(\vec{x})$ since witnesses in
            $\Neighb{A}{\vec{a}}{r}$ can still be found 
            and $\varphi$ is $r$-local.
            This proves that $A \models \psi_{(k,R)}$
            hence that $A \models \Psi$.

        \item Assume that $A \in \Mod(\upsigma)$
            satisfies $\Psi$.
            Then there exists $(k,R)$ such that
            $A \models \psi_{k,R}$
            thus proving that 
            there exists $\vec{b} \in A^k$
            such that
            $\biguplus_{b \in \vec{b}} \Neighb{A}{b}{R} \models \exists
                \vec{x}. \varphi(\vec{x})$.
            Since $r \leq R$ and $\varphi$ is $r$-local
            this proves
            $A \models \exists \vec{x}. \varphi(\vec{x})$.
            \qedhere
    \end{itemize}
\end{proof}

An application of
the Feferman-Vaught technique~\cite{makowsky2004algorithmic, feferman1959first}
allows to transform \emph{almost} basic local sentences
into the \emph{asymmetric} ones introduced
by~\citet*{grohe2004existential}, i.e.\ to sentences of the form
$
\exists \vec{x}.
\bigwedge_{i \neq j} d(x_i, x_j) > 2r \wedge 
\bigwedge_{i} \psi_{i}(x_i)
$
where each $\psi_i$ is a $r$-local formula
when $i$ ranges from $1$ to $|\vec{x}|$.
The combination of this transformation%
\begin{ifappendix}
(detailed in~\cref{app:sec:pgaif:elstoabls})%
\end{ifappendix}
with \cref{lem:ccl:existlocal} generalises a similar statement over
existential sentences~\cite[Theorem 6]{grohe2004existential}
to existential local sentences.

\begin{restatable}[From Almost basic local to Asymmetric Basic local]{lemma}{asblvsabl}
    \label{lem:pos:asblvsabl}
    Every almost basic local sentence
    is equivalent to a disjunction of
    asymmetric basic local sentences.
\end{restatable}
\begin{sendappendix}
    \begin{ifappendix}
        \asblvsabl*
    \end{ifappendix}
\begin{proof}
    Let $\varphi \defined \exists \vec{x}.
    \bigwedge_{i \neq j} d(x_i, x_j) > 2r \wedge 
    \psi(\vec{x})$
    where $\psi$ is a $r$-local formula around $\vec{x}$.
    Following a syntactical proof of Feferman-Vaught
    given by Thomas Colcombet,
    we introduce \emph{types} $T_1, \dots, T_{|\vec{x}|}$
    and define
    typed formulas inductively as follows
    \begin{align*}
        \tau :=\, &
            x:T \mid
            R(x_1 : T, \dots, x_n : T) \quad \text{ when } R \in \upsigma \\
                &\tau \wedge \tau \mid
            \tau \vee \tau \mid
            \neg \tau \mid \\
                &\exists x:T. \tau \mid
            \forall x:T. \tau
    \end{align*}
    In this new language, a relation can only be checked using variables
    of the same type $T$.
    The evaluation of a formula $\tau$ is defined as the
    evaluation of the formula $\tau'$ obtained by removing type annotations.

    An \emph{environment} $\rho$ is a mapping from finitely many variables to
    types. Given an environment $\rho$, we write $\rho[x \mapsto T]$
    to denote the environment $\rho'$ obtained by $\rho'(y) \defined \rho(y)$
    if $x \neq y$ and $\rho'(x) = T$.
    We first translate $\psi$ into a typed formula inductively
    given an environment $\rho$ from the free variables of $\psi$
    to types:
    \begin{itemize}
        \item $f(\exists x. \psi, \rho) = \bigvee_{1 \leq i \leq |\vec{x}|}
            \exists x:T_i. f(\psi, \rho[x \mapsto T_i])$,
        \item $f(\psi_1 \vee \psi_2, \rho) = f(\psi_1, \rho) \vee f(\psi_2,
            \rho)$,
        \item $f(\neg \psi, \rho) = \neg f(\psi, \rho)$,
        \item $f(R(y_1, \dots, y_n), \rho) = R(y_1 : \rho(y_1), \dots, y_n :
            \rho(y_n))$ if all the variables have the same type under $\rho$,
            and
            $f(R(y_1, \dots, y_n), \rho) = \bot$ otherwise.
    \end{itemize}

    By a straightforward induction on the formula $\psi$,
    one can show that for every structure $A$,
    every valuation $\nu$ from variables to elements of $A$,
    and environment $\rho$ mapping free variables of $\psi$
    to types, such that $d_A(\nu(x),\nu(y)) \leq 1$
    implies $\rho(x) = \rho(y)$
    the following holds:
    $A, \nu \models \psi(\vec{x})$ if and only
    if $A, \nu \models f(\psi(\vec{x}), \rho)$.
    The only non trivial case is the translation of relations,
    which is handled of through the hypothesis on $\nu$ and $\rho$.

    As a consequence, for every structure $A$,
    $A \models \exists \vec{x}.
    \bigwedge_{i \neq j} d(x_i, x_j) > 2r \wedge 
    \psi(\vec{x})$
    if and only if 
    $A \models \exists \vec{x}.
    \bigwedge_{i \neq j} d(x_i, x_j) > 2r \wedge 
    f(\psi(\vec{x}), x_i \mapsto T_i)$.

    A second induction allows us to prove that
    typed formulas are equivalent to positive Boolean combinations
    of \emph{monotyped} formulas, i.e. formulas where only one type $T$
    appears.
    The only non-trivial case is the existential quantification,
    handled through the equivalence between
    $\exists x : T. (\psi_1 \wedge \psi_2)$
    and $(\exists x : T. \psi_1) \wedge \psi_2$
    whenever $\psi_2$ contains no variable of type $T$.

    Since the environment $\rho \colon x_i \mapsto T_i$
    assigns a different type to every free variable of $\psi(\vec{x})$,
    the positive Boolean combination obtained by transforming $f(\psi(\vec{x}), \rho)$
    is composed of formulas with exactly one free variable. Let us write
    $\bigvee_{n} \bigwedge_{1 \leq m \leq k} \tau_{n,m} (x_m : T_m)$
    for this positive Boolean combination.
    For every structure $A$,
    $A$ models $\exists \vec{x}.
    \bigwedge_{i \neq j} d(x_i, x_j) > 2r \wedge 
    \psi(\vec{x})$
    if and only if
    $A$ models $\exists \vec{x}.
    \bigwedge_{i \neq j} d(x_i, x_j) > 2r \wedge 
    \bigvee_{n} \bigwedge_{1 \leq m \leq k} \tau_{n,m}(x_m : T_m)
    $.
    By removing type annotations on the formulas $\tau_{n,m}$
    and considering their relativisation to the $r$-neighborhood
    of their single free variable, one obtains $\tau_{n,m}^r(x)$.
    Let us prove that this relativisation preserves the expected equivalence
    with $\varphi$:
    \begin{align*}
        A \models \varphi &\iff
        \exists \vec{a} \in A^k, \bigwedge_{1 \leq i \neq j \leq k} d_A(a_i,a_j)
        > 2r \wedge A, \vec{a} \models \psi(\vec{x}) 
                        \\&\text{$\psi$ is $r$-local} \\
                          &\iff
        \exists \vec{a} \in A^k, \bigwedge_{1 \leq i \neq j \leq k} d_A(a_i,a_j)
        > 2r \wedge \Neighb{A}{\vec{a}}{r}, \vec{a} \models \psi(\vec{x})
                        \\&\text{by definition}
        \\
                          &\iff
        \exists \vec{a} \in A^k, \bigwedge_{1 \leq i \neq j \leq k} d_A(a_i,a_j)
        > 2r \\& \quad\quad\quad \wedge \Neighb{A}{\vec{a}}{r}, 
                          \vec{a} \models 
        \bigvee_{n} \bigwedge_{1 \leq m \leq k} \tau_{n,m}(x_m : T_m)
                        \\&\text{disjunctions}
        \\
                          &\iff
        \exists n. \exists \vec{a} \in A^k, \bigwedge_{1 \leq i \neq j \leq k} d_A(a_i,a_j)
        > 2r \\&\quad\quad\quad \wedge \Neighb{A}{\vec{a}}{r}, \vec{a} \models 
        \bigwedge_{1 \leq m \leq k} \tau_{n,m}(x_m : T_m)
                        \\&\text{relativisation}
        \\
                          &\iff
        \exists n. \exists \vec{a} \in A^k, \bigwedge_{1 \leq i \neq j \leq k} d_A(a_i,a_j)
        > 2r \wedge A, \vec{a} \models 
        \bigwedge_{1 \leq m \leq k} \tau_{n,m}^r(x_m)
                       \\ &\text{$\exists$ and $\wedge$}
        \\
                          &\iff
        \exists n. 
         A \models 
        \exists \vec{x}. \bigwedge_{1 \leq i \neq j \leq k} d(x_i,x_j)
        > 2r \wedge
        \bigwedge_{1 \leq m \leq k} \tau_{n,m}^r(x_m)
        \\
    \end{align*}

    We have expressed $\varphi$ as a disjunction of asymmetric basic local
    sentences
    \begin{equation*}
        \varphi \equiv
        \bigvee_{n}
        \exists x_1, \dots, x_k.
        \bigwedge_{1 \leq i \neq j \leq k} d(x_i, x_j) > 2r \wedge 
        \bigwedge_{1 \leq m \leq k} \tau_{n,m}^{r}(x_m)
        \, . \qedhere
    \end{equation*}
\end{proof}
\end{sendappendix}

\subsection{From Asymmetric Basic Local to Basic Local Sentences}
\label{sec:pgaif:ablstobls}

We are now ready to build the final transformation
between asymmetric basic local sentences
and basic local sentences, 
reusing some of the combinatorics of~\citet[Lemma 4]{grohe2004existential}.

As a convenience, let us write
$\phi/i$ for the sentence $\varphi$ where the variable $x_i$
and local sentence $\psi_i$ are `removed.' For instance,
if $\phi \defined \exists x_1, x_2. d(x_1,x_2) > 2r \wedge \psi_1(x_1) \wedge
\psi_2(x_2)$ then 
$\phi / 1 = \exists x_2. \psi_2(x_2)$ and
$\phi/2 = \exists x_1. \psi_1 (x_1)$.

\begin{fact}[Removing variable weakens]
    \label{fact:loc:removevar}
    If $\varphi
    $ is an asymmetric basic local sentence
    of the form
    $\exists x_1, \ldots, x_k.
    \bigwedge_{i \neq j} d(x_i, x_j) > 2r \wedge 
    \bigwedge_{i} \psi_{i}(x_i)
    $,
    $A \models \varphi$
    and $1 \leq i \leq k$,
    then $A \models \varphi/i$.
\end{fact}

The following lemma allows us to reduce the number of
variables in an asymmetric basic local sentence
under the assumption that some witness is \emph{frequent}.
To simplify notations, let us write
$\exists^{\geq n}_r x. \theta(x)$
as a shorthand for
$\exists x_1, \ldots, x_n.
\bigwedge_{i \neq j} d(x_i, x_j) > 2r \wedge 
\bigwedge_{i} \theta(x_i)$. 
When $\theta(x)$ is a $r$-local
sentence, $\exists^{\geq n}_r x. \theta(x)$ is a basic local sentence.

\begin{lemma}[Repetitions]
    \label{fact:loc:repetitions}
    If $\varphi
    $ is an asymmetric basic local sentence
    of the form
    $\exists x_1, \ldots, x_k.
    \bigwedge_{i \neq j} d(x_i, x_j) > 2r \wedge 
    \bigwedge_{i} \psi_{i}(x_i)
    $
    and $A \models \varphi/i \wedge \exists^{\geq k}_{2r} x. \psi_i(x)$,
    then $A \models \varphi$.
\end{lemma}
\begin{proof}
    Let $A$ be a structure such that
    $A \models \varphi/i \wedge \exists^{\geq k}_{2r} x. \psi_i(x)$.
    By definition of $\varphi/i$ there exists 
    a vector $\vec{a} \in A^{k-1}$
    of points at pairwise distance greater than $2r$
    such that $A, a_j \models \psi_j(x)$
    for $1 \leq j \neq i \leq k$.
    To prove that $A \models \varphi$, it suffices
    to find some witness $a_i$ for $\psi_i$
    that is at distance greater than $2r$
    of $\vec{a}$.

    The fact that $A \models \exists^{\geq k}_{2r} x. \psi_i(x)$
    guarantees that we can find $k$ witnesses for $\psi_i$
    at pairwise distance greater than $4r$, let us write $\vec{b}$
    this set of witnesses.

    Assume by contradiction that $\forall b \in \vec{b},
    \exists a \in \vec{a}. d(a,b) \leq 2r$. Since
    $|\vec{a}| = k-1$ and $|\vec{b}| = k$, there exists a point $a \in \vec{a}$
    such that two elements $b_1$ and $b_2$ of $\vec{b}$ are
    at distance less than $2r$ of $a$. The triangular inequality
    implies that $d(b_1, b_2) \leq 4r$ which is absurd.

    We have proven that there exists a $b \in \vec{b}$
    such that $b$ is at distance greater than $2r$ of all elements in $\vec{a}$,
    therefore $A \models \varphi$.
\end{proof}

Let us temporarily fix a structure $A$.
To transform an asymmetric basic local sentence
into a positive Boolean
combination of basic local sentences,
one can proceed by induction on the number of
outer existential quantifications.
Whenever
\cref{fact:loc:removevar} can be applied,
it suffices to use the induction hypothesis.
As a consequence, the base case of our induction
is when \cref{fact:loc:removevar} cannot be applied
at all. For such an asymmetric basic local sentence
$\varphi \defined 
\exists x_1, \dots, x_k. \bigwedge_{1 \leq i \neq j \leq k} d(x_i,x_j) > 2r
\wedge \bigwedge_{1 \leq i \leq k} \psi_i(x_i)$
, the
set $W$ of elements in $A$
that satisfy at least one $\psi_i$
enjoys some sparsity property.
Namely, it is not possible to find more
than $k(k-1)$ points in $W$
whose neighborhoods of radius $2r$ do not intersect.

To handle this base case, we enumerate the possible behaviours of
the set $W$ through \emph{template graphs}.
Given a structure $A$, 
a non-empty finite set $Q$ of $r$-local properties,
a radius $R$ and 
a vector $\vec{a}$
of elements in $A$ such that
every $a \in \vec{a}$ satisfies at least one property $p \in Q$,
we build the template graph $G_\vec{a}^R$ as follows:
its vertices are the elements of $\vec{a}$ and it has
a labelled edge $(u,v,d_A(u,v))$ whenever $d_A(u,v) \leq R$.
Moreover, a node $a \in V(G)$ is coloured by $p \in Q$
whenever $A, a \models p(x)$.
Be careful that a vertex can have multiple
colours; the set of colours of a vertex $v$ is written $C(v)$.

Given a maximal size $K$, a radius $R$
and a non-empty finite set $Q$ of $r$-local properties,
we define
$\Delta(K,R,Q)$
to be the set of template graphs with at most $K$ vertices,
colours in $Q$ and edges labelled with integers at most~$R$.
Graphs in $\Delta(K,Q,R)$ are
ordered using $G \leq G'$ whenever there exists
an isomorphism $h \colon G \to G'$
between the underlying
graphs respecting edge labels
such that $C(v) \subseteq C(h(v))$ for $v \in V(G)$.
In a structure $A$, a graph $G$ is \emph{represented}
by a vector $\vec{a}$ whenever $G_{\vec{a}}^R \geq G$.

From a graph $G \in \Delta(K,R,Q)$
where $Q$ is a non-empty finite set of $r$-local properties,
one can build the $(R+r)$-local formula $\theta_G^R(x)$ that finds a
representative of~$G$ as follows
\begin{align*}
    \theta_G^R(x) \defined
    &\exists v_1, \ldots, v_{|V(G)|}
    \in \Neighb{}{x}{R}. \\
    &\bigwedge_{(v_i,v_j,h) \in E(G)}
    d(v_i,v_j) = h \\
    &\wedge \bigwedge_{ v_i \in V(G)}
    \bigwedge_{p \in C(i)} p(v_i)
    \\
    &\wedge \bigwedge_{v_i \in V(G)}
        \Neighb{}{v_i}{r} \subseteq \Neighb{}{x}{R} \, .
\end{align*}

\begin{restatable}[Graph representation]{fact}{graphrepresentation}
    \label{fact:loc:abstriaction}
    Let $A$ be a structure,
    $r, R \geq 1$,
    $Q$ be a non-empty finite set of $r$-local properties,
    $G \in \Delta(K,R,Q)$ be a template graph,
    and $a \in A$ be an element of~$A$.  Then
    $A, a \models \theta_G^R(x)$
    if and only if
    $\Neighb{A}{a}{R}$ contains
    points $\vec{b}$
    such that $G_\vec{b}^R \geq G$
    and $\Neighb{A}{\vec{b}}{r} \subseteq \Neighb{A}{a}{R}$.
\end{restatable}

Template graphs miss one key compositional
property, namely if $A$ is a structure
and $\vec{a}, \vec{b}$ are two vectors,
it is not immediate to recover $G_{\vec{ab}}^R$
from $G_{\vec{a}}^R$ and $G_{\vec{b}}^R$.
This is dealt with by restricting our
attention to vectors that are far enough
from one another.

\begin{fact}
    \label{fact:loc:templateunion}
    Let $A$ be a finite structure, $Q$ a non-empty
    finite set of $r$-local properties,
    and $\vec{a}, \vec{b}$ be vectors of elements of $A$
    whose neighborhoods of radius $R \geq r$
    do not intersect.
    Then the template graph $G_{\vec{a}\vec{b}}^R$
    is the disjoint union of $G_{\vec{a}}^R$
    and
    $G_{\vec{b}}^R$.
\end{fact}

To ensure that we can use \cref{fact:loc:templateunion},
we add \emph{security cylinders}
around points, as ensured by the following $\max(3R, R+r)$-local formula
\begin{ifappendix}%
whose behaviour is detailed in \cref{app:sec:pgaif:ablstobls}%
\end{ifappendix}
\begin{align*}
    \pi^R_Q(x) \defined
    &\forall y \in \Neighb{}{x}{3R}.
    &\left(\bigvee_{p \in Q} p(u)\right)
    \implies
    \Neighb{}{y}{r} \subseteq
    \Neighb{}{x}{R}
    \, .
\end{align*}

\begin{sendappendix}%
    \begin{ifappendix}%
        Let us recall the definitions in~\cref{sec:pgaif:ablstobls}
        before proving~\cref{fact:loc:securitycylinders}.
        The formula $\theta_G^R(x)$ is defined as follows
\begin{align*}
    \theta_G^R(x) \defined
    &\exists v_1, \ldots, v_{|V(G)|}
    \in \Neighb{}{x}{R}. \\
    &\bigwedge_{(v_i,v_j,h) \in E(G)}
    d(v_i,v_j) = h \\
    &\wedge \bigwedge_{ v_i \in V(G)}
    \bigwedge_{p \in C(i)} p(v_i)
    \\
    &\wedge \bigwedge_{v_i \in V(G)}
        \Neighb{}{v_i}{r} \subseteq \Neighb{}{x}{R} \, .
\end{align*}


We restate hereafter~\cref{fact:loc:abstriaction}.

\graphrepresentation*
    \end{ifappendix}%
\end{sendappendix}%
\begin{restatable}[Security Cylinders]{lemma}{securitycylinder}
    \label{fact:loc:securitycylinders}
    For every radii $R,r \geq 1$, 
    for every non-empty finite set $Q$ of $r$-local properties,
    for every graph $G \in \Delta(K,R,Q)$,
    for every structure $A$,
    for every points $a,b \in A$
    such that $d(a,b) \leq 2R$,
    $A, a \models \theta_G^R(x)$
    and $A, b \models \pi^R_Q(x)$
    implies $A, b \models \theta_G^R(x)$.
\end{restatable}
\begin{sendappendix}
    \begin{ifappendix}
        \securitycylinder*
    \end{ifappendix}
\begin{proof}
    Since $A,a \models \theta_G^R(x)$
    and using \cref{fact:loc:abstriaction}
    there exists a vector $\vec{c} \in \Neighb{A}{a}{R}$
    such that $G_{\vec{c}}^R = G$ 
    and
    $\Neighb{A}{\vec{c}}{r} \subseteq \Neighb{A}{a}{R}$.
    In particular, every point of $\vec{c}$ satisfies
    at least one property $p \in Q$.
    As $d(a,b) \leq 2R$,
    $\vec{c} \in \Neighb{A}{b}{3R}$.
    As $A, b \models \pi_Q^R(x)$
    his shows that $\Neighb{A}{\vec{c}}{r} \subseteq \Neighb{A}{b}{R}$.
    Using~\cref{fact:loc:abstriaction} this proves 
    that $A, b \models \theta_G^R(x)$.
\end{proof}
\end{sendappendix}

\begin{lemma}[From Asymmetric basic local to Basic Local]
    \label{lem:pos:blvsabl}
Every asymmetric basic local sentence
is equivalent to a positive Boolean combination of
basic local sentences.
\end{lemma}
\begin{proof}
    Let $\varphi$ be 
    of the form
    $\exists x_1, \ldots, x_k.
    \bigwedge_{i \neq j} d(x_i, x_j) > 2r \wedge
    \bigwedge_{i} \psi_{i}(x_i)
    $.
    We prove by induction over $k$ that $\varphi$
    is equivalent
    to a positive Boolean combination
    of basic local sentences.
    When $k \leq 1$, $\varphi$ is already a basic local sentence
    hence we assume $k \geq 2$. 

    For $1 \leq i \leq k$, we apply the induction hypothesis on
    $\varphi/i$, which has strictly fewer existentially quantified
    variables, and call $\overline{\varphi/i}$ the obtained
    positive Boolean combination of basic local sentences.

    Let $Q \defined \{ \psi_1, \ldots, \psi_k\}$
    and
    for convenience, let $\Delta$ be
    a shorthand for
    $\Delta(k, 4^{k^2} 6r, Q)$.
    We define $\mathbb{M}$
    to be the multisets of $\wp(\Delta)$
    with up to $k(k-1)$ repetitions per element
    of $\wp(\Delta)$.
    Given a multiset $M \in \mathbb{M}$
    and a radius $R$
    we can write
    the following
    conjunction of
    basic local sentences
    \begin{equation*}
        \Theta_M^R
        \defined 
        \bigwedge_{S \in \wp(\Delta)}
        \exists_{3R}^{\geq M(S)} x.
        \pi^R_Q(x)
        \wedge
        \bigwedge_{G \in S} \theta_G^R(x)
        \wedge
        \bigwedge_{G \not\in S} \neg \theta_G^R(x)
        \, .
    \end{equation*}

    \newcommand{\obt}{\operatorname{Obt}}
    The set of graphs \emph{obtainable}
    from a multiset $M \in \mathbb{M}$
    is written $\obt(M)$
    and defined
    inductively by
    $\obt(\emptyset) \defined \emptyset$
    and $\obt( \{ S \} + M)$
    defined as the union
    of $\obt(M)$, $S$, and
    $\setof{ G \uplus G'}{G \in S \wedge G' \in \obt(M)}
    $.
    We call a multiset $M$ \emph{valid}
    whenever there exists a graph $G$
    obtainable from $M$ that contains vertices
    $v_1, \dots, v_k$ at pairwise weighted distance
    greater than $2r$ and such that $v_i$ is coloured by
    $\psi_i$. The finite set of valid multisets is written $\mathbb{M}_V$.

    Let us prove that
    $\varphi$ is equivalent to the positive Boolean combination
    of basic local sentences $\Psi$ defined as
    \begin{equation*}
        \Psi
        \defined
        \left(\bigvee_{1 \leq i \leq k} \overline{\varphi/i} \wedge \exists^{\geq
        k}_{2r} x. \psi_i(x)\right)
        \vee
        \left(\bigvee_{M \in \mathbb{M}_V} 
        \bigvee_{6r \leq R \leq 4^{k^2} 6r}
        \Theta_M^R\right)
        \, .
    \end{equation*}

    \bigskip\smallskip
    Assume first that $A \models \varphi$.
            Using \cref{fact:loc:removevar},
            $A \models \varphi/i$ for all $1 \leq i \leq
            k$
            and by construction this proves
            that 
            $A \models \overline{\varphi/i}$ for all $1 \leq i \leq
            k$. As a consequence, we only need
            to treat the case where
            $A \not\models \exists^{\geq
            k}_{2r} x. \psi_i(x)$
            for all $1 \leq i \leq k$.

            In such a structure $A$, let us
            call $W$ the set
            of elements in $A$
            that satisfy at least one $\psi_i$.
            It is not possible to find more
            than $k(k-1)$ points in $W$
            whose neighborhoods of radius $2r$ do not intersect.
            This implies that there exists a vector $\vec{c}$
            in $A$ of size less or equal $k(k-1)$
            such that $W \subseteq \Neighb{A}{\vec{c}}{6r}$.

            Using \cref{lem:ccl:extentedcover}
            over $\vec{c}$ and $6r$ one obtains
            a radius $6r \leq R \leq 4^{k^2} 6 r $
            and a vector $\vec{b}$ such that
            $
            \Neighb{A}{\vec{c}}{6r}
            \subseteq
            \Neighb{A}{\vec{b}}{R}$
            and the neighborhoods of radius $3R$
            around points in $\vec{b}$ do not intersect.

            Given an element $b \in \vec{b}$,
            construct $S_b$
            the collection of 
            the template graphs 
            $G_{\vec{a}}^R$
            when $\vec{a}$
            ranges over the sets of $k$ points
            of $W \cap \Neighb{A}{b}{R}$.
            By construction,
            $S_b \in \wp(\Delta)$.
            Let us write $M_\vec{b}$ the multiset obtained
            by collecting the sets $S_b$ for $b \in \vec{b}$.

            We now prove that $M_\vec{b}$ is valid.
            As $A \models \varphi$, there exists a vector $\vec{a} \in A^k$
            of points at pairwise distance greater than $2r$
            such that $A, a_i \models \psi_i(x)$
            hence it suffices to prove that
            $G_{\vec{a}}^R \in \obt(M_{\vec{b}})$
            to conclude that $M_\vec{b}$ is valid.

            Remark that $\vec{a}$ is included in the disjoint union
            of the $3R$ neighborhoods of the elements of $\vec{b}$.
            As a consequence of \cref{fact:loc:templateunion}
            $G_\vec{a}^R = \biguplus_{b \in \vec{b}} G_{\vec{a}/b}^R$
            where $\vec{a}/b$ is the vector of elements of $\vec{a}$
            that are at distance less than $3R$ of $b$.
            This proves that $G_{\vec{a}}^R \in \obt(M_\vec{b})$.

            As a consequence, $A \models \Theta_{M_\vec{b}}^R$
            with $M$ valid and $6r \leq R \leq 4^{k^2} 6r$
            and therefore $A \models \Psi$.

      \bigskip
      Assume conversely that $A \models \Psi$.
            If $A \models \overline{\varphi/i} \wedge \exists^{\geq
            k}_{2r} x. \psi_i(x)$ for some~$i$,
            then by definition
            $A \models \varphi/i \wedge \exists^{\geq
            k}_{2r} x. \psi_i(x)$
            and using \cref{fact:loc:repetitions},
            $A \models \varphi$.

            Otherwise $A \models \Theta_M^R$ where $M$ is a valid multiset.
            Let us write $M = \{ S_1, \dots, S_n \}$
            and each $S_i$ is repeated $m_i \leq k(k-1)$ times.
            By construction of $\Theta_M^R$, there exist
            points
            $b_i^j$ for $1 \leq i \leq n$
            and $1 \leq j \leq m_i$
            in $A$ such that
            $A, b_i^j \models
            \pi_Q^R(x)
            \wedge
            \bigwedge_{G \in S_i} \theta_G^R(x)
            \wedge
            \bigwedge_{G \not \in S_i} \neg\theta_G^R(x)
            $. Moreover, for each~$i$
            the points $b_i^j$ are at distance
            greater than $6R$.

            Our goal is to prove that the points $b_i^j$ are
            at pairwise distance greater than $2R$ when $i$ varies.
            Assume by contradiction that
            there exists $i,j$ and $i',j'$ such that
            $i \neq i'$ and $d\left(b_i^j, b_{i'}^{j'}\right) \leq 2R$.
            We are going to prove that $S_i = S_{i'}$.
            Assume by contradiction that there exists
            $G \in S_i \setminus S_{i'}$.
            We know that
            $A, b_{i}^{j} \models \theta_G^R(x)$
            and
            $A, b_{i'}^{j'} \models \neg \theta_G^R(x)$.
            As 
            $A, b_{i}^{j} \models \pi^R_Q(x)$,
            \cref{fact:loc:securitycylinders} implies that
            $A, b_{i'}^{j'} \models \theta_G^R(x)$
            which is absurd. Hence, $S_i = S_{i'}$,
            which is in contradiction with the definition of $M$,
            and finally, the points $b_i^j$ must be at distance
            greater than $2R$.

            As the multiset $M$ is valid, there exists a graph
            $G \in \obt(M)$ with vertices 
            $v_1, \dots, v_k$ at pairwise distance
            greater than $2r$
            such that $v_i$ is coloured by $\psi_i$.
            By induction on the construction of $G$,
            and since the $R$-neighborhoods
            around the points $b_i^j$ do not intersect,
            we can use \cref{fact:loc:templateunion}
            to build a vector $\vec{c}$ in $A$
            such that $G_\vec{c}^R \geq G$.
            As a consequence, there exists
            $k$ points $a_1, \dots, a_k$ in $\vec{c}$ that are at pairwise distance 
            greater than $2r$ and such that $A, a_i \models \psi_i(x)$.
            We have proven that $A \models \varphi$.
            \qedhere
\end{proof}

\positivelocality*
\begin{proof}
    Combine \cref{lem:ccl:existlocal}
    with \cref{lem:pos:asblvsabl}
    and
    \cref{lem:pos:blvsabl}.
\end{proof}

Notice that we do not recover the existential
locality
of~\citeauthor{grohe2004existential} using a radius $1$
since we introduce negations in the construction of the basic
local sentences.
Conversely, we did not manage to apply directly their
result to obtain \cref{thm:ccl:mongaifman}.
\begin{ifnotappendix}
    \Cref{thm:ccl:mongaifman} can be understood as a ``decoupling''
    theorem, stating that detecting a neighbourhood around
    several points
    in a structure is the same as
    detecting repetitions of neighbourhoods around single points.
\end{ifnotappendix}

\section{A Local Preservation Theorem}
\label{sec:localpres}\paragraph*{Positive Locality Theorem}
\label{sec:gaifmaninf}

In order to prove \cref{lem:ccl:existlocalnfinfinite},
we use
a classical construction of model theory.
Given a structure $A \in \Mod(\upsigma)$,
one can build a new langage $\mathcal{L}_A$
consisting of the relations in $\upsigma$
extended with constants $c_a$ for $a \in A$;
the structure $A$ is then canonically interpreted
as a structure in $\mathcal{L}_A$ by interpreting
$c_a$ with the element $a$; this structure is
written $\hat{A}$.
Be careful that this construction makes us
temporarily
leave the realm of \emph{finite relational signatures}.
It is possible, given a fragment $\mathsf{F}$ of $\FO$
over the extended language
to build a theory $T^+(\mathsf{F},A)$
consisting of all the sentences of $\mathsf{F}$
true in $\hat{A}$.

\existlocalinfinite*
\begin{proof}
	The implications $(a) \Rightarrow (b) \Rightarrow (c) \Rightarrow (d)$
	are consequences respectively of the definition of $\locleq{r}{q}{k}$,
	\cref{fact:qo:refinement} and \cref{lem:qo:elemloc}.
	As a consequence, let us focus on the implication $(d) \Rightarrow (a)$.



	Assume that $A$ is a structure such that
	$A \models \varphi$.
	Let us build $\mathcal{L}_A$
	the extended language of $A$
	and
	$\mathsf{F}$
	the set of sentences
	that are $r$-local
	around the constants $c_a$.

	Whenever $\hat{B}$
	is a structure over the extended language
	such that $\hat{B} \models T^+(\mathsf{F},A)$
	we have
	an interpretation $\nu$ of the constants
	$c_a$ in $B$. Notice that the function
	$h \colon A \to B$
	defined by $h(a) \defined \nu(c_a)$
	is a local elementary embedding
	from $A$ to $B$.

	\vspace{0.5em}

	Because $A \models \varphi$,
	the theory $T^+(\mathsf{F},A) \cup \{ \neg \varphi \}$
	is inconsistent in the language $\mathcal{L}_A$.
	Indeed,
	assume there exists a structure $\hat{B}$
	such that $\hat{B} \models \neg \varphi$
	and $\hat{B} \models T^+(\mathsf{F},A)$.
	Then there exists an elementary embedding $h \colon A \to B$
	and $A \models \varphi$ so $B \models \varphi$.
	As $\varphi$ does not contain any constant of the form
	$c_a$ this proves that $\hat{B} \models \varphi$
	which is absurd.

	This entails by the compactness theorem
	the existence of
	a sentence $\psi_A$, a finite conjunction
	of sentences in $T^+(\mathsf{F},A)$
	such that $\psi_A \implies \varphi$.
	Moreover, since finitely many constants appear
	in the sentence $\psi_A$, we can build
	a sentence $\theta_A$ by quantifying existentially
	over them. In particular, the sentence $\theta_A$
	is a first-order sentence over the finite relational
	signature $\upsigma$.
	Recall that $\hat{A} \models T^+(\mathsf{F},A)$ by definition,
	hence $\hat{A} \models \psi_A$ and in particular
	$A \models \theta_A$.

	\vspace{0.5em}

	Let us build $T_\varphi \defined
		\setof{\neg \theta_C}{C \models \varphi}$.
	We prove that
	$T_\varphi \cup \{ \varphi \}$ is inconsistent
	as a model $B$ of this theory
	is such that $\neg \theta_B \in T_\varphi$
	hence $B \models \neg \theta_B$, but this contradicts
	the fact that $B \models \theta_B$.

	A  second use of the compactness theorem allows us to extract
	$C_1, \dots, C_n$ such that
	$\varphi \implies \theta_{C_1} \vee \dots \vee \theta_{C_n}$.
	Conversely, since $C_i \models \varphi$,
	whenever $B \models \theta_{C_i}$
	the structure $B$ models $\varphi$.
	As a consequence,
	$\varphi \iff \theta_{C_1} \vee \dots \vee \theta_{C_n}$.
	Notice that this is a finite disjunction of finite
	conjunctions of sentences that are of the form
	$\exists \vec{x}. \tau(\vec{x})$ with $\tau$ local around
	$\vec{x}$,
	hence of the appropriate form.
	\qedhere

\end{proof}

\paragraph*{Adaptation to the Finite Case}
\label{sec:gaifmanccl}
As hinted in the introduction, the proof of \cref{lem:ccl:existlocalnfinfinite}
does not relativise to finite classes of structures
due to its intensive use of the compactness theorem
of first order logic.
We know from \cref{lem:qo:locqofin} that
$(c)$ and $(d)$ remain equivalent in the finite, and will prove
in \cref{sec:ff:generic} that the equivalence between $(a)$ and $(c)$ fails in
the finite.
The following lemma salvages the equivalence between
$(a)$ and $(b)$ in the finite.

\begin{restatable}[Local preservation in the finite]{lemma}{existlocalnf}
	\label{lem:ccl:existlocalnf}
	Let $\XClass \subseteq \Mod(\upsigma)$
	be a class of structures.
	For a sentence $\varphi \in \FO[\upsigma]$
	the following properties are equivalent
	\begin{enumerate}[(a)]
		\item The sentence $\varphi$ is equivalent (over $\XClass$)
		      to an existential local
		      sentence.
		\item There exist $r,q,k \in \mathbb{N}$ such
		      that $\varphi$ is preserved under $\locleq{r}{q}{k}$
		      over $\XClass$.
	\end{enumerate}
\end{restatable}
\begin{proof}
	A sentence of the form $\exists \vec{x}. \tau$
	with $\tau$ a $r$-local sentence of quantifier rank at most $q$ is
	always preserved under $\locleq{r}{q}{|x|}$, which allows us to conclude
	$(a) \implies (b)$.

	For the converse direction,
	assume that $\varphi$ is preserved under $\locleq{r}{q}{k}$
	for some $r,q,k \in \mathbb{N}$.
	Compute $U$ as the image of $\modset{\varphi}_\XClass$
	through the map $\Specter{r}{q}{k}$.
	For a set $T$ of local types in $U$,
	one can build the sentence
	$\psi_T \defined
		\bigwedge_{\tau(\vec{x}) \in T} \exists \vec{x}. \tau(\vec{x})$.
	We define
	$\psi \defined
		\bigvee_{T \in U} \psi_T$
	and prove the following, where the second equality stems from our hypothesis~$(a)$
	\begin{equation}
		\modset{\psi}_\XClass
		=
		\setof{B \in \XClass}{
			\exists A \in \modset{\varphi}_\XClass,
			A \locleq{r}{q}{k} B
		}
		=
		\modset{\varphi}_\XClass\;.
	\end{equation}

	Whenever $B \models \psi$,
	there exists $T \in U$ such that
	$B \models \psi_T$. Hence, there exists $A \in X$ such that
	$A \models \varphi$,
	and
	$T = \Specter{r}{q}{k}(A)$.
	Thanks to \cref{fact:qo:specters},
	$A \locleq{r}{q}{k} B$.

	Conversely, let $A \in \XClass$
	such that $A \models \varphi$ and $B \in \XClass$ such that
	$A \locleq{r}{q}{k} B$.
	Let us write $T \defined \Specter{r}{q}{k}(A)$
	and $T' \defined \Specter{r}{q}{k}(B)$.
	It is always the case that $B$ satisfies its local types
	so
	$B \models \psi_{T'}$
	and since $T \subseteq T'$
	we can conclude that
	$B \models \psi_T$
	which entails that $B \models \psi$.

	We have proven that $\varphi$ is equivalent to $\psi$
	which is a positive Boolean combination of existential local sentences,
	thus, equivalent to an existential local sentence.
\end{proof}

\begin{ifnotappendix}
	From \cref{lem:ccl:existlocalnf}, one
	deduces that a sentence $\varphi$ presreved under
	$\locleq{\infty}{\infty}{\infty}$
	can be rewritten as an existential
	local sentence if and only if
	it is actually preserved
	under
	$\locleq{r}{q}{k}$
	for some finite parameters.
\end{ifnotappendix}

\section{Failure in the finite case}
\label{sec:ffc}
The goal of this section is
to prove that \cref{lem:ccl:existlocalnfinfinite} 
does not relativise to finite structures,
in particular we show in \cref{sec:ff:generic} that $(a) \not \Leftrightarrow (c)$ in the finite.

This
naturally leads to a decision problem: given a sentence $\varphi$
preserved under disjoint unions over $\ModF(\upsigma)$,
can it be rewritten as an existential local sentence?
We prove in \cref{sec:ff:undecidable} that this problem, and two other
associated decision
problems, are undecidable.

We complete the picture in \cref{sec:ff:generalise}
by describing the
parameters $(r,q,k)$ such that
first-order sentences preserved under $\locleq{r}{q}{k}$
over $\ModF(\upsigma)$ are equivalent to an existential local sentence.

\subsection{A Generic Counter Example}
\label{sec:ff:generic}

In the following, we will add whenever necessary a unary predicate $B$
to the signatures in order to construct the following property:
$\badphi \defined \forall x. \neg B(x) \vee \varphi_{CC}$,
where
$\varphi_{CC}$ checks that the Gaifman
graph of the structure
has at least two connected components.
Notice that $\badphi$ may not be
definable in $\FO[\upsigma]$ since $\varphi_{CC}$
may not be definable.
Over a class $\XClass$ of finite structures
where $\varphi_{CC}$ is definable as an existential local sentence,
the sentence $\badphi$ is well-defined
and preserved under disjoint unions.
Whenever the class $\XClass$ contains large enough structures,
$\badphi$ cannot be expressed
using an existential local sentence. Indeed,
such an existential local sentence will not distinguish between
a large connected component with one $B$ node
(not satisfying $\badphi$)
and two connected components with one $B$ node
(satisfying $\badphi$). A prototypical example
is the class of finite paths.

\begin{example}[Finite paths]
    \label{ex:ff:finitepaths}
    If $\sigma = \{ (E, 2) \}$
    and ${\uplus}\mathsf{P}$ is the class of disjoint unions of finite paths,
    $\badphi$ is definable, closed under disjoint unions,
    and is not expressible as
    an existential local sentence.
\end{example}
\begin{proof}
    One can detect the presence of
    two connected components
    using the fact that paths have at most
    two vertices of degree below $2$ through
    $\varphi_{CC}$ defined as the disjunction
    of $\exists x_1, x_2, x_3,
    x_4.
    \bigwedge_{1 \leq i \neq j \leq 4} x_i \neq x_j
    \wedge
    \bigwedge_{i =
    1}^{4} \deg (x_i) = 1)$,
    $\exists x_1, x_2. x_1 \neq x_2 \wedge \deg (x_1) = \deg(x_2) = 0)$,
    and
    $\exists x_1, x_2, x_3.
    x_1 \neq x_2 \wedge x_2 \neq x_3 \wedge x_3 \neq x_1 \wedge
    \deg(x_1) = \deg (x_2) = 1
    \wedge
    \deg (x_3) = 0
    $.

    Assume by contradiction that there exists
    a sentence $\psi = \exists \vec{x}. \theta(x)$
    with $\theta(x)$ a $r$-local sentence
    of quantifier rank $q$
    such that $\badphi$ is equivalent to $\psi$
    over ${\uplus}\mathsf{P}$.
    Consider the family of structures
    $P^{\neg B}_k$ that are paths of length $k$
    labelled by $\neg B$ so that 
    $P^{\neg B}_k \models \badphi$.
    Consider $k > |\vec{x}| \cdot (2 r + 1)$.
    Since $P^{\neg B}_k \models \psi$,
    there exists a vector $\vec a$ with $|\vec{x}|$
    elements of $P^{\neg B}_k$ such that
    $\Neighb{P^{\neg B}_k}{\vec{a}}{r} \models \theta(x)$
    and $\Neighb{P^{\neg B}_k}{\vec{a}}{r} \subsetneq P^{\neg B}_k$.
    Consider a point $b \in P^{\neg B}_k$ that is not in
    $\Neighb{P^{\neg B}_k}{\vec{a}}{r}$, and
    build ${P^{B}_k}$ by colouring this point with $B$.
    The structure $P^{B}_k$
    is still a path
    but does not satisfy $\badphi$.
    However, $P^{B}_k \models \psi$
    by construction, and this is absurd.
\end{proof}

Assume that $\BadOrderS$ is a class of finite structures
defined over $\ModF(\upsigma)$
through a finite axiomatisation $\BadAxiom$
using \emph{universal local sentences} (negations
of existential local sentences). The following fact
allows lifting arguments over $\BadOrderS$
to $\ModF(\upsigma)$.

\begin{fact}[Relativisation to $\BadOrderS$]
    \label{fact:ff:relat}
    A sentence $\varphi$
    is equivalent over $\BadOrderS$ to an existential local sentence
    if and only if $\BadAxiom \implies \varphi$
    is equivalent to an existential local sentence
    over $\ModF(\upsigma)$.
    Similarly,
    a sentence $\varphi$ is preserved under disjoint union
    over $\BadOrderS$ if and only if 
    $\BadAxiom \implies \varphi$
    is preserved under disjoint union over $\ModF(\upsigma)$.
\end{fact}

Note that the class ${\uplus}\mathsf{P}$ of finite paths
has no finite axiomatisation in $\ModF(\upsigma)$. Thus,
\cref{fact:ff:relat} cannot be used to lift \cref{ex:ff:finitepaths}
to the class of all finite structures.
As a workaround, we refine the counter example provided by \citet{tait1959counterexample}
in the case of the {\ltrsk} and leverage the idea given by the class
of finite paths.
For this purpose, we let $\sigma \defined \{ (\leq, 2), (S, 2), (E,2) \}$.
Define $O_n$ to be the structure $\{1, \dots, n\}$ with
$S$ interpreted as the successor relation, $\leq$ as the usual
ordering of natural numbers and $E$ the empty relation.
Given natural numbers
$2 \leq m \leq n$
one can build
a structure denoted by $O_{m} + \cdots + O_{n}$
by extending the disjoint union
$\biguplus_{m \leq i \leq n} O_{i}$
with new relations
$S(a,b)$ whenever $a$ is the last element of $O_{i}$
and $b$ the first of $O_{i+1}$,
and $E(a,b)$ whenever $a \in O_{i}$,
$b \in O_{i+1}$ and $b$ is below $a$ when interpreted
as integers.
An example of such a structure is given in \cref{fig:ff:exstruct}
with $m = 2$ and $n = 5$.
We define then $\BadOrderS$ to be the class of
finite disjoint unions of structures
of the form
$O_{m} + \cdots + O_{n}$
with $2 \leq m \leq n$.

\begin{figure}
    \begin{center}
\begin{tikzpicture}[x={(0,1cm)},y={(-1cm,0)}]
\draw[dashed, Prune, ->] (2,0) -- node[midway, above] {$\leq$} (2,-1);
\draw[myBlue, ->] (4,0) -- node[midway, above] {$E$} (4,-1);
\draw[ultra thick, ->] (3,0) -- node[midway, above] {$S$} (3,-1);
\node[draw,circle] (X0Y1) at (1, 0) {};
\node[draw,circle] (X1Y2) at (2, -2) {};
\node[draw,circle] (X3Y2) at (2, -6) {};
\node[draw,circle] (X0Y0) at (0, 0) {};
\node[draw,circle] (X3Y3) at (3, -6) {};
\node[draw,circle] (X3Y0) at (0, -6) {};
\node[draw,circle] (X3Y1) at (1, -6) {};
\node[draw,circle] (X2Y1) at (1, -4) {};
\node[draw,circle] (X2Y0) at (0, -4) {};
\node[draw,circle] (X2Y3) at (3, -4) {};
\node[draw,circle] (X2Y2) at (2, -4) {};
\node[draw,circle] (X1Y0) at (0, -2) {};
\node[draw,circle] (X3Y4) at (4, -6) {};
\node[draw,circle] (X1Y1) at (1, -2) {};
\draw (X0Y1) edge[->,dashed,draw=Prune,looseness=5] (X0Y1);
\draw (X1Y2) edge[->,dashed,draw=Prune,looseness=5] (X1Y2);
\draw (X3Y0) edge[->,dashed,draw=Prune,bend right=90] (X3Y2);
\draw (X3Y1) edge[->,dashed,draw=Prune,looseness=5] (X3Y1);
\draw (X2Y3) edge[->,dashed,draw=Prune,looseness=5] (X2Y3);
\draw (X3Y1) edge[->,dashed,draw=Prune,bend right=90] (X3Y4);
\draw (X2Y1) edge[->,dashed,draw=Prune,bend right=90] (X2Y3);
\draw (X2Y1) edge[->,dashed,draw=Prune,bend right=90] (X2Y2);
\draw (X2Y0) edge[->,dashed,draw=Prune,bend right=90] (X2Y1);
\draw (X0Y0) edge[->,dashed,draw=Prune,looseness=5] (X0Y0);
\draw (X1Y1) edge[->,dashed,draw=Prune,looseness=5] (X1Y1);
\draw (X3Y0) edge[->,dashed,draw=Prune,bend right=90] (X3Y1);
\draw (X3Y1) edge[->,dashed,draw=Prune,bend right=90] (X3Y2);
\draw (X3Y4) edge[->,dashed,draw=Prune,looseness=5] (X3Y4);
\draw (X2Y0) edge[->,dashed,draw=Prune,bend right=90] (X2Y3);
\draw (X2Y2) edge[->,dashed,draw=Prune,bend right=90] (X2Y3);
\draw (X2Y1) edge[->,dashed,draw=Prune,looseness=5] (X2Y1);
\draw (X2Y0) edge[->,dashed,draw=Prune,bend right=90] (X2Y2);
\draw (X3Y2) edge[->,dashed,draw=Prune,looseness=5] (X3Y2);
\draw (X3Y2) edge[->,dashed,draw=Prune,bend right=90] (X3Y3);
\draw (X3Y0) edge[->,dashed,draw=Prune,bend right=90] (X3Y3);
\draw (X2Y0) edge[->,dashed,draw=Prune,looseness=5] (X2Y0);
\draw (X3Y2) edge[->,dashed,draw=Prune,bend right=90] (X3Y4);
\draw (X1Y0) edge[->,dashed,draw=Prune,looseness=5] (X1Y0);
\draw (X1Y0) edge[->,dashed,draw=Prune,bend right=90] (X1Y1);
\draw (X1Y1) edge[->,dashed,draw=Prune,bend right=90] (X1Y2);
\draw (X3Y3) edge[->,dashed,draw=Prune,bend right=90] (X3Y4);
\draw (X3Y3) edge[->,dashed,draw=Prune,looseness=5] (X3Y3);
\draw (X3Y1) edge[->,dashed,draw=Prune,bend right=90] (X3Y3);
\draw (X0Y0) edge[->,dashed,draw=Prune,bend right=90] (X0Y1);
\draw (X1Y0) edge[->,dashed,draw=Prune,bend right=90] (X1Y2);
\draw (X2Y2) edge[->,dashed,draw=Prune,looseness=5] (X2Y2);
\draw (X3Y0) edge[->,dashed,draw=Prune,looseness=5] (X3Y0);
\draw (X3Y0) edge[->,dashed,draw=Prune,bend right=90] (X3Y4);
\draw (X0Y1) edge[->,ultra thick,draw=black,out=-45,in=120] (X1Y0);
\draw (X0Y0) edge[->,ultra thick,draw=black] (X0Y1);
\draw (X3Y3) edge[->,ultra thick,draw=black] (X3Y4);
\draw (X1Y0) edge[->,ultra thick,draw=black] (X1Y1);
\draw (X3Y0) edge[->,ultra thick,draw=black] (X3Y1);
\draw (X3Y1) edge[->,ultra thick,draw=black] (X3Y2);
\draw (X1Y1) edge[->,ultra thick,draw=black] (X1Y2);
\draw (X2Y2) edge[->,ultra thick,draw=black] (X2Y3);
\draw (X1Y2) edge[->,ultra thick,draw=black,out=-45,in=120] (X2Y0);
\draw (X2Y0) edge[->,ultra thick,draw=black] (X2Y1);
\draw (X2Y3) edge[->,ultra thick,draw=black,out=-45,in=120] (X3Y0);
\draw (X3Y2) edge[->,ultra thick,draw=black] (X3Y3);
\draw (X2Y1) edge[->,ultra thick,draw=black] (X2Y2);
\draw (X2Y1) edge[->,draw=myBlue] (X3Y1);
\draw (X1Y1) edge[->,draw=myBlue] (X2Y0);
\draw (X2Y0) edge[->,draw=myBlue] (X3Y0);
\draw (X0Y1) edge[->,draw=myBlue] (X1Y0);
\draw (X1Y2) edge[->,draw=myBlue] (X2Y1);
\draw (X2Y2) edge[->,draw=myBlue] (X3Y1);
\draw (X2Y3) edge[->,draw=myBlue] (X3Y2);
\draw (X2Y2) edge[->,draw=myBlue] (X3Y0);
\draw (X2Y3) edge[->,draw=myBlue] (X3Y3);
\draw (X2Y2) edge[->,draw=myBlue] (X3Y2);
\draw (X2Y3) edge[->,draw=myBlue] (X3Y1);
\draw (X1Y2) edge[->,draw=myBlue] (X2Y0);
\draw (X1Y0) edge[->,draw=myBlue] (X2Y0);
\draw (X2Y3) edge[->,draw=myBlue] (X3Y0);
\draw (X1Y1) edge[->,draw=myBlue] (X2Y1);
\draw (X0Y0) edge[->,draw=myBlue] (X1Y0);
\draw (X0Y1) edge[->,draw=myBlue] (X1Y1);
\draw (X2Y1) edge[->,draw=myBlue] (X3Y0);
\draw (X1Y2) edge[->,draw=myBlue] (X2Y2);

		\begin{scope}[on background layer]
            \foreach \i in {0,...,3} {
                \pgfmathsetmacro{\ip}{int(\i + 1)}
                \def\A{X\i Y0}
                \def\B{X\i Y\ip}
			\draw[line width=0.7cm, Prune!5,
				fill=Prune!5,
				rounded corners=1mm,
				line cap=round]
                (\A .south) -- (\B .north);
            }
		\end{scope}

\end{tikzpicture}
    \end{center}
    \caption{The structure $O_2 + O_3 + O_4 + O_5$}
    \label{fig:ff:exstruct}
\end{figure}
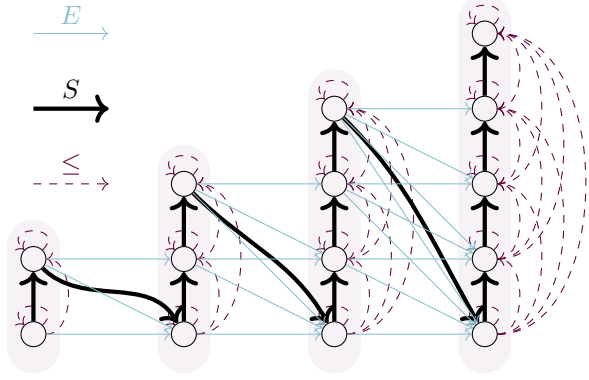

\begin{lemma}[Definition of $\varphi_{CC}$ for $\BadOrderS$]
    \label{lem:ff:phiccbads}
    There exists a sentence $\varphi_{CC}$
    that is existential local such that
    for all $A \in \BadOrderS$,
    $A \models \varphi_{CC}$ if and only
    if $A$ contains at least two connected components.
\end{lemma}
\begin{proof}
    Define $\varphi_{CC} \defined
    \exists x_1,x_2.x_1\neq x_2\wedge
    \forall y. \neg S(y,x_1) \wedge \neg S(y,x_2)
    $.
\end{proof}

\begin{ifappendix}
The translation of \cref{ex:ff:finitepaths} to
$\BadOrderS$ is a simple exercise detailed in \cref{app:sec:ff:generic}.
\end{ifappendix}

\begin{restatable}[Counter-example over $\BadOrderS$]{lemma}{badphicord}
    \label{lem:ff:counterbads}
    The sentence $\badphi \defined \forall x. B(x) \vee \varphi_{CC}$
    is preserved under disjoint unions over
    $\BadOrderS$
    but
    is not equivalent to any
    existential local sentence over $\BadOrderS$.
\end{restatable}
\begin{sendappendix}
    \begin{ifappendix}
        \badphicord*
    \end{ifappendix}
\begin{proof}
    The sentence $\badphi \defined \forall x. B(x) \vee \varphi_{CC}$
    is well defined and preserved under disjoint unions
    thanks to \cref{lem:ff:phiccbads}.

    Assume by contradiction that $\badphi$
    is equivalent over $\BadOrderS$ to a
    sentence $\psi = \exists x_1, \dots, x_k. \theta(\vec{x})$
    where $\theta$ is a $r$-local sentence.
    Construct the structure
    $A \defined O_1 + O_2 + \cdots + O_{2 k \cdot (2r+1) + 1}$
    in $\BadOrderS$
    and colour all the nodes with $\neg B$.
    This structure satisfies $\badphi$
    hence $\psi$. However,
    there exists at least one node of $A$
    not covered by the $r$-balls around the $k$
    witnesses of $\psi$, the structure $\hat{A}$
    obtained by colouring one of those node
    with $B$
    still satisfies $\psi$ but
    does not satisfy $\badphi$.
\end{proof}
\end{sendappendix}

We give the following axiomatic definition
of $\BadOrderS$ as follows
\begin{enumerate}[(i)]
\item $\leq$ is transitive, reflexive and antisymmetric,
\item $\leq$ has connected components of size at least two:
    $\forall a, \exists b. a < b \vee b < a$,
\item $S$ is an injective partial function without fixed points,
\item $S$ and $\leq$ cannot
    conflict:
    $\forall a,b. \neg (S(a,b) \wedge b \leq a)$,
\item There exists a proto-induction principle:
    $\forall a < b, \exists c, S(a,c) \wedge a \leq c \wedge c \leq b$,
\item Edges $E$ can be factorised trough
    $(\leq) (S \setminus {\leq}) (\leq)$:
    $\forall a,b.
    E(a,b) \implies \exists c_1, c_2. 
    a \leq c_1 \wedge c_1 S c_2 \wedge \neg (c_1 \leq c_2) \wedge c_2 \leq b$,
\item Pre-images through $E$ form a suffix of the ordering:
    $\forall a,b,c.
    a \leq b \wedge E(a,c) \wedge S(a,b) \implies E(b,c)$,
\item Images through $E$ form a prefix of the ordering:
    $\forall a,b,c.
    b \leq c \wedge E(a,c) \wedge S(b,c) \implies E(a,b)
    $,
\item Images through $E$ are strictly increasing subsets:
    $\forall a,b.
    a S b \wedge a \leq b \implies
    \exists c. E(b,c) \wedge \neg E(a,b)$,
\item Pre-images through $E$ are strictly decreasing subsets:
    $\forall a,b.
    a S b \wedge a \leq b \wedge
    (\exists c. E(c,a))
    \implies
    \exists c. E(c,a) \wedge \neg E(c,b)$,
\item The last element of an order cannot be obtained
    through $E$: 
    $\forall a,b.
    E(a,b) \implies \exists c. S(b,c) \wedge b \leq c$,
\item The relation $(\leq)(S \setminus {\leq})$ is included in $E$:
    $\forall a, b.
    \exists c. a \leq c \wedge S(c,b) \wedge \neg (c \leq b)
    \implies
    E(a,b)
    $.
\end{enumerate}

The goal of the rest of this section is to prove
that $\BadOrder = \BadOrderS$.
To simplify notations, let us
consider for a structure $A \in \ModF(\upsigma)$
the structure $(A,\leq)$ to be
the structure $A$ without the relations $S$ and $E$.
Using this convention, a $\leq$-component of a structure $A$
is defined as a connected component of the Gaifman Graph of $(A, \leq)$.
\begin{ifappendix}
The detailed proofs of \cref{lem:ff:shapebis,lem:ff:shape}
can be found in \cref{app:sec:ff:generic}.
\end{ifappendix}

\begin{restatable}[$\BadOrderS$ models $\BadAxiom$]{lemma}{badorderaxiomatic}
\label{lem:ff:shapebis}
The class $\BadOrderS$ is contained in $\BadOrder$.
\end{restatable}
\begin{sendappendix}
\begin{ifappendix}
    \badorderaxiomatic*
\end{ifappendix}
\begin{proof}
    Before listing the axiomatic $\BadAxiom$, notice that
    it suffices to prove that connected structures in $\BadOrderS$
    models $\BadAxiom$ as they are closed under disjoint unions
    of models.

    Let $A$ be a connected structure in $\BadOrderS$,
    that is $A = O_n + \cdots + O_m$ with $2 \leq n \leq m$.
\begin{enumerate}[(i)]
\item $\leq$ is transitive, reflexive and antisymmetric
    over each $O_i$. As a consequence, $\leq$ is transitive 
    reflexive and antisymmetric over their disjoint unions.
\item $\leq$ has connected components of size at least two.
    This holds because $2 \leq n \leq m$.
\item $S$ is an injective partial function without fixed points.
    This holds because $S$ is the successor relation on $O_i$
    and the maximal element of $O_i$ is connected through $S$
    to the minimal element of $O_{i+1}$.
\item $S$ and $\leq$ cannot
    conflict:
    $\forall a,b. \neg (S(a,b) \wedge b \leq a)$.
    This holds because $S$ is the successor relation over $O_i$.
\item There exists a proto-induction principle:
    $\forall a < b, \exists c, S(a,c) \wedge a \leq c \wedge c \leq b$.
    This holds because $S$ is the successor relation over $O_i$.
\item Edges $E$ can be factorised trough
    $(\leq) (S \setminus {\leq}) (\leq)$:
    $\forall a,b.
    E(a,b) \implies \exists c_1, c_2. 
    a \leq c_1 \wedge c_1 S c_2 \wedge \neg (c_1 \leq c_2) \wedge c_2 \leq b$.
    This holds because edges $E$ between $O_i$ and $O_j$ only 
    appear if $j = i+1$ and the maximal element of $O_i$ for $\leq$
    is connected through $S$ to the minimal element of $O_j$ for $\leq$.
\item Pre-images through $E$ form a suffix of the ordering:
    $\forall a,b,c.
    a \leq b \wedge E(a,c) \wedge S(a,b) \implies E(b,c)$.
    This holds because edges $E(a,b)$ between $O_i$ and $O_j$ 
    exists if and only if $a \leq b$ when considered as integers.
\item Images through $E$ form a prefix of the ordering:
    $\forall a,b,c.
    b \leq c \wedge E(a,c) \wedge S(b,c) \implies E(a,b)
    $.
    This holds because edges $E(a,b)$ between $O_i$ and $O_j$ 
    exists if and only if $a \leq b$ when considered as integers.
\item Images through $E$ are strictly increasing subsets:
    $\forall a,b.
    a S b \wedge a \leq b \implies
    \exists c. E(b,c) \wedge \neg E(a,b)$.
    This holds because edges $E(a,b)$ between $O_i$ and $O_j$ 
    exists if and only if $a \leq b$ when considered as integers.
\item Pre-images through $E$ are strictly decreasing subsets:
    $\forall a,b.
    a S b \wedge a \leq b \wedge
    (\exists c. E(c,a))
    \implies
    \exists c. E(c,a) \wedge \neg E(c,b)$.
    This holds because edges $E(a,b)$ between $O_i$ and $O_j$ 
    exists if and only if $a \leq b$ when considered as integers.
\item The last element of an order cannot be obtained
    through $E$: 
    $\forall a,b.
    E(a,b) \implies \exists c. S(b,c) \wedge b \leq c$.
    This holds because edges $E(a,b)$ between $O_i$ and $O_j$ 
    exists if and only if $a \leq b$ when considered as integers.
\item The relation $(\leq)(S \setminus {\leq})$ is included in $E$:
    $\forall a, b.
    \exists c. a \leq c \wedge S(c,b) \wedge \neg (c \leq b)
    \implies
    E(a,b)
    $.
    This holds because edges $E(a,b)$ between $O_i$ and $O_j$ 
    exists if and only if $a \leq b$ when considered as integers
    and $O_i$ is of size $i$.
    \qedhere
\end{enumerate}
\end{proof}
\end{sendappendix}

\begin{restatable}[$\BadOrder \subseteq \BadOrderS$]{lemma}{badorderaxiomone}
    \label{lem:ff:shape}
    A structure $A$ that models $\BadAxiom$ satisfies the following 
    properties:
    \begin{enumerate}[(a)]
        \item If $B$ is a $\leq$-component of $A$
            then the substructure induced by $B$ in $A$
            is isomorphic to a total ordering of size greater than two
            with no $E$ relations, i.e. 
            $(\{ 1, \dots, n \}, \leq,+1, \emptyset)$ with $n \geq 2$;
        \item If $B_1$ and $B_2$ are two $\leq$-components of $A$
            that are connected in $A$ with the relation $S$, either
            the last element of $B_1$ is connected to the first one of $B_2$
            or the last element of $B_2$ is connected to the first one of $B_1$;
        \item If $B_1$ and $B_2$ are two $\leq$-components of $A$
            connected through the relation $E$,
            then $B_1$ and $B_2$ are connected through the relation $S$;
        \item If $B_1$ and $B_2$ are two $\leq$-components of $A$
            connected through the relation $S$,
            the function
            $f \colon a \mapsto \max_{\leq} \setof{d}{E(a,d)}$
            is a $\leq$-strictly increasing non-surjective 
            function
            from $B_1$ to $B_2$,
            mapping the $\leq$-minimal
            element
            of $B_1$ to the $\leq$-minimal of $B_2$
            satisfying $f(S(a)) = S(f(a))$.

        \item Connected components of $A$ are in $\BadOrderS$.
    \end{enumerate}
\end{restatable}
\begin{sendappendix}
    \begin{ifappendix}
        \badorderaxiomone*
    \end{ifappendix}
\begin{proof}
Let $A$ be a structure in $\BadOrder$,
without loss of generality, assume that the Gaifman graph of $A$
has a single connected component.

\begin{enumerate}[(a)]
    \item
        Notice that the property $(v)$ combined with
        the fact that $S$ is a partial injective function 
        with no fixed points
        and that $A$ is finite
        allows to prove
        by induction that
        whenever $a \leq b$ there exists $0 \leq k$
        such that $S^k(a) = b$
        and all the intermediate points are between $a$ and $b$
        for $\leq$.

        As a consequence, if $a \leq b_1$ and $a \leq b_2$,
        there exists $0 \leq k$ and $0 \leq l$
        such that $S^k(a) = b_1$ and $S^l(a) = b_2$.
        Without loss of generality $k \leq l$ 
        and $b_2 = S^{k-l}(b_1)$ and 
        $b_1 \leq b_2$.
        Similarly, if $a_1 \leq b$ and $a_2 \leq b$
        there exists $0 \leq k,l$ such that
        $S^{k}(a_1) = b$ and $S^{l}(a_2) = b$

        Combined with $(i)$ and $(iii)$
        this proves that connected components of $A$
        through the relation $\leq$ are \emph{totally ordered}
        by $\leq$ and over these components $S$ is the successor
        relation, while $(ii)$ states that this connected component must have at
        least two elements.

        Assume by contradiction that
        some $\leq$-connected component of $A$
        contains a relation $E(a,b)$
        notice that property $(vii)$ provides $a \leq c_1 S c_2 \leq b$
        and $\neg (c_1 \leq c_1)$ but this is absurd 
        since the connected component is a total ordering
        and property $(v)$ states that we cannot have $c_1 S c_2$ and
        $c_2 \leq c_1$.

\item Assume $B_1$ and $B_2$ are two $\leq$-components of $A$
    connected through the relation $S$.
    As $S$ is a partial injective function
    shows that this can only happen by
    connecting an element that has no successor to an element
    that has no predecessor.
    As a consequence, it is only possible to connect the last element of $B_1$
    to the first one of $B_2$ and vice-versa.

\item Assume that $B_1$ and $B_2$ are two $\leq$-components of $A$
    connected through the relation $E$, that is there exists 
    $a \in B_1$ and $b \in B_2$ such that $E(a,b)$ holds.
    The property $(vi)$ provides
    $c_1$ and $c_2$ such that
    $a \leq c_1$, $c_1 S c_2$, $c_2 \leq b$ and $\neg (c_1 \leq c_2)$.
    As a consequence $c_1$ is in $B_1$ (connected through $\leq$),
    $c_2 \in B_2$ (connected through $\leq$)
    and therefore $B_1$ is connected to $B_2$ through the relation $S$.

\item Assume $B_1$ and $B_2$ are two $\leq$-components of $A$
    connected through the relation $S$.
    The function $g \colon a \mapsto \setof{b}{E(a,b)}$
    from $B_1$ to $\wp(B_2)$ is well-defined
    since we proved that there cannot be edges $E$ outside of $B_2$.
    Since $B_2$ is a finite total ordering with respect to $\leq$
    and the image of $g$ is non-empty thanks to $(xii)$
    the function $f$ is well-defined.

    Property $(ix)$ states that whenever $a,b \in B_1$ and $a S b$ then
    $g(a) \subsetneq g(b)$, while property $(viii)$ states
    that $g(a)$ is $\leq$-downwards-closed in $B_2$.
    As a consequence, $f$ must be strictly increasing.

    Similarly, property $(vii)$ combined with property
    $(x)$ states that if $a S b$ in $B_1$,
    then $|g(a)|+1 = |g(b)|$
    and as a consequence $S(f(a)) = f(b)$.

    Finally, property $(xi)$ states that
    $g$ never covers the last element of $B_2$, and in particular
    $f$ is not surjective.

\item Notice that we proved edges $E$ can only appear
    between two $\leq$-components $B_1$ and $B_2$.
    Let us write $0_1$ and $0_2$ their respective $\leq$-minimal elements.
    We showed that
    $f(S^k(0_1)) = S^k(0_2)$ whenever
    $S^k(0_1)$ is in $B_1$.
    Moreover, if $a = S^k(0_1) \in B_1$
    then
    $\setof{d}{E(a,d)} = \downarrow_{\leq} f(a)
    = \setof{ S^l(0_2) }{l \leq k}$.
    Moreover, property $(x)$
    shows that there cannot exist more than two
    points in $B_2$ that are not obtainable through $f$,
    and as a consequence $|B_2| = |B_1| + 1$.

    Finally, a general connected component
    of $A$
    is of the form $O_n + \cdots + O_m$
    with $2 \leq n \leq m$.
    \qedhere
\end{enumerate}
\end{proof}
\end{sendappendix}

Through \cref{lem:ff:shape} and
\cref{lem:ff:shapebis} we learn that $\BadOrderS$ is definable
using finitely many universal local sentences.
We can lift the counter example provided in
\cref{lem:ff:counterbads} using \cref{fact:ff:relat}.

\begin{corollary}[Counter example over $\ModF(\upsigma)$]
    \label{cor:ff:counterex}
    There exists a sentence $\varphi$
    preserved under disjoint unions over $\ModF(\upsigma)$
    but not equivalent to an existential local
    sentence over $\ModF(\upsigma)$.
\end{corollary}


\subsection{Undecidability}
\label{sec:ff:undecidable}

We have proven that deciding whether a sentence preserved under 
local elementary embeddings is equivalent to an
existential local sentence is a non-trivial problem in $\ModF(\upsigma)$.
We strengthen this by proving said problem is actually undecidable.
This work was initially
inspired by the work of \citet*{Kuperberg21}
proving that Lyndon's Positivity Theorem fails for finite words
and providing undecidability of the associated decision problem.
Moreover, the statements and theorems can be interpreted
as variations of those from
\citet*{flum2021},
who considers preservation under induced substructures.

Given a Universal Turing Machine $U$
over an alphabet $\Sigma$ and with control states
$Q$,
we extend the signature of $\BadOrderS$ with
unary predicates $q/1$ for $q \in Q$,
$P_a/1$ for $a \in \Sigma$,
$P_\$/1$ and $P_\square/1$
to encode configurations of $U$ in $\leq$-components.
Without loss of generality, we assume that this Universal
Turing Machine accepts only on a specific state $q_f^a \in Q$
and rejects only on a specific state $q_f^r \in Q$;
those two states being the only ones with no possible forward transitions.
In a structure $A$, we call $C(a)$ the $\leq$-component of $a \in A$.

\begin{fact}
    There exists a $1$-local formula $\theta_C(x)$
    such that for all structure $A \in \BadOrderS$, element $a \in A$,
    $A, a \models \theta_C(x)$ if and only if
    $C(a)$ represents a valid configuration of $U$.
\end{fact}

The only difficulty in representing 
runs of the machine $U$ is
to map positions from one $\leq$-component
to its successor.
To that end, we exploit the
first-order definability of the function 
$f \colon a \mapsto \max_{\leq} \setof{d}{E(a,d)}$
that links $\leq$-components
\begin{ifappendix}%
    (see~\cref{app:sec:ff:undecidable}).%
\end{ifappendix}

\begin{restatable}[Transitions are definable]{lemma}{definabletransitions}
    There exists a $1$-local formula $\theta_T(x,y)$
    such that for every structure
    $A \in \BadOrderS$ and points $a,b \in A$
    the following two properties are equivalent:
    \begin{enumerate}[(i)]
        \item The $\leq$-components of $a$ and $b$ 
            are connected through $S$ and
            represent valid configurations $C$ and $C'$
            satisfying $C \to_U C'$,
        \item $A, ab \models \theta_T(x,y)$.
    \end{enumerate}
\end{restatable}
\begin{sendappendix}
\begin{ifappendix}
    \definabletransitions*
\end{ifappendix}
\begin{proof}
    This follows the standard encoding of transitions.
    We check that the sentence is $1$-local since
    the ball of radius $1$ around $a$ (or $b$) contains
    both $\leq$-components entirely.

    The only technical issue is relating
    positions in the configuration $C(a)$ to positions
    in the configuration $C(b)$, which 
    is done through the use of the function
    $f \colon a \mapsto \max_{\leq} \setof{d}{E(a,d)}$
    which is first-order definable and $1$-local.
    Let us spell out the definition of $f$ as a formula:
    \begin{equation*}
        \phi_f(x,y) \defined E(x,y) \wedge \forall z. y < z \implies \neg E(x,z)
    \end{equation*}

    For instance, to assert
    that every letter except those near the current position
    in the tape are left unchanged, one can first write a formula
    stating that the position of the head of the Turing Machine
    in the $\leq$-component of $x$ is not close to $x$:
    \begin{equation*}
        \phi_Q(x) \defined \forall z. (z \leq x \vee x \leq z) \wedge 
        (S(x,z) \vee S(z,x) \vee x = z) \implies \bigwedge_{q \in Q} \neg q(x)
    \end{equation*}

    As a shorthand, let us write $z \in C(x)$ instead of $z \leq x \vee x \leq
    z$.
    Using $\phi_f$ and $\phi_Q$ it is easy to write a formula stating
    that letters far from the head of the Turing Machine are unchanged
    in a transition:
    \begin{align*}
        \phi(x,y) \defined
        &\forall z. z \in C(x) \wedge \phi_Q(z) \\
        &\implies \exists z'. z' \in C(y)
        \wedge \phi_f(z,z')
        \wedge
        \bigwedge_{a \in \Sigma \cup \{ \$, \square \}} P_a(z) \Leftrightarrow P_a(z')
    \end{align*}

    We can assert that some specific transition has been
    taken similarly, by first checking the movement of the head,
    change of letters around the head, and the evolution of the state.
\end{proof}
\end{sendappendix}

\begin{fact}
    Given a word $w \in \Sigma^*$
    there exists a $1$-local formula $\theta_I^w(x)$
    such that for $A \in \BadOrderS$
    and $a \in A$,
    $A,a \models \theta_I^w(x)$ if and only
    if $C(a)$ is the initial configuration of $U$
    on $w$
    and has no $S$-predecessor in $A$.
\end{fact}

\begin{fact}
    Given a word $w \in \Sigma^*$
    there exists a $1$-local formula $\theta_F(x)$
    such that for $A \in \BadOrderS$ and $a \in A$,
    $A,a \models \theta_F(x)$ if and only if
    $C(a)$ is a final configuration of $U$
    and has no $S$-successor in A.
\end{fact}

\begin{fact}
    There exists a $1$-local formula $\theta_N(x)$
    such that for $A \in \BadOrderS$
    and $a \in A$, 
    $A, a \models \theta_N(x)$ if and only
    if $C(a)$ has no $S$-successor in $A$.
\end{fact}

\begin{theorem}[Undecidability]
    \label{thm:ff:undecidable}
It is in general not possible
to decide whether a sentence $\varphi$
closed under disjoint unions
over $\ModF(\upsigma)$
has an existential local equivalent
form.
\end{theorem}
\begin{proof}
    Without loss of generality
    thanks to \cref{fact:ff:relat}, we only work over
    $\BadOrderS$, and
    we reduce from the halting problem.

    Let $M$ be a Turing Machine
    and $\langle M \rangle$ its code in the alphabet of the
    Universal Turing Machine $U$.
    Let $\varphi_M = \exists x. \theta_I^{\langle M \rangle}(x)
    \wedge
    \forall x,y. S(x,y) \wedge \neg (x \leq y) \implies
    \theta_T(x,y)$.
    We consider the sentence $\varphi \defined \varphi_M \vee \varphi_{CC}$
    that is closed under disjoint unions over $\BadOrderS$. This sentence
    is computable from the data $\langle M \rangle$.

    Assume that $M$ halts in at most $k'$ steps,
    there exists a bound $k$ for the run of the universal
    Turing Machine $U$.
    Given a size $n \in \mathbb{N}$, we define
    $\varphi_M^n$ to be the existential local sentence
    $\exists x_1, \dots, x_n. \theta_I^{\langle M \rangle}(x_1)
    \wedge \theta_T(x_1, x_2) \wedge \dots \wedge \theta_T(x_{n-1}, x_n)
    \wedge \theta_N(x_n)$.
    It is a routine check that $\varphi$ is equivalent to
    $\varphi_{CC} \vee \bigvee_{1 \leq k \leq n} \varphi_M^n$ over $\BadOrderS$.

    Assume that $M$ does not halt.
    The universal Turing Machine $U$ does not halt
    on the word $\langle M \rangle$.
    Assume by contradiction that $\varphi$ is equivalent to a sentence
    $\psi \defined \exists x_1, \dots, x_k. \theta$
    where $\theta$ is $r$-local.
    Find a run of size greater than $k \cdot (2r + 1)$
    and evaluate $\psi$.  It cannot look at all the configurations
    simultaneously; change the state of the one not seen: it still satisfies
    $\psi$ but this is no longer a run of $U$, which is absurd.
\end{proof}

\begin{ifappendix}%
    An analogous proof allows us to conclude
    that the semantic property of closure under disjoint unions
    is also undecidable (see~\cref{app:sec:ff:undecidable}).
    \Cref{thm:ff:undecidable} is then strengthened
    to prove that equivalent existential sentences
    are uncomputable in general.%
\end{ifappendix}

\begin{restatable}[Undecidable semantic property]{theorem}{undecidabledisunion}
    It not possible, given a sentence $\varphi$,
    to decide whether or not it is closed under disjoint unions
    over $\ModF(\upsigma)$.
\end{restatable}%
\begin{sendappendix}%
    \begin{ifappendix}%
        \undecidabledisunion*%
    \end{ifappendix}%
\begin{proof}
    Without loss of generality, we only work over
    $\BadOrderS$, and
    we reduce from the halting problem.
    Consider the sentence
    $\varphi_M \defined \exists x. \theta_I^{\langle M \rangle}(x) 
                  \wedge 
    \exists x. \theta_F(x) 
                  \wedge
    \forall x,y. S(x,y) \wedge \neg (x \leq y) \Rightarrow
    \theta_T(x,y)$,
    and let $\varphi = \varphi_M \wedge \neg \varphi_{CC}$.

    Assume that $M$ halts. There
    exists exactly one model of $\varphi$, the unique
    run of the universal Turing machine $U$,
    and this contradicts closure under disjoint unions.

    Assume that $M$ does not halt.
    The sentence $\varphi$ has no finite model and
    in particular is closed under disjoint unions.

    Hence, the sentence $\varphi$ is
    closed under disjoint unions
    if and only if $M$ does not halt.
\end{proof}
\end{sendappendix}%

\begin{theorem}[Uncomputable equivalence]
    There is no algorithm,
    which given a sentence $\varphi$
    that is equivalent to an existential local
    sentence over $\ModF(\upsigma)$
    computes such a sentence.
\end{theorem}
\begin{proof}
    Without loss of generality
    thanks to \cref{fact:ff:relat}, we only work over
    $\BadOrderS$, and
    we reduce from the halting problem.
    Let $M$ be a Turing Machine,
    $
        \varphi_M \defined \exists x. \theta_I^{\langle M \rangle}(x) 
                  \wedge 
    \exists x. \theta_F(x) 
                  \wedge
    \forall x,y. S(x,y) \wedge \neg (x \leq y) \implies
    \theta_T(x,y)
    $,
    and $\varphi \defined \varphi_M \vee \varphi_{CC}$.

    Assume that $M$ halts. Then $\varphi$
    is equivalent to an existential local sentence
    $\varphi' \vee \varphi_{CC}$,
    as noticed in the proof of \cref{thm:ff:undecidable}.
    Assume that $M$ does not halt.
    Then $\varphi$ is equivalent to $\varphi_{CC}$,
    which is existential local.
    We have proven that $\varphi$ is equivalent to an existential
    local sentence in all cases.

    Assume by contradiction that there exists 
    an algorithm
    computing an existential local sentence $\mu$ 
    that is equivalent to $\varphi_M \vee \varphi_{CC}$ over $\BadOrderS$.
    One can use $\mu$ to decide whether $M$ halts.
    Let us write $k$
    for the number of existential quantifiers of $\mu$
    and $r$ for the locality radius of its inner formula.
    If the sentence $\mu$ accepts the coding of a run of size greater
    than $2k \cdot (2r + 1)$, then it accepts
    structures that are not coding runs of $U$,
    and $\mu$ cannot be equivalent to $\varphi$.
    As a consequence, if $M$ terminates, it must
    terminate in
    at most $2k \cdot (2r + 1)$ steps, which is decidable.
\end{proof}





\subsection{Generalisation to weaker preorders}
\label{sec:ff:generalise}

We characterise in this section the set of parameters $(r,q,k)$ for which
a sentence preserved under $\locleq{r}{q}{k}$ over
$\ModF(\upsigma)$ is equivalent to an existential local sentence
through a finer analysis of
$\BadAxiom$.

\begin{lemma}[Failure at arbitrary radius]
\label{lem:failure:counterexgen}
There exists a sentence $\badphi$
preserved under $\locleq{\infty}{q}{k}$ over $\ModF(\upsigma)$
for $k \geq 2$ and $q \geq 1$
that is not equivalent to an existential local
sentence over $\ModF(\upsigma)$.
\end{lemma}
\begin{proof}
    Using \cref{fact:qo:refinement} it suffices to consider
    the case where $k = 2$ and $q = 1$ to conclude.

    A first check is that $\BadOrderS$
    can alternatively be defined by
    sentences that have at most $2$ universal quantifiers
    over variables $x_1$ and $x_2$
    and then one variable $y$ at distance at most 1
    from the tuple $x_1x_2$.
    This proves  that $\BadOrderS$ is definable
    and downwards closed for $\locleq{1}{1}{2}$.
    In particular, sentences preserved under $\locleq{\infty}{1}{2}$
    over $\BadOrderS$ are preserved under $\locleq{\infty}{1}{2}$
    over $\ModF(\upsigma)$, which is a strengthening of
    \cref{fact:ff:relat}.

    Using the same syntactical analysis, $\varphi_{CC}$ is preserved
    under $\locleq{\infty}{1}{2}$ over $\BadOrderS$.
    As a consequence it is an easy check that
    $\badphi \defined \forall x. \neg B(x) \vee \varphi_{CC}$
    is preserved under $\locleq{\infty}{1}{2}$
    over $\BadOrderS$.
    Moreover, it was stated in \cref{lem:ff:counterbads}
    that $\badphi$ cannot be defined as an existential local sentence
    over $\BadOrderS$.
\end{proof}

\begin{ifappendix}%
Using similar techniques one can tackle the case $q = \infty$, $k \geq 2$
and $r \geq 1$.
When $k = 1$ such methods will not apply because
the preorders $\locleq{r}{q}{1}$
cannot distinguish a structure $A$ from $A \uplus A$.
We use this fact in~\cref{app:sec:ff:generalise}
to provide a positive answer in this case.%
\end{ifappendix}

\begin{restatable}[Failure at arbitrary quantifier rank]{lemma}{arbitraryquantifierrank}
    \label{lem:ff:arbitraryqr}
    For every $r \geq 1, k \geq 2$, 
    there exists a sentence $\varphi$ preserved under
    $\locleq{r}{\infty}{k}$
    over $\ModF(\upsigma)$
    but not equivalent to an existential local sentence
    over $\ModF(\upsigma)$.
\end{restatable}
\begin{sendappendix}%
    \begin{ifappendix}%
        \arbitraryquantifierrank*%
    \end{ifappendix}%
\begin{proof}
    Using \cref{fact:qo:refinement} it suffices
    to consider the case $r = 1$ and $k = 2$.
    The proof follows the same pattern as \cref{lem:failure:counterexgen},
    we enumerate the axioms from $\BadAxiom$
    and notice that they are of the form $\forall x. \theta(x)$
    where $\theta(x)$ is a $1$-local formula.
    As a consequence,
    sentences preserved under $\locleq{1}{\infty}{2}$
    over $\BadOrderS$ are preserved under $\locleq{1}{\infty}{2}$
    over $\ModF(\upsigma)$, which is a strengthening of
    \cref{fact:ff:relat}.

    Moreover, $\varphi_{CC}$ preserved under $\locleq{1}{\infty}{2}$
    using a simple syntactical analysis.
    We now check that $\badphi$ is preserved under $\locleq{1}{\infty}{2}$
    over $\BadOrderS$.
    Let $A, B \in \BadOrderS$ such that $A \models \badphi$
    and $A \locleq{1}{\infty}{1} B$.

    \begin{itemize}
        \item Let us first examine the case
     where $A$ has a single connected component.

    Let $a \in A$; since $\Neighb{A}{a}{1}$ is finite,
    there exists a $1$-local formula $\psi_a(x)$
    of quantifier rank less than $|\Neighb{A}{a}{1}|+1$
    such that $B, b \models \psi_a(x)$ if and only 
    if $\Neighb{A}{a}{1}$ is isomorphic to $\Neighb{B}{b}{1}$.
    In particular, if $C$ is a 
    $\leq$-component of $A$, it is of radius less than $1$
    and there exits $C'$ a $\leq$-component of $B$
    isomorphic to $C$.
    Moreover, the $1$-neighborhood of a $\leq$-component
    contains the previous and next $\leq$-component for $S$.

    If $B$ has a single connected component, then
    two distinct $\leq$-components in $B$ must have distinct sizes.
    Using the fact that the components of $A$ are all found in $B$
    and that their relative position is preserved,
    this proves that $B$ contains exactly the same $\leq$-components as $A$,
    which
    only happens if $A=B$.

    If $B$ has at least two connected components,
    then it satisfies $\varphi_{CC}$ which implies $\badphi$.
        \item In the case where $A$ has two connected components,
            $A \models \varphi_{CC}$ but then
            $B \models \varphi_{CC}$ and $B \models \badphi$.
    \end{itemize}

    Moreover, we know from \cref{lem:ff:counterbads}
    that $\badphi$ cannot be defined as an existential local sentence
    over $\BadOrderS$.
\end{proof}%
\end{sendappendix}%

\begin{sendappendix}
\begin{ifappendix}
Before studying the case $k = 1$, let us
describe the behaviour of $\locleq{r}{q}{k}$
with respect to disjoint unions. In particular,
we prove that at a fixed $k$, the preorder cannot
distinguish between more than $k$ copies of
the same structure.
\end{ifappendix}
\begin{example}[Disjoint unions]
    \label{ex:loc:disjointunions}
    Let $A, B$ be finite structures.
    For all $0 \leq r,q  \leq \infty$
    and $1 \leq k < \infty$,
    $A \locleq{r}{q}{k} A \uplus B$
    and $\biguplus_{i = 1}^{k + n} A \locleq{r}{q}{k} \biguplus_{i = 1}^{k} A$.
\end{example}
\begin{proof}
    Let us prove that
    $A \locleq{r}{q}{k} A \uplus B$.
    Consider a vector $\vec{a} \in A^k$,
    it is clear that this vector appears as-is in $A \uplus B$
    and since the union is disjoint,
    $\Neighb{A \uplus B}{\vec{a}}{r} = \Neighb{A}{\vec{a}}{r}$.
    In particular, $\NeighbT{A}{\vec{a}}{r}{q} = \NeighbT{A \uplus
    B}{\vec{a}}{r}{q}$ for all $q \geq 0$.

    Let us write $A^{\uplus k} \defined \biguplus_{i = 1}^{k} A$
    and $A^{\uplus k+n} \defined \biguplus_{i = 1}^{k+n} A$.
    Remark that the previous statement shows
    $A^{\uplus k} \locleq{r}{q}{k} A^{\uplus k+n}$.
    Let us prove now $A^{\uplus k+n} \locleq{r}{q}{k} A^{\uplus k}$.
    Consider a vector $\vec{a} \in (A^{\uplus k+n})^k$,
    this vector has elements in at most $k$ copies of $A$,
    hence one can select one copy of $A$ in $A^{\uplus k}$ for each
    of those and consider the exact same elements in those copies.
    As the unions are disjoint, the obtained neighborhoods are
    isomorphic to those in $A^{\uplus k+n}$ and in particular share
    the same local types.
\end{proof}

\end{sendappendix}

\begin{restatable}[Success at $k = 1$]{lemma}{successatone}
    \label{lem:ff:successatone}
    For every $0 \leq r,q \leq \infty$,
    for every sentence
    $\varphi$ preserved under $\locleq{r}{q}{1}$
    over $\ModF(\upsigma)$,
    there exists an existential local sentence $\psi$
    that is equivalent to $\varphi$ over $\ModF(\upsigma)$.
\end{restatable}
\begin{sendappendix}%
\begin{ifappendix}%
    \successatone*%
\end{ifappendix}%
\begin{proof}
    Without loss of generality thanks to \cref{fact:qo:refinement}
    we consider a sentence $\varphi$ preserved under
    $\locleq{\infty}{\infty}{1}$.
    Let us prove that $\varphi$ is preserved under $\locleq{r}{q}{k}$
    where all parameters are finite.
    This property combined with \cref{lem:ccl:existlocalnf} will prove that
    $\varphi$ is equivalent to an existential local sentence.

    Let us write $\exists^{\geq k}_r x. \psi(x)$
    as a shorthand for
    $\exists x_1, \dots, x_k. \bigwedge_{1 \leq i \neq j \leq k} d(x_i, x_j) >
    2r \wedge \bigwedge_{1 \leq i \leq k} \psi(x_i)$.
    Thanks to Gaifman's Locality Theorem, we can
    assume that $\varphi$ is a Boolean combination of the following
    basic local sentences for $1 \leq i \leq n$: $\theta_i = \exists^{\geq k_i}_{r_i} x. \psi_i(x)$
    where $\psi_i(x)$ is a $r_i$-local sentence of quantifier rank $q_i$.
    Define $r \defined \max \setof{r_i}{1 \leq i \leq n}$,
    $q \defined \max \setof{q_i}{1 \leq i \leq n}$
    and $k \defined \max \setof{k_i}{1 \leq i \leq n}$.

    Let $A,B$ be two finite structures
    such that $A \models \varphi$ and
    $A \locleq{r}{q}{k} B$. Our goal is to prove that $B \models \varphi$.
    Let us write $B^{\uplus k}$ the disjoint union of $k$ copies of $B$.

    Let us show that $A \uplus B^{\uplus k} \models \varphi$ if and
    only if $B^{\uplus k} \models
    \varphi$. To that end, let us fix $1 \leq i \leq n$ and prove
    that $A \uplus B^{\uplus k} \models \exists^{\geq k_i}_{r_i} x. \psi_i(x)$
    if and only if $B^{\uplus k} \models \exists^{\geq k_i}_{r_i} x. \psi_i(x)$.

    \begin{itemize}

        \item Assume that $A \uplus B^{\uplus k} \models \exists^{\geq k_i}_{r_i} x.
            \psi_i(x)$,
    there exists a vector $\vec{c}$ of witnesses
    of $\psi$
    at pairwise distance
    greater than $2r_i$ in $A \uplus B^{\uplus k}$. 
    If $\vec{c}$ lies in $B^{\uplus k}$ then we conclude, otherwise
    some element of $\vec{c}$ lies in $A$.
    In particular, $A \models \exists x. \psi(x)$. Since $A \locleq{r}{q}{k} B$,
    we know that
    $B \models \exists x. \psi_i(x)$ and
    as $k_i \leq k$, $B^{\uplus k} \models \exists^{\geq k_i}_{r_i} x. \psi_i(x)$.
\item Conversely, whenever $B^{\uplus k} \models \exists^{\geq k_i}_{r_i} x. \psi_i(x)$
    the structure $A \uplus B^{\uplus k}$ satisfies $\exists^{\geq k_i}_{r_i} x. \psi(x)$
    as basic local sentences are preserved under disjoint unions.
    \end{itemize}

    Since $A \locleq{\infty}{\infty}{1} A \uplus B \locleq{\infty}{\infty}{1} A \uplus B^{\uplus k}$
    and $A \models \varphi$, we know that $A \uplus B^{\uplus k} \models \varphi$.
    This implies that $B^{\uplus k} \models \varphi$, and
    we deduce from \cref{ex:loc:disjointunions}
    that $B^{\uplus k}$ is $\locleq{\infty}{\infty}{1}$-equivalent to
    $B$. In particular $B^{\uplus k} \models \varphi$ implies $B \models \varphi$.
\end{proof}
\end{sendappendix}%

We provide in \cref{fig:map-of-all-the-cases}
a panel of the existence 
of an existential local form
for different values for $r,q$ and $k$
over the class of finite structures $\ModF(\upsigma)$,
collecting the results from
\cref{lem:failure:counterexgen,lem:ff:arbitraryqr,lem:ff:successatone},
\cref{lem:ccl:existlocalnf},
and the forthcoming \cref{cor:pthm:localisable}
in the case of induced substructures.


\begin{figure}[htb]
\centering
\def\svgwidth{1\columnwidth}
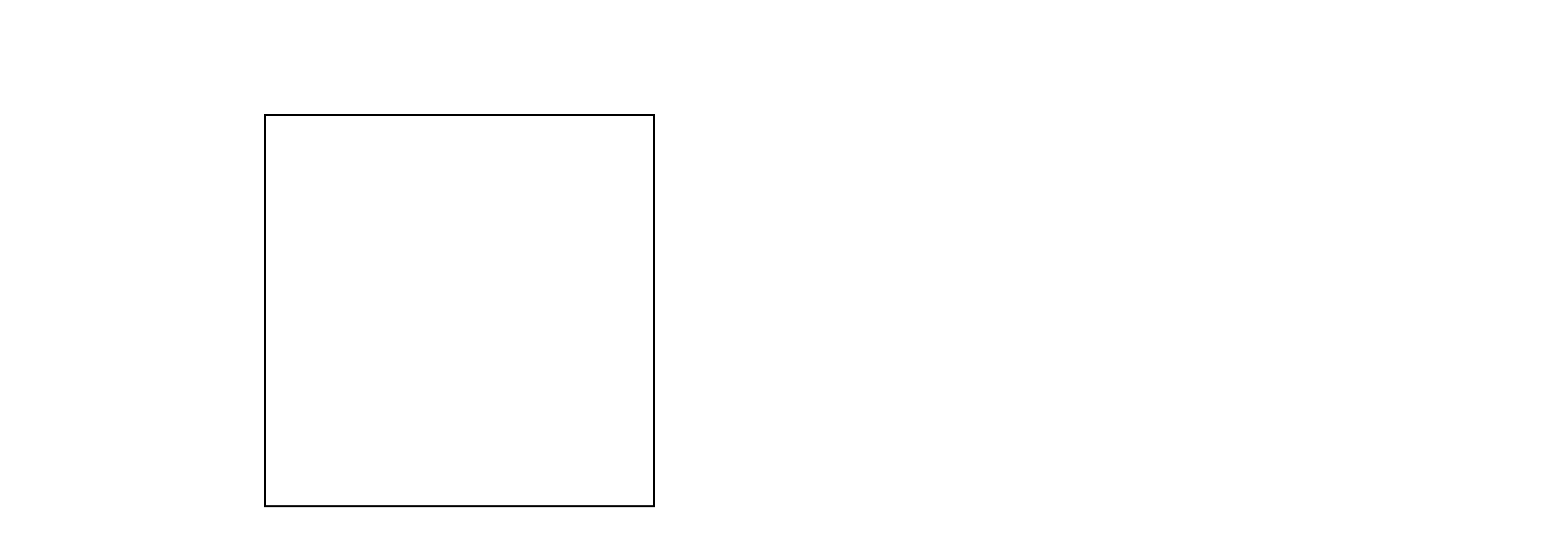

\vspace{-1em}
\caption{Parameters {\boldmath $(r,q,k)$} leading to an existential normal form
(white), those with a counter example (dots).}
\label{fig:map-of-all-the-cases}
\end{figure}

\section{Localising the \ltrsk}
\label{sec:presthm}Let us recall our proof scheme to handle preservation under extensions
over subclasses of $\ModF(\upsigma)$ in two distinct steps:
\begin{inparaenum}
	\item[(\localisable)]
        $\XClass$ is \emph{localisable}, i.e.
	sentences preserved under extensions over $\XClass$
	are equivalent to existential local sentences $\XClass$,
	\item[(\locreduce)]
        $\XClass$ satisfies \emph{existential local preservation under extensions}, i.e.
        existential local sentences preserved under extensions
        over $\XClass$
        are equivalent to existential sentences over $\XClass$.
\end{inparaenum}
We prove in \cref{sec:presthm:inducedsub}
that step (\localisable) can be done on every hereditary class of finite
structures stable under disjoint unions, generalising the proof
of preservation under extensions over structures of bounded degree
of \citet[Theorem 4.3]{atserias2008preservation}.
Moreover, we show in \cref{sec:presthm:localdecstruct} that (\locreduce) can be done
under very mild assumptions on the class. These two results
are represented as a path from sentences preserved under extensions to existential 
sentences in \cref{fig:intro:diagonal}, which
allows us to
significantly improve known classes of finite structures
enjoying preservation under extensions in \cref{sec:presthm:localwb},
as depicted in \cref{fig:considered_classes}.

\subsection{Localisable Classes (\localisable)}
\label{sec:presthm:inducedsub}

Given a sentence $\varphi$ preserved under extensions,
our goal is to prove that it has an existential local normal form.
As a typical first step in preservation theorems, we will
focus on the $\subseteq_i$-minimal
models of $\varphi$. For this study to make sense
we restrict our attention to classes $\XClass$ that are
preserved under induced substructures, also known as hereditary
classes.

\begin{lemma}[Minimal models]
    \label{lem:induced:locmm}
    Let $\XClass$ be a hereditary class of finite structures.
    A sentence $\varphi$ preserved under extensions over
    $\XClass$ has an equivalent existential local normal form
    if and only if there exists $k\geq 1,r \geq 0$
    such that
    its $\subseteq_i$-minimal models are all found in $\Balls{\XClass}{r}{k}$.
\end{lemma}
\begin{proof}[Proof of $\Rightarrow$]
            If $\varphi \equiv \exists \vec{x}. \tau(\vec{x})$ over
        $\XClass$ where $\tau(x)$ is a $r$-local formula, then a
        minimal model $A\in\XClass$ of $\varphi$ necessarily contains
        a vector $\vec{a}$ such that
        $\Neighb{A}{\vec{a}}{r}, \vec{a} \models \tau$. This shows
        that $\Neighb{A}{\vec{a}}{r} \models \varphi$
        where $\Neighb{A}{\vec{a}}{r}\subseteq_i A$ by definition.
        Since $\XClass$ is hereditary,
        $\Neighb{A}{\vec{a}}{r}\in\XClass$, thus $A
        = \Neighb{A}{\vec{a}}{r}$ by minimality and
        $A \in \Balls{\XClass}{r}{|\vec{x}|}$.
\end{proof}\begin{proof}[Proof of $\Leftarrow$]
        Assume that the minimal models
        of $\varphi$ are all found in $\Balls{\XClass}{r}{k}$.
        Let $q$ be the quantifier rank of $\varphi$.
        We are going to show that $\varphi$
        is preserved under
        $\locleq{r}{qr}{k}$ over~$\XClass$ and deduce
        by \cref{lem:ccl:existlocalnf}
        that $\varphi$ has an equivalent existential local
        sentence over~$\XClass$.

        Let $A \models \varphi$ and $A \locleq{r}{qr}{k} B$.  Since
        $A \models \varphi$, there exists a minimal model
        $A_0 \in \Balls{\XClass}{r}{k}$ with $A_0\subseteq_i A$. Let
        $\vec{a} \in A^k$ be the centers of the balls of
        radius $r$ in $A$ that contain $A_0$.  Since
        $A \locleq{r}{qr}{k} B$ there exists a vector $\vec{b} \in B^k$
        such that $\NeighbT{A}{\vec{a}}{r}{qr}
        = \NeighbT{B}{\vec{b}}{r}{qr}$.

        Notice that $A_0 \subseteq_i \Neighb{A}{\vec{a}}{r}$, hence
        $\Neighb{A}{\vec{a}}{r} \models \varphi$ since $\varphi$ is
        preserved under extensions.  Thus, $\Neighb{A}{\vec{a}}{r},\vec
        a\models \varphi_{\leq r}(\vec x)$ where $\varphi_{\leq
        r}(\vec x)$ is an $r$-local formula around its $k$
        variables~$\vec x$ with quantifier rank $qr$. This shows that
        $\Neighb{B}{\vec{b}}{r},\vec b \models \varphi_{\leq r}(\vec
        x)$ using the equivalence of types up to quantifier rank $qr$.
        To conclude, observe that this entails
        $\Neighb{B}{\vec{b}}{r}\models\varphi$, hence
        $B \models \varphi$ since $\Neighb{B}{\vec{b}}{r} \subseteq_i
        B$ and $\varphi$ is preserved
        under extensions.
\end{proof}

The proof of preservation under extensions
over some specific classes
provided by \citet[Theorem 4.3]{atserias2008preservation}
is done by contradiction, using the fact that
minimal structures that are large enough
must contain large scattered sets of points.
Forgetting about the size of the structure,
this actually proves that minimal models
do not have large scattered sets, hence are
in some $\Balls{\XClass}{r}{k}$ for well-chosen parameters.
\begin{ifappendix}
See \cref{app:sec:presthm:inducedsub} for a proof of the following
variant of \citeauthor{atserias2008preservation}'s result.
\end{ifappendix}
\begin{restatable}[Minimal models]{lemma}{inducedminimalmodels}
    \label{lem:failure:inducedminimalmodels}
    Let $\XClass \subseteq \ModF(\upsigma)$
    be a hereditary class of finite structures
    closed under disjoint unions
    and $\varphi \in \FO[\upsigma]$ be
    a sentence
    preserved under $\subseteq_i$ over $\XClass$.
    There exist $R,K$ 
    such that the minimal models of $\varphi$
    are in $\Balls{\XClass}{R}{K}$.
\end{restatable}
\begin{sendappendix}
We provide here a self-contained proof
of \citet[Theorem 4.3]{atserias2008preservation}
tailored to the study of local neighborhoods $\Balls{\XClass}{r}{k}$.
\begin{ifappendix}
\inducedminimalmodels*
\end{ifappendix}
To prove \cref{lem:failure:inducedminimalmodels}, a first attempt
is to
consider a formula $\varphi$ in Gaifman normal form,
and a minimal model $A$ of $\varphi$ with respect to $\subseteq_i$.
The structure $A$ models a conjunction of (potential negations of)
sentences
of the form $\exists^{\geq k}_r x. \psi(x)$,
which we use as a shorthand
for
$\exists x_1, \dots, x_k. \bigwedge_{1 \leq i \neq j \leq k} d(x_i, x_j) >
2r \wedge \bigwedge_{1 \leq i \leq k} \psi(x_i)$.
Considering all basic local sentences that appear positively in the conjunction,
one can build a vector $\vec{a} \in A$ containing the witnesses of the
existential quantifications.
The main hope being that $\Neighb{A}{\vec{a}}{r} \models \varphi$, as
this would imply, by minimality of $A$, that $A = \Neighb{A}{\vec{a}}{r}$.
By letting $R = r$ and $K = |\vec{a}|$, which is bounded
independently of $A$, we would conclude.

Unfortunately the structure $\Neighb{A}{\vec{a}}{r}$ does not
satisfy $\varphi$ in general.
The crucial issue comes
from \emph{intersections}
of neighborhoods: there are new neighborhoods appearing in
$\Neighb{A}{\vec{a}}{r}$ as the intersection of the neighborhood 
of $\vec{a}$ with the neighborhood of a point inside $\Neighb{A}{\vec{a}}{r}$.
This is problematic because $\varphi$ contains basic local sentences that
appear negatively, and the fact that $A$ does not contain some local behaviour
does not transport to $\Neighb{A}{\vec{a}}{r}$.

To tackle this issue, we temporarily leave the realm of first-order
logic and consider $\MSO$ local types,
written $\NeighbM{A}{\vec{a}}{r}{q}$. There are
finitely many $\MSO$ local types up to
logical equivalence at a
given quantifier rank and locality radius.
We update our \emph{type-collector} function accordingly through
\begin{equation}
    \SpecterM{r}{q}{k}(A)
\defined \setof{
\NeighbM{A}{\vec{a}}{r}{q}
}{\vec{a} \in A^k} \, .
\end{equation}

As for first-order local types, $\MSO$ local types
are enough to characterise the natural preorder associated
to \emph{existential local $\MSO$-sentences}, that is, sentences
of the form $\exists \vec{x}. \theta(\vec{x})$, where
$\theta(\vec{x})$ is an $\MSO$ formula $r$-local around $\vec{x}$.
Before going through the proof of~\cref{lem:failure:inducedminimalmodels}
we translate the main properties of local types and existential local
sentences to $\MSO$-local types and existential local $\MSO$-sentences.

\begin{fact}
    \label{fact:pres:mso:specter1}
For all structures $A,B$ such that $\SpecterM{r}{q}{k}(B)$
contains $\SpecterM{r}{q}{k}(A)$,
for all $r$-local $\MSO$ formula $\theta(\vec{x})$ of quantifier rank
less than $q$,
$A \models \exists \vec{x}. \theta(x)$ implies $B \models \exists \vec{x}.
\theta(x)$.
\end{fact}

\begin{fact}
    \label{fact:pres:mso:specter2}
Let $A,B$ be structures such that 
for all $r$-local $\MSO$ formula $\theta(\vec{x})$ of quantifier rank
less than $q$,
$A \models \exists \vec{x}. \theta(x)$ implies $B \models \exists \vec{x}.
\theta(x)$. Then $\SpecterM{r}{q}{k}(B)$
contains $\SpecterM{r}{q}{k}(A)$.
\end{fact}

\begin{lemma}[Local normal form adapted to $\MSO$]
    \label{lem:pres:mso:existlocalnf}
    Let $\XClass$ be a class of finite structures.
    Let $\varphi$ be a sentence preserved over $\XClass$ under $\MSO$-local types with
    finite parameters $(r,k,q)$.
    Then $\varphi$ is equivalent over $\XClass$
    to an existential local $\MSO$ sentence $\psi = 
    \exists \vec{x}. \theta(\vec{x})$
    where $\theta(\vec{x})$ is an $r$-local $\MSO$-formula.
\end{lemma}
\begin{proof}
    It is the same proof as~\cref{lem:ccl:existlocalnf}
    when replacing $\Specter{r}{q}{k}$ with $\SpecterM{r}{q}{k}$,
    leveraging
    \cref{fact:pres:mso:specter1,fact:pres:mso:specter2}.
\end{proof}

An $\MSO$ $(q,r)$-type $t$ with a single free variable
is $R$-\emph{covered} by a subset $C$
if for all realisations $a$ of $t$ in $A$
the $r$-neighborhood $\Neighb{A}{a}{r}$
is included in
$\Neighb{A}{C}{R}$.
The following lemma is a uniform version
of the one given by \citet*[][Lemma 8]{dawar2006approximation}.
It can be thought as a generalisation of
the technique from~\cref{lem:ccl:extentedcover}
to describe the spatial repartition of points of interest
in a structure.

\begin{lemma}[Type covering]
    \label{lem:pres:mso:typecover}
    For all $r,q,k \geq 0$,
    there exists $K_m$ and $R_m$ such that
    for all structures $A \in \Mod(\upsigma)$
    there exists $r \leq R \leq R_m$,
    a subset $C^A \subseteq A$
    and a subset $G^A \subseteq A$
    satisfying the following properties
    \begin{enumerate}[(i)]
        \item $|C^A| \leq K_m$.
        \item $|G^A| \leq K_m$.
        \item Elements of $G^A$ are at pairwise distance greater than $2R$
            and at distance greater than $2R$ of $C^A$.
        \item For every $a \in A$,
            either $\NeighbM{A}{a}{r}{q}$ is $R$-covered
            by $C^A$
            or there exists $k$ elements $b_1, \dots, b_k$
            in $G^A$
            such that
            $\NeighbM{A}{a}{r}{q} = \NeighbM{A}{b_i}{r}{q}$.
    \end{enumerate}
\end{lemma}
\begin{proof}
    Let $r,q,k \geq 0$ be natural numbers and
    consider $Q$ the number of different $\MSO$-local types
    at radius $r$ and quantifier rank $q$ around $1$ variable.
    Define $K_m = Q \times Q \times k$ and $R_m = 3^{Q+1} r$.

    Let $A$ be a structure and consider
    $T \defined \SpecterM{r}{q}{1}(A)$.
    By definition, $|T| \leq Q$.
    We construct by induction sets $S_i \subseteq T$
    and $C_i \subseteq A$ such that
    every type of $S_i$ is $3^i r$-covered by $C_i$
    and
    $|C_i| \leq K_m$.
    Let $S_0 \defined \emptyset$, $C_0 \defined \emptyset$.
    Assume that $S_i$ and $C_i$
    have been defined and that the following holds:
    \begin{align*}
        \forall G \subseteq A,
        &\, \exists u, v \in G, d_A(u,v) \leq 2 \times 3^i r \\
        &\vee  \exists u \in G, d_A(u,C_i) \leq 2 \times 3^i r \\
        &\vee  \exists t \in T \setminus S_i,
        \neg \left(\exists u_1 \neq u_2 \neq \dots \neq u_k \in G, \bigwedge_{1 \leq i \leq k}
        \NeighbT{A}{u_i}{r}{q} = t\right)
    \end{align*}

    We enumerate types in $T \setminus S_i$ in a sequence
    $(t_p)_{1 \leq p \leq |T \setminus S_i|}$.
    Using this sequence, we construct
    iteratively a set $G_i^j$ of size at most $Q \times k$
    such that points of $G_i^j$ are at pairwise distance greater than
    $2 \times 3^i r$ and at distance greater than $2 \times 3^i r$  from $C_i$
    as follows. Let $G_i^0 \defined \emptyset$,
    and construct $G_i^{j+1}$ by selecting the first type $t_p$
    that has less than $k$ witnesses in $G_i^{j}$, this is possible
    by the assumption made on $S_i$ and $C_i$. 
    If there exists a point $a \in A$ at distance greater than $2 \times 3^i r$
    from $G_i^j$ and $C_i$ and of type $t_p$,
    we can add it to $G_i^j$ to build $G_i^{j+1}$.
    Otherwise, every choice of point $a \in A$ of type $t_p$
    is at distance at most $2 \times 3^i r$ from $C_i \cup G_i^j$,
    in this case let $C_{i+1} = C_i \cup G_i^j$ and $S_{i+1} = S_i \cup \{ t_p
    \}$, the hypothesis on $S_i$ is that
    every type in $S_i$ is $3^{i} r$-covered by $C_i$
    and we showed that $t_p$ was $3^{i+1} r$-covered by $C_i \cup G_i^j$.

    The set $S_i$ is strictly increasing and of size bounded by $Q$,
    as a consequence, there exists a step $i \leq Q$ such that one cannot
    build $S_{i+1}$ and $C_{i+1}$.
    By definition, this means that one can build a subset $G \subseteq A$
    such that 
    \begin{align*}
        &\forall u, v \in G, d_A(u,v) > 2 \times 3^i r \\
        &\wedge  \forall u \in G, d_A(u,C_i) > 2 \times 3^i r \\
        &\wedge  \forall t \in T \setminus S_i,
        \exists u_1 \neq u_2 \neq \dots \neq u_k \in G, \bigwedge_{1 \leq i \leq k}
        \NeighbT{A}{u_i}{r}{q} = t
    \end{align*}

    By extracting only $k$ elements per type in $T \setminus S_i$
    from such a set $G$,
    we can construct that $G^A$ of size at most $Q \times k \leq K_m$.
    Let $C^A \defined C_i$ and $R = 3^i r$. The only thing left to check
    is the size of $C^A$, which is below $Q \times Q \times k$ since
    we do at most $Q$ steps and each step adds at most $Q \times k$ elements
    to $C_i$.
\end{proof}


Following the terminology used
by \citeauthor*{dawar2006approximation}
given a subset $C^A$ obtained through
the preceding lemma, types covered by $C^A$
will be called \emph{rare}
and those obtained in $G^A$ will be called
\emph{frequent}.

\begin{figure}[ht]
    \centering
\def\svgwidth{1\columnwidth}
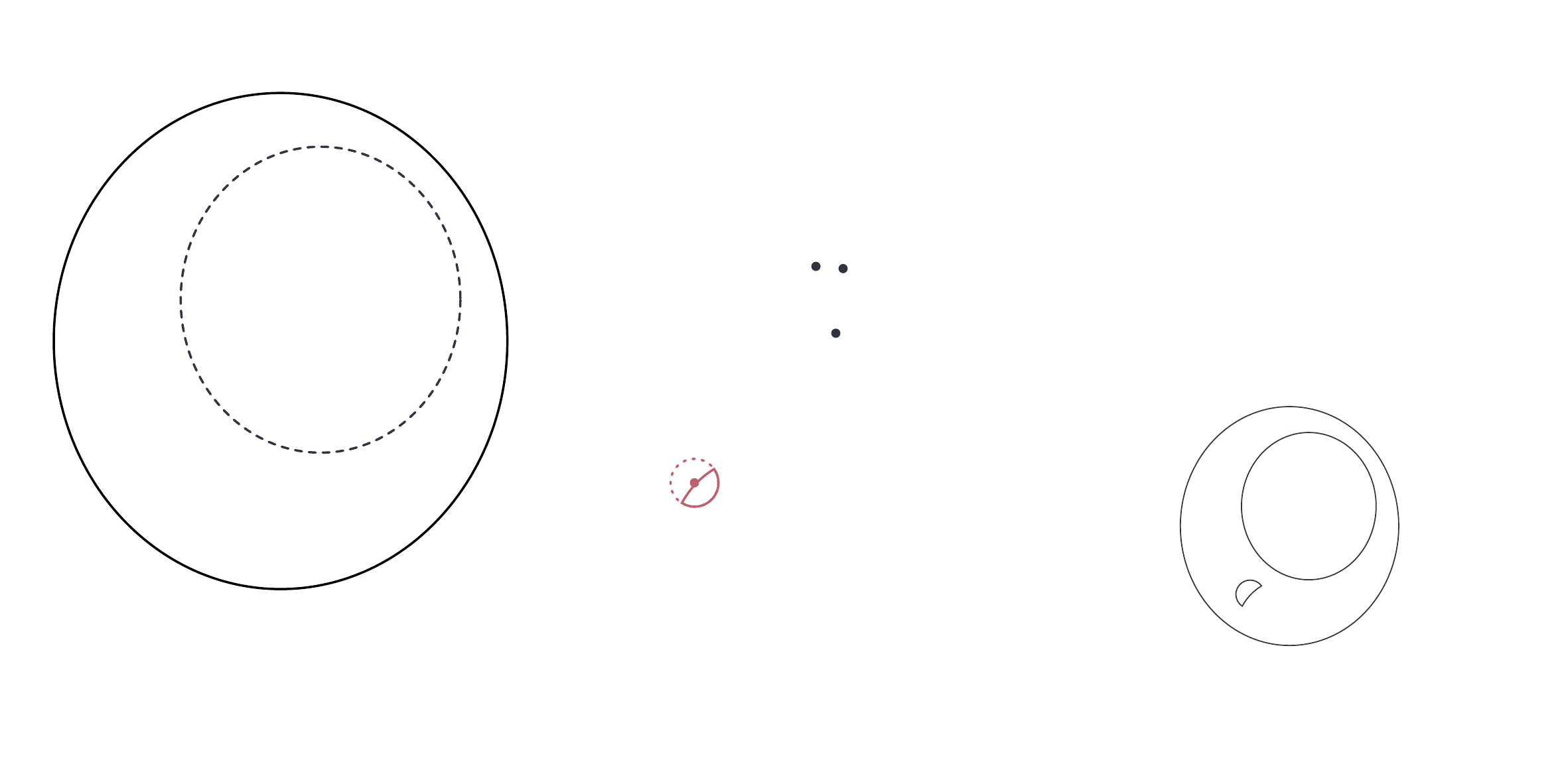

    \caption{Construction of $\Neighb{A}{C^A}{5r} \uplus E'$ using $\MSO$-types.}
    \label{fig:mso_construction}
\end{figure}

\begin{lemma}[Preservation under $\MSO$ types]
    \label{lem:failure:presmso}
    Let $\XClass \subseteq \ModF(\upsigma)$
    that is hereditary and closed under
    disjoint unions
    and $\varphi \in \FO[\upsigma]$
    a sentence
    preserved under $\subseteq_i$ over $\XClass$.
    There exists $R,Q,K$ such that for all $A, B \in \XClass$,
    $A \models \varphi$ and $\SpecterM{R}{Q}{K}(A)
    \subseteq \SpecterM{R}{Q}{K}(B)$
    implies $B \models \varphi$.
\end{lemma}
\begin{proof}
    Consider a sentence $\varphi$
    preserved under $\subseteq_i$ over $\XClass$.
    As a first step, we 
    write $\varphi$ in Gaifman normal form
    and collect $\theta_1, \dots, \theta_l$
    the basic local sentences appearing in this normal form.
    The sentences $\theta_i$ are
    of the form $\exists^{\geq k_i}_{r_i} x. \psi_i(x)$
    where $\psi_i$ is an $r_i$-local formula around $x$
    of quantifier rank $q_i$.
    Recall that we use $\exists^{\geq k_i}_{r_i} x. \psi_i(x)$
    as a shorthand
    for
    $\exists x_1, \dots, x_{k_i}. \bigwedge_{1 \leq p \neq q \leq k_i} d(x_p, x_q) >
    2r_i \wedge \bigwedge_{1 \leq p \leq k} \psi_i(x_p)$.

    Let $r$ be the maximal locality radius of these
    sentences, $q$ their maximal quantifier rank
    and $k$ the maximal number of external existential
    quantifications.
    We use \cref{lem:pres:mso:typecover}
    over the tuple $(2r, 2 \times k \times l, q + 1)$ to obtain
    numbers $2r \leq R_m$ and $k \leq K_m$.
    Define $K = 2 K_m$, $R = 2R_m$ and $Q = 2R_m + k + q + 1 + \max \rk \theta_i$.
    Our goal is to prove that
    $\varphi$ is preserved under $\MSO$-local types with
    parameters $(R,Q,K)$.

    Let us now consider $A \models \varphi$.
    We build
    $C^A$ and $G^A$ as provided
    by \cref{lem:pres:mso:typecover}
    over the tuple $(2r,2k, q + 1)$. Without loss of generality
    assume that the radius given is $R_m$ (the largest possible one).
    The sets $C^A$ and $G^A$ are of size below $K_m$.
    We call $I_f$ the set of indices $1 \leq i \leq l$
    such that $\psi_i(x)$ is in a type represented by $G^A$.
    We call $I_m$ the set of indices $1 \leq i \leq l$
    such that $\exists X. \psi_i^X(x)$
    is in a type represented by $G^A$, where
    $\psi_i^X(x)$ is the relativisation to the set variable $X$
    of $\psi_i$.

    Let $B$ be a structure such that
    $\SpecterM{R}{Q}{K}(A)$
    is contained in
    $\SpecterM{R}{Q}{K}(B)$.
    By definition, there exists
    sets $C^B$ and $G^B$
    in $B$ such that
    \begin{equation}
        \NeighbM{A}{C^A G^A}{R}{Q}
        =
        \NeighbM{B}{C^B G^B}{R}{Q}
        \, .
    \end{equation}

    As $Q$ is large enough
    to check distances up to $R$,
    the distance between two elements in $G^B$
    is greater than $2R$,
    and the distance between one element of $G^B$
    and one element of $C^B$ is greater than $2R$.

    Let us define
    $E \defined B \setminus \Neighb{B}{C^B}{R}$.
    Notice that $\Neighb{B}{G^B}{r} \subseteq E$
    because $2r \leq R_m \leq 2R_m = R$ and the distance between $E$ and $C^B$
    is greater than $2R_m$.
    Let $i \in I_m$, one can choose $k$ elements in $G^B$
    such that $B, b \models \exists X. \psi_i^X(x)$,
    let us call this vector $\vec{b}_i^m$.
    Let $i \in I_f$, one can choose $k$ elements in $G^B$
    such that $B, b \models \psi_i^X(x)$,
    let us call this vector $\vec{b}_i^f$.
    Without loss of generality since types in $G^B$ have
    $2 \times k \times l$ witnesses, we can assume that
    none of these vectors share elements. 
    Let $i \in I_m$ and $b \in \vec{b}_i^m$,
    there exists $b \in F_b \subseteq \Neighb{B}{b}{r}$
    such that $\Neighb{B}{b}{r} \cap F_b, b \models \psi_i(x)$.
    Let us build $E'$ as the structure $E$ where the complementaries
    of the sets $F_b$ have been removed, i.e.
    $E' \defined E \setminus \bigcup_{i \in I_m} \bigcup_{b \in \vec{b}_i^m}
    \Neighb{B}{b}{r} \setminus F_b$.

    We assert that for every $1 \leq i \leq l$
    the following properties are equivalent
    \begin{enumerate}[(i)]
        \item $A \uplus E' \models \theta_i$,
        \item $\Neighb{A}{C^A}{R_m} \uplus E' \models \theta_i$,
        \item $\Neighb{B}{C^B}{R_m} \uplus E' \models \theta_i$.
    \end{enumerate}

    Since $\NeighbM{A}{C^A}{R}{Q} = \NeighbM{B}{C^B}{R}{Q}$
    and $r \leq R_m \leq R$, it is clear that the local types
    $\NeighbM{A}{C^A}{R_m}{Q}$ and $\NeighbM{B}{C^B}{R_m}{Q}$
    are equal.
    Therefore,
    $\Neighb{A}{C^A}{R_m} \uplus E'$ and
    $\Neighb{B}{C^B}{R_m} \uplus E'$ satisfy the same first-order
    sentences at quantifier rank less than $Q$.
    As $Q \geq \rk \theta_i$ for $1 \leq i \leq l$, this proves the equivalence
    between $(iii)$ and $(ii)$.
    Let us now prove that $(i)$ is equivalent to $(ii)$.
    \begin{itemize}
        \item Assume $A \uplus E' \models \theta_i$.
            \begin{itemize}
                \item If all witnesses $a$ of $\theta_i$
                    are such that $\Neighb{A}{a}{r} \subseteq
                    \Neighb{A}{C^A}{R_m}$
                    or $a \in E'$ then we conclude.
                \item If some witness $a$ of $\theta_i$
                    is in $A$ but $\Neighb{A}{a}{r} \not \subseteq
                    \Neighb{A}{C^A}{R_m}$. By definition of $C^A$ the type
                    of $a$ is found at least $2 kl$ times in $G^A$.
                    As a consequence, the vector $\vec{b}_i^f$
                    of size $k$ is found in $G^B$ such that
                    $\Neighb{B}{\vec{b}_i^f}{r} \subseteq E'$
                    and
                    we have proven that $E' \models \theta_i$.
                    In turn, this implies $\Neighb{A}{C^A}{r_1} \uplus E' \models
                    \theta_i$.
            \end{itemize}
        \item Assume that $\Neighb{A}{C^A}{r_1} \uplus E' \models \theta_i$.
            \begin{itemize}
                \item If all witnesses $a$ of $\theta_i$
                    are such that $\Neighb{A}{a}{r} \subseteq
                    \Neighb{A}{C^A}{R_m}$
                    or $a \in E'$ then we conclude.
                \item If some witness $a$
                    is found in $\Neighb{A}{C^A}{R_m}$
                    but $\Neighb{A}{a}{r}$ is not included in
                    $\Neighb{A}{C^A}{R_m}$,
                    then the type of $a$ in $A$ is found at least
                    $2kl$ times in $G^A$.
                    Notice that $A, a \models \exists X. \psi_i^X(x)$,
                    where $X$ can be chosen as $\Neighb{A}{a}{r} \cap
                    \Neighb{A}{C^A}{R_m}$.

                    Let $b \in \vec{b}_i^m$,
                    $E', b \models \psi_i(x)$ because
                    $\Neighb{B}{b}{r}\cap F_b \models \psi_i(x)$
                    and $\Neighb{E'}{b}{r}$ equals $\Neighb{B}{b}{r} \cap F_b$.
                    As the elements of $\vec{b}_i^m$ are at pairwise distance
                    greater than $2R_m \geq 2r$.

                    We have proven that $E' \models \theta_i$,
                    hence, that $A \uplus E' \models \theta_i$.
            \end{itemize}
    \end{itemize}

    The assumption that $A \models \varphi$, which is closed under
    disjoint unions, 
    directly leads to $A \uplus E' \models \varphi$. Moreover, the equivalences
    above
    assert that $\Neighb{B}{C^B}{R_m} \uplus E' \models \varphi$.
    As $\Neighb{B}{C^B}{R_m} \uplus E' \subseteq_i B$
    and since $\varphi$ is preserved under $\subseteq_i$
    we conclude that $B \models \varphi$.

\end{proof}
We finally are able to prove that sentences preserved under extensions
have their minimal models in $\Balls{\XClass}{r}{k}$ for some $r,k \geq 0$
whenever $\XClass$ is hereditary and stable under disjoint unions.
\cref{lem:failure:inducedminimalmodels}.
    \inducedminimalmodels*
\begin{proof}
    Applying  \cref{lem:failure:presmso}
    provides $R,Q,K$
    such that $\varphi$ is preserved
    over $\XClass$ under $(R,Q,K)$ $\MSO$ types.
    Using \cref{lem:pres:mso:existlocalnf},
    $\varphi$ is equivalent over $\XClass$ to a sentence
    $\psi = \exists \vec{x}. \theta(\vec{x})$
    where $\theta(\vec{x})$ is a $r$-local $\MSO$-formula.

    Consider $A$ a minimal model of $\psi$ with respect to $\subseteq_i$,
    in $\XClass$. As $A \models \psi$, there exists a vector $\vec{a} \in
    A^{|\vec{x}|}$
    such that $A, \vec{a} \models \theta(\vec{x})$,
    and since $\theta$ is $r$-local, this
    proves that $\Neighb{A}{\vec{a}}{r}, \vec{a} \models \theta(\vec{x})$.
    As a consequence, $\Neighb{A}{\vec{a}}{r} \models \psi$
    and lies in $\XClass$ at is it hereditary.
    The minimality of $A$ proves that
    $A = \Neighb{A}{\vec{a}}{r}$, thus, that
    $A \in \Balls{\XClass}{r}{|\vec{x}|}$ which is independent of $A$.
\end{proof}
\end{sendappendix}

\begin{theorem}[Localisable Classes \localisable]
    \label{lem:failure:inducedpres}
    If $\XClass \subseteq \ModF(\upsigma)$
    is hereditary and closed under disjoint
    unions then it is localisable.
\end{theorem}
\begin{sendappendix}
    \begin{proof}
        Consider a sentence $\varphi$ preserved under local elementary
        embeddings over $\XClass$. 
        Using \cref{lem:failure:inducedminimalmodels} its minimal models
        are in some $\Balls{\XClass}{r}{k}$,
        and using \cref{lem:induced:locmm}
        this provides an equivalent existential local sentence over $\XClass$.
    \end{proof}
\end{sendappendix}

\begin{corollary}
    \label{cor:pthm:localisable}
    A sentence $\varphi$ preserved under extensions over $\ModF(\upsigma)$
    is equivalent over $\ModF(\upsigma)$ to an existential local sentence.
\end{corollary}

The use of disjoint unions was crucial
in the construction, and removing the assumption
that $\XClass$ is closed under this operation
provides counter-examples to \cref{lem:failure:inducedpres}
\begin{ifappendix}%
    (see~\cref{app:sec:presthm:inducedsub})%
\end{ifappendix}.

\begin{restatable}[Counter example without disjoint unions]{example}{exwithoutdisjointunions}
    Let $\mathcal{C}$ be the downwards closure of the class
    of finite cycles for $\subseteq_i$.
    The sentence $\varphi \defined \forall x. \deg(x) = 2$
    is preserved under extensions
    over $\mathcal{C}$ but not equivalent to an existential
    local sentence over $\mathcal{C}$.
\end{restatable}
\begin{sendappendix}%
    \begin{ifappendix}%
        \exwithoutdisjointunions*%
    \end{ifappendix}%
    \begin{proof}
        This class is \emph{not}
        closed under disjoint unions.
        The sentence $\varphi$ detects the cycles in $\mathcal{C}$ which are
        all incomparable maximal elements for $\subseteq_i$.
        As a consequence, $\varphi$ is monotone for $\subseteq_i$ in $\mathcal{C}$,
        and for every pair $k,r$
        there exists a
        minimal models for $\varphi$
        that is not in $\Balls{\mathcal{C}}{k}{r}$.
\end{proof}
\end{sendappendix}


\subsection{Existential Local Preservation Under Extensions (\locreduce)}
\label{sec:presthm:localdecstruct}

Given an existential local sentence $\varphi$ preserved under extensions,
we want to prove that $\varphi$ is equivalent to an existential sentence.
As existential local sentences focus on neighbourhoods of the structures,
we decompose our
class $\XClass$ of finite structures
into \emph{local neighbourhoods}, that is $\Balls{\XClass}{r}{k}$ with
$r, k$ ranging over the natural numbers.
As we assume $\XClass$ to be
hereditary, this is a subset $\Balls{\XClass}{r}{k}\subseteq\XClass$
of the structures in~$\XClass$.
Quite surprisingly, we do not need closure under disjoint unions
to carry on step (\locreduce).

\begin{lemma}[Locally well-behaved classes]
    \label{lem:pres:locwbpres}
    Let $\XClass$ be a hereditary class of finite structures.
    \begin{enumerate}[(i)]
        \item $\Balls{\XClass}{r}{k}$ satisfies preservation under
        extensions for $r,k \geq 0$,
        \item existential local sentences preserved under extensions
            over $\XClass$ are equivalent to existential sentences.
    \end{enumerate}
\end{lemma}
\begin{proof}
Since
existential sentences are preserved under extensions in general
$(iii) \Rightarrow (i)$, while
$(i) \Rightarrow (ii)$ is simply the hypothesis that $\XClass$ is localisable.

In order to prove $(ii) \Rightarrow (iii)$,
let us consider a sentence $\varphi$
preserved under extensions over $\XClass$.
As $\XClass$ is localisable,
$\varphi$ is equivalent to an existential local sentence
and
\cref{lem:induced:locmm}
provides $k$ and $r$ such that
the minimal models of $\varphi$
are all found in $\Balls{\XClass}{r}{k}$.

As $\Balls{\XClass}{r}{k} \subseteq \XClass$,
the sentence $\varphi$ is also preserved under extensions over
$\Balls{\XClass}{r}{k}$. Since {\puext} holds over
$\Balls{\XClass}{r}{k}$, there exists an existential sentence
$\theta$ equivalent to $\varphi$ over $\Balls{\XClass}{r}{k}$.
Moreover, $\XClass$ is hereditary, which entails that $\theta$ has finitely many
$\subseteq_i$-minimal models $A_1, \dots, A_m$ in $\XClass$.

Let us define $M \defined \max_{1 \leq i \leq m} |A_i|$.  Consider a
minimal model $B$ of $\varphi$ in $\XClass$ and assume by
contradiction that $|B| > M + k + r \cdot M$.  Since
$B \models \varphi$ and $B \in \Balls{\XClass}{r}{k}$,
$B\models\theta$ and there exists a $\subseteq_i$-minimal model $A$ of
$\theta$ such that $A\subseteq_i B$. There exists a vector~$\vec b\in
B^k$ such that $B=\Neighb{B}{\vec b}{r}$.  It is then possible to
define an induced substructure $A'\subseteq_i B$ that contains both $A$
and~$\vec b$ and belongs to $\Balls{\XClass}{r}{k}$: this is ensured
by adding paths taken from~$B$ of length at most~$r$ from each element of $A$
to an element of~$\vec b$.  Hence, $A' \models \theta$ because
$A\subseteq_i A'$, and therefore $A'\models\varphi$ because
$A'\in\Balls{\XClass}{r}{k}$. But $|A'| \leq
|A| + k + r \cdot |A| \leq M + k + r\cdot M < |B|$, which is
absurd.
\begin{ifappendix}
    See \cref{fig:presthm:extracting} in \cref{app:sec:presthm:localdecstruct}
    for a graphical representation
    of this construction.%
\end{ifappendix}%
\begin{ifnotappendix}%
    See \cref{fig:presthm:extracting} for a graphical representation
    of this construction.
\end{ifnotappendix}

\begin{sendappendix}
\begin{figure}
\begin{center}

\begin{tikzpicture}
    \newcommand{\tikztemplategraph}{
        \node[draw, circle,added] (A) at (0.5,0) {};
        \node[draw, circle,temporary] (B) at (2,0) {$b_1$};
        \node[draw, circle,noise] (C) at (3.5,0) {};
        \node[draw, circle,substruct] (D) at (5,0) {};
        \node[draw, circle,substruct] (E) at (0,1) {};
        \node[draw, circle,substruct] (F) at (1,1) {};
        \node[draw, circle,substruct] (G) at (2,1) {};
        \node[draw, circle,noise] (H) at (3,1) {};
        \node[draw, circle,temporary] (I) at (4,1) {$b_2$};
        \node[draw, circle,substruct] (J) at (5,1) {};

        \draw[added] (A) -- (E);
        \draw[added] (A) -- (F);
        \draw[noise] (G) -- (H) -- (C) -- (I) -- (H);
        \draw[substruct] (E) -- (F) -- (G);
        \draw[temporary] (B) -- (G);
        \draw[temporary] (B) -- (A);
        \draw[temporary] (I) -- (D);
        \draw[temporary] (I) -- (J);
    }
    \begin{scope}[yshift=0cm,
        noise/.style={color=black},
        temporary/.style={color=blue,font=\tiny,inner sep=0pt},
        substruct/.style={color=red},
        added/.style={color=teal},
        ]
        \draw (3,-1) node {The whole structure $B\models\varphi$
        is centered around $(b_1,b_2)$.};
        \tikztemplategraph
    \end{scope}
    \begin{scope}[yshift=3cm,
        noise/.style={color=white},
        temporary/.style={color=blue,font=\tiny,inner sep=0pt},
        substruct/.style={color=red},
        added/.style={color=teal},
        ]
        \draw (3,-1) node {The extended substructure $A'\in\Balls{\XClass}{2}{2}$.};
        \tikztemplategraph
    \end{scope}
    \begin{scope}[yshift=6cm,
        noise/.style={color=white},
        temporary/.style={color=white},
        added/.style={color=white},
        substruct/.style={color=red}]
        \draw (3,-1) node {The induced substructure $A\models\theta$.};
        \tikztemplategraph
    \end{scope}

\end{tikzpicture}

\end{center}
\caption{Extracting a structure in $\Balls{\XClass}{2}{2}$
from an induced substructure $A$ of a larger structure $B$.}
\label{fig:presthm:extracting}
\end{figure}
\end{sendappendix}

Therefore, $\subseteq_i$-minimal models of $\varphi$ have bounded size.
It is a usual
consequence that
$\varphi$ is definable as an
existential sentence.
\end{proof}

\subsection{Preservation under Extensions on Locally Well-Behaved Classes}
\label{sec:presthm:localwb}

We can now combine our study of step (\localisable)
in \cref{sec:presthm:inducedsub} and step (\locreduce)
in \cref{sec:presthm:localdecstruct}
to harvest new classes enjoying preservation under extensions,
by characterising classes satisfying
preservation under extensions as those
locally satisfying preservation under extensions.
Quite surprisingly given the
non-relativisation properties of preservation theorems, we are able to
state an equivalence between preservation under extensions over a set
$\XClass$ and preservation under extensions on its local
neighbourhoods.

\begin{restatable}[Local {\puext}]{theorem}{localtarski}
    \label{prop:pres:localtrsk}
    Let $\XClass$ be a hereditary class of finite structures 
    stable under disjoint unions.
    Preservation under extensions holds over $\XClass$
    if and only if 
    preservation under extension holds over $\Balls{\XClass}{r}{k}$
    for all $r \geq 0$ and $k \geq 1$.
\end{restatable}
\begin{sendappendix}
    \begin{ifappendix}
        \localtarski*
    \end{ifappendix}
\begin{proof}
    Assume that the {\puext} holds over 
    $\Balls{\XClass}{r}{k}$ for $r,k \geq 0$. Let $\varphi$ 
    be a sentence preserved under extensions over $\XClass$.
    Because $\XClass$ is hereditary and closed under disjoint unions,
    we can apply
    \cref{lem:failure:inducedpres}, and
    $\XClass$ satisfies (\localisable). Hence, $\varphi$ is equivalent
    to an existential local sentence $\psi$ over $\XClass$.
    Since $\XClass$ is hereditary and locally satisfies
    preservation under extensions, we can apply
    \cref{lem:pres:locwbpres}, and $\XClass$ satisfies (\locreduce).
    Therefore, $\psi$ is equivalent over $\XClass$ to an existential sentence.
    We have proven that $\XClass$ satisfies preservation under extensions.

    Conversely, assume that the {\puext} holds over $\XClass$.
    Let us fix $r,k \in \mathbb{N}^2$
    and consider
    a sentence $\varphi$ preserved under
    extensions over $\Balls{\XClass}{r}{k}$.

    Let $\psi \defined \exists x_1, \dots, x_k.
    \varphi_{|\Neighb{}{\vec{x}}{r}}$,
    i.e.
    the relativisation of $\varphi$ to a neighborhood
    of size $r$ around some elements $x_1, \dots, x_k$.
    It is left as an exercice to prove that $\psi$ and $\varphi$
    are equivalent on $\Balls{\XClass}{r}{k}$.

    Assume that $A, B \in \XClass^2$ are such that
    $A \models \psi$ and $A \subseteq_i B$.
    This amounts to the existence of a strong injective
    morphism $h \colon A \to B$.
    Since $A \models \psi$, there exists $\vec{a} \in A^k$
    (possibly with repetition) 
    such that $\Neighb{A}{\vec{a}}{r} \models \varphi$.
    Moreover, $\Neighb{A}{\vec{a}}{r} \subseteq_i
    \Neighb{B}{h(\vec{a})}{r}$ through the morphism $h$.
    As both neighborhoods are built using radius $r$ and
    tuples of $k$ elements, they belong to $\Balls{\XClass}{r}{k}$.
    The assumption on $\varphi$ implies
    that $\Neighb{B}{h(\vec{a})}{r} \models \varphi$
    and in turn
    this proves
    $B \models \psi$.

    Using the {\puext} over $\XClass$ on $\psi$
    gives an existential sentence $\theta$
    equivalent to $\psi$ over $\XClass$.
    In particular, $\theta$ is equivalent to $\psi$ over
    $\Balls{\XClass}{r}{k}$ since it is a subset of $\XClass$,
    thus $\theta$ is equivalent to $\psi$ on $\Balls{\XClass}{r}{k}$.

    We have proven that {\puext} holds over $\Balls{\XClass}{r}{k}$.
\end{proof}

\end{sendappendix}

The spaces $\Balls{\XClass}{r}{k}$ appear naturally in the study of
\emph{sparse} structures~\cite{Neetil12},
through the notions of \emph{wideness} and \emph{quasi-wideness},
which were already at play in recent proofs
of preservation
theorems~\cite{atserias2006preservation,atserias2008preservation}.
Indeed, they rely crucially on the existence of 
$(r,m)$-scattered sets, that is,
sets of $m$ points with disjoint neighborhoods of size $r$.
Recall that a class $\XClass \subseteq \ModF(\upsigma)$ is \emph{wide}
when there exists
$\rho \colon \mathbb{N}^2 \to \mathbb{N}$,
such that for all $n,m \in \mathbb{N}$,
for all $A \in \XClass$ of size greater than $\rho(r,m)$,
there exists a $(r,m)$-scattered set in $A$.
In particular~\citeauthor*{atserias2008preservation}
prove that for a \emph{wide}, \emph{hereditary},
\emph{closed under disjoint unions}
class of structures $\XClass$, {\puext}
holds~\cite[Theorem 4.3]{atserias2008preservation}.
Hereditary wide classes are exactly those that are
locally finite%
.

\begin{restatable}[\protect{\citet*[Theorem 5.1]{Nesetril2011}}]{fact}{finiteballs}
\label{prop:presthm:finiteballs}
For a hereditary class $\XClass$
it is equivalent to ask for $\XClass$ to be \emph{wide}
or for $\Balls{\XClass}{r}{k}$ to be \emph{finite} for all $r, k \geq 1$.
\end{restatable}
\begin{sendappendix}
    \begin{ifappendix}
        \finiteballs*
    \end{ifappendix}
\begin{proof}
Notice that whenever $A$
has a $(r,m)$-scattered set
then it cannot be a member of $\Balls{\XClass}{r}{k}$ with $k < r$.
Conversely, if $A$ does not contain any $(r,m)$-scattered set,
then it is a union of $m-1$ balls of radius $r$.
Assume that $\XClass$ is wide, the set
$\Balls{\XClass}{r}{k}$ contains only structures of size less than $\rho(r,k)$
and is therefore finite.
Assume that $\Balls{\XClass}{r}{k}$ is finite for every $r,k \in \mathbb{N}$,
then we let $\rho(r,m) \defined \max \setof{|A|}{A \in \Balls{\XClass}{r}{k}, k
< m}$,
whenever $|A| > \rho(r,m)$, the structure $A$ cannot be in $\Balls{\XClass}{r}{k}$
with $k < m$, hence it has a $(r,m)$-scattered set.
\end{proof}
\end{sendappendix}

Over a finite set of finite models, every sentence $\varphi$ is
equivalent to an existential sentence, hence {\puext} holds
trivially. We can recover
\citeauthor{atserias2008preservation}'s result: hereditary wide 
classes closed ender disjoint unions enjoy preservation under extensions.
Indeed,
hereditary wide classes are locally finite (\cref{prop:presthm:finiteballs}),
and hereditary locally finite classes stable under disjoint unions
satisfy preservation under extensions (\cref{prop:pres:localtrsk}).

Let us apply~\cref{prop:pres:localtrsk} to properties known to imply
preservation under extensions: finiteness, bounded tree-depth,
well-quasi-orderings, as depicted in \cref{fig:considered_classes}.
Recall that arrows represent
strict inclusions of properties and dashed boxes are the new properties
introduced in this paper.
\begin{ifnotappendix}
As no logic is involved in the generation of
those classes, we effectively \emph{decoupled} our proofs of preservation
theorems in a \emph{locality argument} followed by
a \emph{combinatorial argument}.
\end{ifnotappendix}

The remaining of this section is devoted to proving that the inclusions
in \cref{fig:considered_classes} are strict, and introducing
the mentioned properties.
Recall here that the \emph{tree-depth} $\td(G)$ of a graph~$G$ is the
minimum height of the comparability graphs~$F$ of partial orders such
that~$G$ is a subgraph of~$F$~\cite[Chapter~6]{Neetil12}.  This
extends as usual to structures by considering
the tree-depth of the Gaifman graphs of said structures.  We shall say
that $\XClass$ has locally bounded tree-depth, if for all $r,k\geq 1$,
there is a bound on the tree-depth of the structures in
$\Balls{\XClass}{r}{k}$.

Note that working with $\Balls{\XClass}{r}{k}$ rather than
$\Balls{\XClass}{r}{1}$ is a somewhat uncommon way to localise
properties.  Thankfully, for properties that are well-behaved with
respect to disjoint unions, the localisation using a
single ball or several ones will coincide; examples of such properties
are wideness, exclusion of a minor, or bounded clique-width.  The
following proposition illustrates this point in the case of bounded
tree-depth.

\begin{restatable}[Locally bounded tree-depth]{lemma}{locboundedtd}
\label{prop:presthm:locboundtd}
A class $\XClass \subseteq \Mod(\upsigma)$ has locally bounded
tree-depth
if and only if
$\exists \rho \colon \mathbb{N} \to \mathbb{N},
\forall A \in \XClass, \forall a \in A, \forall r \geq 1,
\td(\Neighb{A}{a}{r}) \leq \rho(r)$.
\end{restatable}
\begin{sendappendix}
    \begin{ifappendix}
        \locboundedtd*
    \end{ifappendix}
    \begin{proof}
    Assuming that for all $k$ there is a bound
    on the tree-depth of
    elements of $\Balls{\XClass}{r}{k}$, we define
    $\rho(r)$ to be
    the maximum of $\td(A)$
    for $A$ in $\Balls{\XClass}{r}{1}$.

    Conversely, assume that there exists
    an increasing function $\rho$ such that
    $\td(A) \leq \rho(r)$ for every $A \in \Balls{\XClass}{r}{1}$
    a simple induction on $A$ shows that
    $\td(A)  \leq \rho(r \times (2k + 1))$ for every $A \in \Balls{\XClass}{r}{k}$.
    \begin{itemize}
        \item If $A \in \Balls{\XClass}{r}{k}$ can be written $A_1 \uplus A_2$
            then $\td(A) = \max(\td(A_1), \td(A_2))$ and
            we conclude by induction hypothesis since $\rho$ is increasing.
        \item If $A$ is totally connected and in $\Balls{\XClass}{r}{k}$
            then it is included in a ball of radius $r \times (2k + 1)$
            hence the tree-depth is bounded by
            $\rho(r \times (2k + 1))$.
            \qedhere
    \end{itemize}
    \end{proof}
\end{sendappendix}

The following examples
\begin{ifappendix}(detailed
in~\cref{app:sec:presthm:localwb})\end{ifappendix}
prove that the inclusions between these new classes are strict, i.e., that
they strictly generalise previously known classes where preservation
under extensions holds

\begin{example}[Cliques]
    \label{ex:cliques}
    Consider the class $\mathsf{\uplus K}$
    of finite disjoint unions of cliques.
    This class is not of locally bounded tree-depth
    but is well-quasi-ordered (and locally well-quasi-ordered).
\end{example}

\begin{example}[Diamonds]
    \label{ex:diamond}
    Let us call $D_n$
    the cycle of length $n$ extended with two new points
    $a$ and $b$ such that both are connected to every node 
    of the cycle using a path of length $n$.
    Consider the class
    $\mathsf{\uplus D}$ of induced subgraphs of finite
    disjoint unions of some $D_n$.
    This class is not well-quasi-ordered, but is of
    locally bounded tree-depth, hence, locally well-quasi-ordered.
    Moreover, it is not locally finite.
\end{example}

\begin{restatable}[Pointed graphs]{example}{pointedgraphs}
    \label{ex:cpp}
    Consider the class $\Delta_2$ of graphs of degree bounded by $2$.
    The class $P\Delta_2$ is obtained by adding one point connected
    to every other point in a structure of $\Delta_2$.  The class
    is not locally well-quasi-ordered but satisfies
    preservation under extensions.
\end{restatable}
\begin{sendappendix}
    \begin{ifappendix}
        \pointedgraphs*
    \end{ifappendix}
\begin{proof}
    The class $P\Delta_2$ is not locally well-quasi-ordered
    because it contains the infinite antichain of \emph{wheels},
    that are cycles with one added point connected to the whole cycle.
    Les us prove that it satisfies preservation under extensions.

    The class $C\Delta_2$ of
    elements of $\Delta_2$ with one added colour $c$
    satisfies preservation under extensions
    because it is hereditary, closed under disjoint unions 
    and locally finite.

    We will use generic stability properties of preservation
    theorems~\cite[e.g.][Section 5]{lopez2020preservation}
    to transport the theorem over $C\Delta_2$ to
    $P\Delta_2$.
    The class $I\Delta_2$ defined
    as elements of $C\Delta_2$ satisfying
    $\exists x. c(x)$ and $\forall x,y. c(x) \wedge c(y) \implies x = y$
    satisfies preservation under extensions
    as it is a Boolean combination of definable upwards closed
    subsets of $C\Delta_2$.

    Let us define
    $I \colon I\Delta_2 \to P\Delta_2$
    the first-order interpretation
    via $I_E(x,y) \defined c(x) \vee c(y) \vee E(x,y)$.
    The map $I$ is surjective
    and monotone with respect to $\subseteq_i$,
    as a consequence, $P\Delta_2$
    satisfies preservation under extensions.
\end{proof}
\end{sendappendix}

\begin{sendappendix}
\begin{figure}[t]
    \centering
    \begin{tikzpicture}[scale=0.5, circle/.style = {
            inner sep=2pt,
        }]
        \def \diamondsize {5}
        \def \diamondprop {360 / \diamondsize}
        \newcommand{\tpath}[3]{
            \pgfmathsetmacro{\moinsun}{int(#1 - 1)}
            \pgfmathsetmacro{\moinsdeux}{int(#1 - 2)}
            \foreach \i in {1,...,\moinsun} {
                \pgfmathsetmacro{\coefa}{\i / #1}
                \pgfmathsetmacro{\coefb}{1 - \i / #1}
                \node[draw,circle] (#2#3p{\i}) at
                    (barycentric cs:#2=\coefa,#3=\coefb)
                    {};
            }
            \draw (#3) -- (#2#3p{1});
            \draw (#2#3p{\moinsun}) --  (#2);
            \foreach \i in {1,...,\moinsdeux} {
                \pgfmathsetmacro{\plusun}{int(\i + 1)}
                \draw (#2#3p{\i}) --  (#2#3p{\plusun});
            }
        }
        \node[draw,circle] (A) at (-3,0) {};
        \node[draw,circle] (B) at (3,0)  {};
        \node[draw,circle] (C) at (0,2)  {};
        \node[draw,circle] (D) at (0,-2) {};

        \tpath{3}{A}{B};
        \draw (A)  .. controls (0,3.5) .. (B);
        \tpath{4}{C}{A};
        \tpath{4}{C}{B};
        \tpath{4}{D}{A};
        \tpath{4}{D}{B};

        \foreach \y in {1,...,2} {
            \tpath{4}{C}{ABp{\y}}
            \tpath{4}{D}{ABp{\y}}
        }

    \end{tikzpicture}
    \caption{An element of $\mathsf{\uplus D}$, $D_4$}
    \label{fig:diamond_class}
\end{figure}
\end{sendappendix}

\section{Concluding Remarks}
\label{sec:remarks}We investigate
the notion of positive locality
through three fragments of first-order logic: existential local sentences,
a positive variant of the Gaifman normal form, and
sentences preserved under local elementary embeddings.
We prove that those three fragments are equally expressive in
the case of arbitrary structures, but that this fails in the finite.
Following the line of undecidability results for preservation
theorems~\cite{Kuperberg21,flum2021},
we prove that most of the associated decision problems are undecidable
in the case of finite structures.

Maybe surprisingly,
the study of this seemingly arbitrary notion of positive locality
has a direct application in the study of preservation under extensions
over classes of finite structures.
In the case of finite structures,
our notion of local elementary embeddings
describes exactly disjoint unions,
and might explain why they
featured so prominently in the study of
preservation under
extensions~\citep{atserias2006preservation,atserias2008preservation,Rossman08,harwath2014preservation}.

We prove that under mild assumptions on the class $\XClass$
of structures considered, sentences preserved under extensions
can be rewritten as existential local sentences.
We leverage this to craft a locality principle
relating preservation under extensions over the neighborhoods of a class
and preservation under extensions over the whole class.

This allows us to build new classes of structures
where preservation under extensions holds
by localising known properties.
This proof scheme does not always yield 
new classes:
for instance, \emph{nowhere dense} classes
\citep*[see][Chapter 5]{Neetil12}
are locally nowhere dense and vice-versa.

\clearpage
\section*{Acknowledgments} 
 I thank Jean Goubault-Larrecq and Sylvain Schmitz
for their unconditional help and support in writing this paper.
I thank Thomas Colombet for sharing with me his
syntactical proof of the Feferman-Vaught
theorem.

\bibliographystyle{ACM-Reference-Format}
\bibliography{globals/ressources}

\clearpage
\appendix\pagenumbering{roman}
\setlength\textwidth{390pt}
\setlength\oddsidemargin{31pt}
\onecolumn

\begin{ifappendix}
   
\section{Appendix to ``Locality preorders'' (Section \ref{sec:locqo})}\label{app:sec:locqo}

    \begin{ifappendix}
        \lemmaelemloc*
    \end{ifappendix}
\begin{proof}
    Fix $q,r,k \in \mathbb{N}$,
    and consider a tuple $\vec{a} \in A^k$.
    By construction, the tuple $h(\vec{a}) \in B^k$
    satisfies the same local $\FO$ formulas,
    and in particular, $\NeighbT{A}{\vec{a}}{r}{q}
    = \NeighbT{B}{h(\vec{a})}{r}{q}$, thus $A \locleq{r}{q}{k} B$.
\end{proof}

    \begin{ifappendix}
        \exampleinfinitegrid*
    \end{ifappendix}
\begin{proof}
    Assume by contradiction that there exists an elementary embedding $h$ from
    $G'$ to $G$. There exist two points $u, v$ in $G'$
    that are at an infinite distance,
    and their images $h(u), h(v)$ are at a finite distance $r$ in $G$.
    The sentence $\psi(x,y)$ stating that $x$ and $y$ are at distance
    less than $r$ is $r$-local around $x$ and $y$, hence
    is satisfied for $h(u)$ and $h(v)$ if and only if
    it is satisfied in $G'$ for $u$ and $v$; this is absurd.

    However, $G' \locleq{\infty}{\infty}{\infty} G$
    since given $r,q \geq 0$, $k \geq 1$ and a vector $\vec{u} \in (G')^k$
    one can find points in $G$ at
    long enough distances such that the neighborhoods coincide. \qedhere
\end{proof}

\begin{ifappendix}
    \extensionpreorder*
\end{ifappendix}
\begin{proof}
    We first notice that $\locleq{0}{q}{\infty}$
    and $\locleq{r}{0}{\infty}$ cannot quantify
    over the neighborhoods of a vector $\vec{a}$
    since they either contain only $\vec{a}$
    or cannot quantify at all. This explains the first equality.

    \begin{enumerate}
        \item Assume $A \locleq{0}{q}{\infty} B$ for some $q \in \mathbb{N} \cup
            \{ \infty \}$, consider
            $k \defined |A|$
            and let $\vec{a} \in A^k$ be the list of all elements in $A$.
            Assume without loss of generality that $q$ is finite.
            The inequality provides a vector $\vec{b} \in B^k$
            such that $\NeighbT{A}{\vec{a}}{0}{q} = \NeighbT{B}{\vec{b}}{0}{q}$.
            In particular, $\Neighb{A}{\vec{a}}{0} = A$ and
            $\Neighb{B}{\vec{b}}{0} = \vec{b}$.
            Since it is possible to express using quantifier free formulas
            all the relations between elements in $A$,
            the substructure of $B$ induced by $\vec{b}$
            is isomorphic to $A$.
            In particular, $A \subseteq_i B$.

        \item Assume that $A \subseteq_i B$ though a map $h$,
            let $k \in \mathbb{N}$
            and $\vec{a} \in A^k$.
            Let us define $\vec{b} \defined h(\vec{a})$;
            in particular, $\Neighb{A}{\vec{a}}{0} = A$ and
            $\Neighb{B}{\vec{b}}{0} = h(\vec{a})$.
            By definition of the map $h$,
            $h(\vec{a})$ induces a substructure of $B$
            that is isomorphic to the substructure induced by $\vec{a}$ in $A$.
            This entails that the two structures are elementary equivalent,
            hence that
            $\NeighbT{A}{\vec{a}}{0}{q} = \NeighbT{B}{\vec{b}}{0}{q}$
            for all $q \geq 0$.
            We have proven that $A \locleq{0}{q}{k}$ for all $q,k \in
            \mathbb{N}$, hence also for $q,k \in \mathbb{N} \cup \{ \infty \}$.
            \qedhere
\end{enumerate}
\end{proof}

\section{Appendix to ``Positive Gaifman Normal Form'' (Section \ref{sec:pgaif})}\label{app:sec:pgaif}

    \begin{ifappendix}
        \extendedcover*
    \end{ifappendix}
\begin{proof}
    We proceed by induction over $k$.
    \begin{itemize}
        \item When $k \leq 1$ it suffices to take
            $R \defined r$ and for every vector $\vec{a} \in A^{\leq k}$
            build
            $\vec{b} \defined \vec{a}$
            and notice that
            $\Neighb{A}{\vec{a}}{r}
            = \Neighb{A}{\vec{b}}{R}$.

        \item When $k \geq 2$, we proceed by a simple case analysis
            \begin{enumerate}
                \item Either the balls $\Neighb{A}{a}{3r}$ are pairwise disjoint
                    when $a$ ranges over $\vec{a}$; 
                    in which case it suffices to consider $\vec{b} \defined
                    \vec{a}$ and $R \defined r$ to conclude.
                \item Or at least two of the balls $\Neighb{A}{a}{3r}$
                    intersect 
                    when $a$ ranges over $\vec{a}$,
                    and we can assume without loss of generality
                    that the neighborhoods of $a_1$ and $a_2$ at
                    radius $3r$ intersect.

                    Let us consider $c \in \Neighb{A}{a_1}{3r} \cap
                    \Neighb{A}{a_2}{3r}$.
                    Define $\vec{c} \defined (c, a_3, \ldots, a_{k})$,
                    this vector is of size $1$ when
                    $k = 2$.

                    Because $d(a_1, c) \leq 3r$ and $d(a_2, c) \leq 3r$,
                    $\Neighb{A}{a_1 a_2}{r} \subseteq
                    \Neighb{A}{c}{4r}$.

                    By induction hypothesis,
                    there exists a radius $4r \leq R \leq 4^{k-1} (4r)$
                    and a vector $\vec{b} \in A^{\leq k - 1}$
                    such that 
                    $\Neighb{A}{\vec{c}}{4r} \subseteq 
                    \Neighb{A}{\vec{b}}{R}$
                    and
                    $\forall b\neq b' \in \vec{b},
                    \Neighb{A}{b}{3R} \cap \Neighb{A}{b'}{3R} = \emptyset$.

                    Since
                    $\Neighb{A}{\vec{a}}{r} \subseteq 
                    \Neighb{A}{\vec{c}}{4r}$
                    and $r \leq 4r \leq R \leq 4^k r$
                    we concluded.
                    \qedhere
            \end{enumerate}
    \end{itemize}
\end{proof}
\subsection{Appendix to ``From Existential Local Sentences to Asymmetric Basic Local Sentences'' (Section \ref{sec:pgaif:elstoabls})}\label{app:sec:pgaif:elstoabls}

    \begin{ifappendix}
        \asblvsabl*
    \end{ifappendix}
\begin{proof}
    Let $\varphi \defined \exists \vec{x}.
    \bigwedge_{i \neq j} d(x_i, x_j) > 2r \wedge 
    \psi(\vec{x})$
    where $\psi$ is a $r$-local formula around $\vec{x}$.
    Following a syntactical proof of Feferman-Vaught
    given by Thomas Colcombet,
    we introduce \emph{types} $T_1, \dots, T_{|\vec{x}|}$
    and define
    typed formulas inductively as follows
    \begin{align*}
        \tau :=\, &
            x:T \mid
            R(x_1 : T, \dots, x_n : T) \quad \text{ when } R \in \upsigma \\
                &\tau \wedge \tau \mid
            \tau \vee \tau \mid
            \neg \tau \mid \\
                &\exists x:T. \tau \mid
            \forall x:T. \tau
    \end{align*}
    In this new language, a relation can only be checked using variables
    of the same type $T$.
    The evaluation of a formula $\tau$ is defined as the
    evaluation of the formula $\tau'$ obtained by removing type annotations.

    An \emph{environment} $\rho$ is a mapping from finitely many variables to
    types. Given an environment $\rho$, we write $\rho[x \mapsto T]$
    to denote the environment $\rho'$ obtained by $\rho'(y) \defined \rho(y)$
    if $x \neq y$ and $\rho'(x) = T$.
    We first translate $\psi$ into a typed formula inductively
    given an environment $\rho$ from the free variables of $\psi$
    to types:
    \begin{itemize}
        \item $f(\exists x. \psi, \rho) = \bigvee_{1 \leq i \leq |\vec{x}|}
            \exists x:T_i. f(\psi, \rho[x \mapsto T_i])$,
        \item $f(\psi_1 \vee \psi_2, \rho) = f(\psi_1, \rho) \vee f(\psi_2,
            \rho)$,
        \item $f(\neg \psi, \rho) = \neg f(\psi, \rho)$,
        \item $f(R(y_1, \dots, y_n), \rho) = R(y_1 : \rho(y_1), \dots, y_n :
            \rho(y_n))$ if all the variables have the same type under $\rho$,
            and
            $f(R(y_1, \dots, y_n), \rho) = \bot$ otherwise.
    \end{itemize}

    By a straightforward induction on the formula $\psi$,
    one can show that for every structure $A$,
    every valuation $\nu$ from variables to elements of $A$,
    and environment $\rho$ mapping free variables of $\psi$
    to types, such that $d_A(\nu(x),\nu(y)) \leq 1$
    implies $\rho(x) = \rho(y)$
    the following holds:
    $A, \nu \models \psi(\vec{x})$ if and only
    if $A, \nu \models f(\psi(\vec{x}), \rho)$.
    The only non trivial case is the translation of relations,
    which is handled of through the hypothesis on $\nu$ and $\rho$.

    As a consequence, for every structure $A$,
    $A \models \exists \vec{x}.
    \bigwedge_{i \neq j} d(x_i, x_j) > 2r \wedge 
    \psi(\vec{x})$
    if and only if 
    $A \models \exists \vec{x}.
    \bigwedge_{i \neq j} d(x_i, x_j) > 2r \wedge 
    f(\psi(\vec{x}), x_i \mapsto T_i)$.

    A second induction allows us to prove that
    typed formulas are equivalent to positive Boolean combinations
    of \emph{monotyped} formulas, i.e. formulas where only one type $T$
    appears.
    The only non trivial case is the existential quantification,
    handled through the equivalence between
    $\exists x : T. (\psi_1 \wedge \psi_2)$
    and $(\exists x : T. \psi_1) \wedge \psi_2$
    whenever $\psi_2$ contains no variable of type $T$.

    Since the environment $\rho \colon x_i \mapsto T_i$
    assigns a different type to every free variable of $\psi(\vec{x})$,
    the positive Boolean combination obtained by transforming $f(\psi(\vec{x}), \rho)$
    is composed of formulas with exactly one free variable. Let us write
    $\bigvee_{n} \bigwedge_{1 \leq m \leq k} \tau_{n,m} (x_m : T_m)$
    for this positive Boolean combination.
    For every structure $A$,
    $A \models \exists \vec{x}.
    \bigwedge_{i \neq j} d(x_i, x_j) > 2r \wedge 
    \psi(\vec{x})$
    if and only if
    $A \models \exists \vec{x}.
    \bigwedge_{i \neq j} d(x_i, x_j) > 2r \wedge 
    \bigvee_{n} \bigwedge_{1 \leq m \leq k} \tau_{n,m}(x_m : T_m)
    $.
    By removing type annotations on the formulas $\tau_{n,m}$
    and considering their relativisation to the $r$-neighborhood
    of their single free variable, one obtains $\tau_{n,m}^r(x)$.
    Let us prove that this relativisation preserves the expected equivalence
    with $\varphi$:
    \begin{align*}
        A \models \varphi &\iff
        \exists \vec{a} \in A^k, \bigwedge_{1 \leq i \neq j \leq k} d_A(a_i,a_j)
        > 2r \wedge A, \vec{a} \models \psi(\vec{x}) & \text{definition of
        $\varphi$}\\
                          &\iff
        \exists \vec{a} \in A^k, \bigwedge_{1 \leq i \neq j \leq k} d_A(a_i,a_j)
        > 2r \wedge \Neighb{A}{\vec{a}}{r}, \vec{a} \models \psi(\vec{x})
                          &\text{$\psi$ is $r$-local}
        \\
                          &\iff
        \exists \vec{a} \in A^k, \bigwedge_{1 \leq i \neq j \leq k} d_A(a_i,a_j)
        > 2r \wedge \Neighb{A}{\vec{a}}{r}, \vec{a} \models 
        \bigvee_{n} \bigwedge_{1 \leq m \leq k} \tau_{n,m}(x_m : T_m)
                          &\text{by definition}
        \\
                          &\iff
        \exists n. \exists \vec{a} \in A^k, \bigwedge_{1 \leq i \neq j \leq k} d_A(a_i,a_j)
        > 2r \wedge \Neighb{A}{\vec{a}}{r}, \vec{a} \models 
        \bigwedge_{1 \leq m \leq k} \tau_{n,m}(x_m : T_m)
                          &\text{$\vee$}
        \\
                          &\iff
        \exists n. \exists \vec{a} \in A^k, \bigwedge_{1 \leq i \neq j \leq k} d_A(a_i,a_j)
        > 2r \wedge A, \vec{a} \models 
        \bigwedge_{1 \leq m \leq k} \tau_{n,m}^r(x_m)
                          &\text{relativisation}
        \\
                          &\iff
        \exists n. 
         A \models 
        \exists \vec{x}. \bigwedge_{1 \leq i \neq j \leq k} d(x_i,x_j)
        > 2r \wedge
        \bigwedge_{1 \leq m \leq k} \tau_{n,m}^r(x_m)
                          &\text{$\exists$ and $\wedge$}
        \\
    \end{align*}

    We have expressed $\varphi$ as a disjunction of asymmetric basic local
    sentences
    \begin{equation*}
        \varphi \equiv
        \bigvee_{n}
        \exists x_1, \dots, x_k.
        \bigwedge_{1 \leq i \neq j \leq k} d(x_i, x_j) > 2r \wedge 
        \bigwedge_{1 \leq m \leq k} \tau_{n,m}^{r}(x_m)
        \, . \qedhere
    \end{equation*}
\end{proof}
\subsection{Appendix to ``From Asymmetric Basic Local to Basic Local Sentences'' (Section \ref{sec:pgaif:ablstobls})}\label{app:sec:pgaif:ablstobls}
    \begin{ifappendix}%
        Let us recall the definitions in~\cref{sec:pgaif:ablstobls}
        before proving~\cref{fact:loc:securitycylinders}.
        The formula $\theta_G^R(x)$ is defined as follows
\begin{align*}
    \theta_G^R(x) \defined
    &\exists v_1, \ldots, v_{|V(G)|}
    \in \Neighb{}{x}{R}. \\
    &\bigwedge_{(v_i,v_j,h) \in E(G)}
    d(v_i,v_j) = h \\
    &\wedge \bigwedge_{ v_i \in V(G)}
    \bigwedge_{p \in C(i)} p(v_i)
    \\
    &\wedge \bigwedge_{v_i \in V(G)}
        \Neighb{}{v_i}{r} \subseteq \Neighb{}{x}{R} \, .
\end{align*}


We restate hereafter~\cref{fact:loc:abstriaction}.

\graphrepresentation*
    \end{ifappendix}%

    \begin{ifappendix}
        \securitycylinder*
    \end{ifappendix}
\begin{proof}
    Since $A,a \models \theta_G^R(x)$
    and using \cref{fact:loc:abstriaction}
    there exists a vector $\vec{c} \in \Neighb{A}{a}{R}$
    such that $G_{\vec{c}}^R = G$ 
    and
    $\Neighb{A}{\vec{c}}{r} \subseteq \Neighb{A}{a}{R}$.
    In particular, every point of $\vec{c}$ satisfies
    at least one property $p \in Q$.
    As $d(a,b) \leq 2R$,
    $\vec{c} \in \Neighb{A}{b}{3R}$.
    As $A, b \models \pi_Q^R(x)$
    his shows that $\Neighb{A}{\vec{c}}{r} \subseteq \Neighb{A}{b}{R}$.
    Using~\cref{fact:loc:abstriaction} this proves 
    that $A, b \models \theta_G^R(x)$.
\end{proof}

\section{Appendix to ``Failure in the finite case'' (Section \ref{sec:ffc})}\label{app:sec:ffc}

\subsection{Appendix to ``A Generic Counter Example'' (Section \ref{sec:ff:generic})}\label{app:sec:ff:generic}

    \begin{ifappendix}
        \badphicord*
    \end{ifappendix}
\begin{proof}
    The sentence $\badphi \defined \forall x. B(x) \vee \varphi_{CC}$
    is well defined and preserved under disjoint unions
    thanks to \cref{lem:ff:phiccbads}.

    Assume by contradiction that $\badphi$
    is equivalent over $\BadOrderS$ to a
    sentence $\psi = \exists x_1, \dots, x_k. \theta(\vec{x})$
    where $\theta$ is a $r$-local sentence.
    Construct the structure
    $A \defined O_1 + O_2 + \cdots + O_{2 k \cdot (2r+1) + 1}$
    in $\BadOrderS$
    and colour all of the nodes with $\neg B$.
    This structure satisfies $\badphi$
    hence $\psi$. However,
    there exists at least one node of $A$
    not covered by the $r$-balls around the $k$
    witnesses of $\psi$, the structure $\hat{A}$
    obtained by colouring one of those node
    with $B$
    still satisfies $\psi$ but
    does not satisfy $\badphi$.
\end{proof}

\begin{ifappendix}
    \badorderaxiomatic*
\end{ifappendix}
\begin{proof}
    Before listing the axiomatic $\BadAxiom$, notice that
    it suffices to prove that connected structures in $\BadOrderS$
    models $\BadAxiom$ as they are closed under disjoint unions
    of models.

    Let $A$ be a connected structure in $\BadOrderS$,
    that is $A = O_n + \cdots + O_m$ with $2 \leq n \leq m$.
\begin{enumerate}[(i)]
\item $\leq$ is transitive, reflexive and antisymmetric
    over each $O_i$. As a consequence, $\leq$ is transitive 
    reflexive and antisymmetric over their disjoint unions.
\item $\leq$ has connected components of size at least two.
    This holds because $2 \leq n \leq m$.
\item $S$ is an injective partial function without fixed points.
    This holds because $S$ is the successor relation on $O_i$
    and the maximal element of $O_i$ is connected through $S$
    to the minimal element of $O_{i+1}$.
\item $S$ and $\leq$ cannot
    conflict:
    $\forall a,b. \neg (S(a,b) \wedge b \leq a)$.
    This holds because $S$ is the successor relation over $O_i$.
\item There exists a proto-induction principle:
    $\forall a < b, \exists c, S(a,c) \wedge a \leq c \wedge c \leq b$.
    This holds because $S$ is the successor relation over $O_i$.
\item Edges $E$ can be factorised trough
    $(\leq) (S \setminus {\leq}) (\leq)$:
    $\forall a,b.
    E(a,b) \implies \exists c_1, c_2. 
    a \leq c_1 \wedge c_1 S c_2 \wedge \neg (c_1 \leq c_2) \wedge c_2 \leq b$.
    This holds because edges $E$ between $O_i$ and $O_j$ only 
    appear if $j = i+1$ and the maximal element of $O_i$ for $\leq$
    is connected through $S$ to the minimal element of $O_j$ for $\leq$.
\item Pre-images through $E$ form a suffix of the ordering:
    $\forall a,b,c.
    a \leq b \wedge E(a,c) \wedge S(a,b) \implies E(b,c)$.
    This holds because edges $E(a,b)$ between $O_i$ and $O_j$ 
    exists if and only if $a \leq b$ when considered as integers.
\item Images through $E$ form a prefix of the ordering:
    $\forall a,b,c.
    b \leq c \wedge E(a,c) \wedge S(b,c) \implies E(a,b)
    $.
    This holds because edges $E(a,b)$ between $O_i$ and $O_j$ 
    exists if and only if $a \leq b$ when considered as integers.
\item Images through $E$ are strictly increasing subsets:
    $\forall a,b.
    a S b \wedge a \leq b \implies
    \exists c. E(b,c) \wedge \neg E(a,b)$.
    This holds because edges $E(a,b)$ between $O_i$ and $O_j$ 
    exists if and only if $a \leq b$ when considered as integers.
\item Pre-images through $E$ are strictly decreasing subsets:
    $\forall a,b.
    a S b \wedge a \leq b \wedge
    (\exists c. E(c,a))
    \implies
    \exists c. E(c,a) \wedge \neg E(c,b)$.
    This holds because edges $E(a,b)$ between $O_i$ and $O_j$ 
    exists if and only if $a \leq b$ when considered as integers.
\item The last element of an order cannot be obtained
    through $E$: 
    $\forall a,b.
    E(a,b) \implies \exists c. S(b,c) \wedge b \leq c$.
    This holds because edges $E(a,b)$ between $O_i$ and $O_j$ 
    exists if and only if $a \leq b$ when considered as integers.
\item The relation $(\leq)(S \setminus {\leq})$ is included in $E$:
    $\forall a, b.
    \exists c. a \leq c \wedge S(c,b) \wedge \neg (c \leq b)
    \implies
    E(a,b)
    $.
    This holds because edges $E(a,b)$ between $O_i$ and $O_j$ 
    exists if and only if $a \leq b$ when considered as integers
    and $O_i$ is of size $i$.
    \qedhere
\end{enumerate}
\end{proof}

    \begin{ifappendix}
        \badorderaxiomone*
    \end{ifappendix}
\begin{proof}
Let $A$ be a structure in $\BadOrder$,
without loss of generality, assume that the Gaifman graph of $A$
has a single connected component.

\begin{enumerate}[(a)]
    \item
        Notice that the property $(v)$ combined with
        the fact that $S$ is a partial injective function 
        with no fixed points
        and that $A$ is finite
        allows to prove
        by induction that
        whenever $a \leq b$ there exists $0 \leq k$
        such that $S^k(a) = b$
        and all the intermediate points are between $a$ and $b$
        for $\leq$.

        As a consequence, if $a \leq b_1$ and $a \leq b_2$,
        there exists $0 \leq k$ and $0 \leq l$
        such that $S^k(a) = b_1$ and $S^l(a) = b_2$.
        Without loss of generality $k \leq l$ 
        and $b_2 = S^{k-l}(b_1)$ and 
        $b_1 \leq b_2$.
        Similarly, if $a_1 \leq b$ and $a_2 \leq b$
        there exists $0 \leq k,l$ such that
        $S^{k}(a_1) = b$ and $S^{l}(a_2) = b$

        Combined with $(i)$ and $(iii)$
        this proves that connected components of $A$
        through the relation $\leq$ are \emph{totally ordered}
        by $\leq$ and over these components $S$ is the successor
        relation, while $(ii)$ states that this connected component must have at
        least two elements.

        Assume by contradiction that
        some $\leq$-connected component of $A$
        contains a relation $E(a,b)$
        notice that property $(vii)$ provides $a \leq c_1 S c_2 \leq b$
        and $\neg (c_1 \leq c_1)$ but this is absurd 
        since the connected component is a total ordering
        and property $(v)$ states that we cannot have $c_1 S c_2$ and
        $c_2 \leq c_1$.

\item Assume $B_1$ and $B_2$ are two $\leq$-components of $A$
    connected through the relation $S$.
    As $S$ is a partial injective function
    shows that this can only happen by
    connecting an element that has no successor to an element
    that has no predecessor.
    As a consequence, it is only possible to connect the last element of $B_1$
    to the first one of $B_2$ and vice-versa.

\item Assume that $B_1$ and $B_2$ are two $\leq$-components of $A$
    connected through the relation $E$, that is there exists 
    $a \in B_1$ and $b \in B_2$ such that $E(a,b)$ holds.
    The property $(vi)$ provides
    $c_1$ and $c_2$ such that
    $a \leq c_1$, $c_1 S c_2$, $c_2 \leq b$ and $\neg (c_1 \leq c_2)$.
    As a consequence $c_1$ is in $B_1$ (connected through $\leq$),
    $c_2 \in B_2$ (connected through $\leq$)
    and therefore $B_1$ is connected to $B_2$ through the relation $S$.

\item Assume $B_1$ and $B_2$ are two $\leq$-components of $A$
    connected through the relation $S$.
    The function $g \colon a \mapsto \setof{b}{E(a,b)}$
    from $B_1$ to $\wp(B_2)$ is well-defined
    since we proved that there cannot be edges $E$ outside of $B_2$.
    Since $B_2$ is a finite total ordering with respect to $\leq$
    and the image of $g$ is non-empty thanks to $(xii)$
    the function $f$ is well-defined.

    Property $(ix)$ states that whenever $a,b \in B_1$ and $a S b$ then
    $g(a) \subsetneq g(b)$, while property $(viii)$ states
    that $g(a)$ is $\leq$-downwards-closed in $B_2$.
    As a consequence, $f$ must be strictly increasing.

    Similarly, property $(vii)$ combined with property
    $(x)$ states that if $a S b$ in $B_1$,
    then $|g(a)|+1 = |g(b)|$
    and as a consequence $S(f(a)) = f(b)$.

    Finally, property $(xi)$ states that
    $g$ never covers the last element of $B_2$, and in particular
    $f$ is not surjective.

\item Notice that we proved edges $E$ can only appear
    between two $\leq$-components $B_1$ and $B_2$.
    Let us write $0_1$ and $0_2$ their respective $\leq$-minimal elements.
    We showed that
    $f(S^k(0_1)) = S^k(0_2)$ whenever
    $S^k(0_1)$ is in $B_1$.
    Moreover, if $a = S^k(0_1) \in B_1$
    then
    $\setof{d}{E(a,d)} = \downarrow_{\leq} f(a)
    = \setof{ S^l(0_2) }{l \leq k}$.
    Moreover, property $(x)$
    shows that there cannot exist more than two
    points in $B_2$ that are not obtainable through $f$,
    and as a consequence $|B_2| = |B_1| + 1$.

    Finally, a general connected component
    of $A$
    is of the form $O_n + \cdots + O_m$
    with $2 \leq n \leq m$.
    \qedhere
\end{enumerate}
\end{proof}
\subsection{Appendix to ``Undecidability'' (Section \ref{sec:ff:undecidable})}\label{app:sec:ff:undecidable}

\begin{ifappendix}
    \definabletransitions*
\end{ifappendix}
\begin{proof}
    This follows the standard encoding of transitions.
    We check that the sentence is $1$-local since
    the ball of radius $1$ around $a$ (or $b$) contains
    both $\leq$-components entirely.

    The only technical issue is relating
    positions in the configuration $C(a)$ to positions
    in the configuration $C(b)$, which 
    is done through the use of
    $f \colon a \mapsto \max_{\leq} \setof{d}{E(a,d)}$
    which is first-order definable and $1$-local.
    Let us spell out the definition of $f$ as a formula:
    \begin{equation*}
        \phi_f(x,y) \defined E(x,y) \wedge \forall z. y < z \implies \neg E(x,z)
    \end{equation*}

    For instance, to assert
    that every letter except those near the current position
    in the tape are left unchanged, one can first write a formula
    stating that the position of the head of the Turing Machine
    in the $\leq$-component of $x$ is not close to $x$:
    \begin{equation*}
        \phi_Q(x) \defined \forall z. (z \leq x \vee x \leq z) \wedge 
        (S(x,z) \vee S(z,x) \vee x = z) \implies \bigwedge_{q \in Q} \neg q(x)
    \end{equation*}

    As a shorthand, let us write $z \in C(x)$ instead of $z \leq x \vee x \leq
    z$.
    Using $\phi_f$ and $\phi_Q$ it is easy to write a formula stating
    that letters far from the head of the Turing Machine are unchanged
    in a transition:
    \begin{equation*}
        \phi(x,y) \defined
        \forall z. z \in C(x) \wedge \phi_Q(z) \implies \exists z'. z' \in C(y)
        \wedge \phi_f(z,z')
        \wedge
        \bigwedge_{a \in \Sigma \cup \{ \$, \square \}} P_a(z) \Leftrightarrow P_a(z')
    \end{equation*}

    We can assert that some specific transition has been
    taken in a similar fashion, by first checking the movement of the head,
    change of letters around the head, and the evolution of the state.
\end{proof}
    \begin{ifappendix}%
        \undecidabledisunion*%
    \end{ifappendix}%
\begin{proof}
    Without loss of generality, we only work over
    $\BadOrderS$ and
    we reduce from the halting problem.
    Consider the sentence
    $\varphi_M \defined \exists x. \theta_I^{\langle M \rangle}(x) 
                  \wedge 
    \exists x. \theta_F(x) 
                  \wedge
    \forall x,y. S(x,y) \wedge \neg (x \leq y) \Rightarrow
    \theta_T(x,y)$,
    and let $\varphi = \varphi_M \wedge \neg \varphi_{CC}$.

    Assume that $M$ halts. There
    exists exactly one model of $\varphi$, the unique
    run of the universal Turing machine $U$,
    and this contradicts closure under disjoint unions.

    Assume that $M$ does not halt.
    The sentence $\varphi$ has no finite model and
    in particular is closed under disjoint unions.

    Hence, the sentence $\varphi$ is
    closed under disjoint unions
    if and only if $M$ does not halt.
\end{proof}
\subsection{Appendix to ``Generalisation to weaker preorders'' (Section \ref{sec:ff:generalise})}\label{app:sec:ff:generalise}
    \begin{ifappendix}%
        \arbitraryquantifierrank*%
    \end{ifappendix}%
\begin{proof}
    Using \cref{fact:qo:refinement} it suffices
    to consider the case $r = 1$ and $k = 2$.
    The proof follows the same pattern as \cref{lem:failure:counterexgen},
    we enumerate the axioms from $\BadAxiom$
    and notice that they are of the form $\forall x. \theta(x)$
    where $\theta(x)$ is a $1$-local formula.
    As a consequence,
    sentences preserved under $\locleq{1}{\infty}{2}$
    over $\BadOrderS$ are preserved under $\locleq{1}{\infty}{2}$
    over $\ModF(\upsigma)$, which is a strengthening of
    \cref{fact:ff:relat}.

    Moreover, $\varphi_{CC}$ preserved under $\locleq{1}{\infty}{2}$
    using a simple syntactical analysis.
    We now check that $\badphi$ is preserved under $\locleq{1}{\infty}{2}$
    over $\BadOrderS$.
    Let $A, B \in \BadOrderS$ such that $A \models \badphi$
    and $A \locleq{1}{\infty}{1} B$.

    \begin{itemize}
        \item Let us first examine the case
     where $A$ has a single connected component.

    Let $a \in A$; since $\Neighb{A}{a}{1}$ is finite,
    there exists a $1$-local formula $\psi_a(x)$
    of quantifier rank less than $|\Neighb{A}{a}{1}|+1$
    such that $B, b \models \psi_a(x)$ if and only 
    if $\Neighb{A}{a}{1}$ is isomorphic to $\Neighb{B}{b}{1}$.
    In particular, if $C$ is a 
    $\leq$-component of $A$, it is of radius less than $1$
    and there exits $C'$ a $\leq$-component of $B$
    isomorphic to $C$.
    Moreover, the $1$-neighborhood of a $\leq$-component
    contains the previous and next $\leq$-component for $S$.

    If $B$ has a single connected component, then
    two distinct $\leq$-components in $B$ must have distinct sizes.
    Using the fact that the components of $A$ are all found in $B$
    and that their relative position is preserved,
    this proves that $B$ contains exactly the same $\leq$-components as $A$,
    which
    only happens if $A=B$.

    If $B$ has at least two connected components,
    then it satisfies $\varphi_{CC}$ which implies $\badphi$.
        \item In the case where $A$ has two connected components,
            $A \models \varphi_{CC}$ but then
            $B \models \varphi_{CC}$ and $B \models \badphi$.
    \end{itemize}

    Moreover, we know from \cref{lem:ff:counterbads}
    that $\badphi$ cannot be defined as an existential local sentence
    over $\BadOrderS$.
\end{proof}%

\begin{ifappendix}
Before studying the case $k = 1$, let us
descibe the behaviour of $\locleq{r}{q}{k}$
with respect to disjoint unions. In particular,
we prove that at a fixed $k$, the preorder cannot
distinguish between more than $k$ copies of
the same structure.
\end{ifappendix}
\begin{example}[Disjoint unions]
    \label{ex:loc:disjointunions}
    Let $A, B$ be finite structures.
    For all $0 \leq r,q  \leq \infty$
    and $1 \leq k < \infty$,
    $A \locleq{r}{q}{k} A \uplus B$
    and $\biguplus_{i = 1}^{k + n} A \locleq{r}{q}{k} \biguplus_{i = 1}^{k} A$.
\end{example}
\begin{proof}
    Let us prove that
    $A \locleq{r}{q}{k} A \uplus B$.
    Consider a vector $\vec{a} \in A^k$,
    it is clear that this vector appears as-is in $A \uplus B$
    and since the union is disjoint,
    $\Neighb{A \uplus B}{\vec{a}}{r} = \Neighb{A}{\vec{a}}{r}$.
    In particular, $\NeighbT{A}{\vec{a}}{r}{q} = \NeighbT{A \uplus
    B}{\vec{a}}{r}{q}$ for all $q \geq 0$.

    Let us write $A^{\uplus k} \defined \biguplus_{i = 1}^{k} A$
    and $A^{\uplus k+n} \defined \biguplus_{i = 1}^{k+n} A$.
    Remark that the previous statement shows
    $A^{\uplus k} \locleq{r}{q}{k} A^{\uplus k+n}$.
    Let us prove now $A^{\uplus k+n} \locleq{r}{q}{k} A^{\uplus k}$.
    Consider a vector $\vec{a} \in (A^{\uplus k+n})^k$,
    this vector has elements in at most $k$ copies of $A$,
    hence one can select one copy of $A$ in $A^{\uplus k}$ for each
    of those and consider the exact same elements in those copies.
    As the unions are disjoint, the obtained neighborhoods are
    isomorphic to those in $A^{\uplus k+n}$ and in particular share
    the same local types.
\end{proof}
\begin{ifappendix}%
    \successatone*%
\end{ifappendix}%
\begin{proof}
    Without loss of generality thanks to \cref{fact:qo:refinement}
    we consider a sentence $\varphi$ preserved under
    $\locleq{\infty}{\infty}{1}$.
    Let us prove that $\varphi$ is preserved under $\locleq{r}{q}{k}$
    where all parameters are finite.
    This property combined with \cref{lem:ccl:existlocalnf} will prove that
    $\varphi$ is equivalent to an existential local sentence.

    Let us write $\exists^{\geq k}_r x. \psi(x)$
    as a shorthand for
    $\exists x_1, \dots, x_k. \bigwedge_{1 \leq i \neq j \leq k} d(x_i, x_j) >
    2r \wedge \bigwedge_{1 \leq i \leq k} \psi(x_i)$.
    Thanks to Gaifman's Locality Theorem, we can
    assume that $\varphi$ is a Boolean combination of the following
    basic local sentences for $1 \leq i \leq n$: $\theta_i = \exists^{\geq k_i}_{r_i} x. \psi_i(x)$
    where $\psi_i(x)$ is a $r_i$-local sentence of quantifier rank $q_i$.
    Define $r \defined \max \setof{r_i}{1 \leq i \leq n}$,
    $q \defined \max \setof{q_i}{1 \leq i \leq n}$
    and $k \defined \max \setof{k_i}{1 \leq i \leq n}$.

    Let $A,B$ be two finite structures
    such that $A \models \varphi$ and
    $A \locleq{r}{q}{k} B$. Our goal is to prove that $B \models \varphi$.
    Let us write $B^{\uplus k}$ the disjoint union of $k$ copies of $B$.

    Let us show that $A \uplus B^{\uplus k} \models \varphi$ if and
    only if $B^{\uplus k} \models
    \varphi$. To that end, let us fix $1 \leq i \leq n$ and prove
    that $A \uplus B^{\uplus k} \models \exists^{\geq k_i}_{r_i} x. \psi_i(x)$
    if and only if $B^{\uplus k} \models \exists^{\geq k_i}_{r_i} x. \psi_i(x)$.

    \begin{itemize}

        \item Assume that $A \uplus B^{\uplus k} \models \exists^{\geq k_i}_{r_i} x.
            \psi_i(x)$,
    there exists a vector $\vec{c}$ of witnesses
    of $\psi$
    at pairwise distance
    greater than $2r_i$ in $A \uplus B^{\uplus k}$. 
    If $\vec{c}$ lies in $B^{\uplus k}$ then we conclude, otherwise
    some element of $\vec{c}$ lies in $A$.
    In particular, $A \models \exists x. \psi(x)$. Since $A \locleq{r}{q}{k} B$,
    we know that
    $B \models \exists x. \psi_i(x)$ and
    as $k_i \leq k$, $B^{\uplus k} \models \exists^{\geq k_i}_{r_i} x. \psi_i(x)$.
\item Conversely, whenever $B^{\uplus k} \models \exists^{\geq k_i}_{r_i} x. \psi_i(x)$
    the structure $A \uplus B^{\uplus k}$ satisfies $\exists^{\geq k_i}_{r_i} x. \psi(x)$
    as basic local sentences are preserved under disjoint unions.
    \end{itemize}

    Since $A \locleq{\infty}{\infty}{1} A \uplus B \locleq{\infty}{\infty}{1} A \uplus B^{\uplus k}$
    and $A \models \varphi$, we know that $A \uplus B^{\uplus k} \models \varphi$.
    This implies that $B^{\uplus k} \models \varphi$, and
    we deduce from \cref{ex:loc:disjointunions}
    that $B^{\uplus k}$ is $\locleq{\infty}{\infty}{1}$-equivalent to
    $B$. In particular $B^{\uplus k} \models \varphi$ implies $B \models \varphi$.
\end{proof}
\section{Appendix to ``Localising \ltrsk'' (Section \ref{sec:presthm})}\label{app:sec:presthm}

\subsection{Appendix to ``Localisable Classes'' (Section \ref{sec:presthm:inducedsub})}\label{app:sec:presthm:inducedsub}

We provide here a self-contained proof
an adaptation of \citet[Theorem 4.3]{atserias2008preservation}
tailored to the study of local neighborhoods $\Balls{\XClass}{r}{k}$.

\inducedminimalmodels*

To prove \cref{lem:failure:inducedminimalmodels}, a first attempt
is to
consider a formula $\varphi$ in Gaifman normal form,
and a minimal model $A$ of $\varphi$ with respect to $\subseteq_i$.
The structure $A$ models a conjunction of (potential negations of)
sentences
of the form $\exists^{\geq k}_r x. \psi(x)$,
which we use as a shorthand
for
$\exists x_1, \dots, x_k. \bigwedge_{1 \leq i \neq j \leq k} d(x_i, x_j) >
2r \wedge \bigwedge_{1 \leq i \leq k} \psi(x_i)$.
Considering all basic local sentences that appear positively in the conjunction,
one can build a vector $\vec{a} \in A$ containing the witnesses of the
existential quantifications.
The main hope being that $\Neighb{A}{\vec{a}}{r} \models \varphi$, as
this would imply, by minimality of $A$, that $A = \Neighb{A}{\vec{a}}{r}$.
By letting $R = r$ and $K = |\vec{a}|$, which is bounded
independently of $A$, we would conclude.

Unfortunately the structure $\Neighb{A}{\vec{a}}{r}$ does not
satisfy $\varphi$ in general.
The crucial issue comes
from \emph{intersections}
of neighborhoods: there are new neighborhoods appearing in
$\Neighb{A}{\vec{a}}{r}$ as the intersection of the neighborhood 
of $\vec{a}$ with the neighborhood of a point inside $\Neighb{A}{\vec{a}}{r}$.
This is problematic because $\varphi$ contains basic local sentences that
appear negatively, and the fact that $A$ does not contain some local behaviour
does not transport to $\Neighb{A}{\vec{a}}{r}$.

To tackle this issue, we temporarily leave the realm of first-order
logic and consider $\MSO$ local types,
written $\NeighbM{A}{\vec{a}}{r}{q}$. There are
finitely many $\MSO$ local types up to
logical equivalence at a
given quantifier rank and locality radius.
We update our \emph{type-collector} function accordingly through
\begin{equation}
    \SpecterM{r}{q}{k}(A)
\defined \setof{
\NeighbM{A}{\vec{a}}{r}{q}
}{\vec{a} \in A^k} \, .
\end{equation}

As for first-order local types, $\MSO$ local types
are enough to characterise the natural preorder associated
to \emph{existential local $\MSO$-sentences}, that is, sentences
of the form $\exists \vec{x}. \theta(\vec{x})$, where
$\theta(\vec{x})$ is an $\MSO$ formula $r$-local around $\vec{x}$.
Before going through the proof of~\cref{lem:failure:inducedminimalmodels}
we translate the main properties of local types and existential local
sentences to $\MSO$-local types and existential local $\MSO$-sentences.

\begin{fact}
    \label{fact:pres:mso:specter1}
For all structures $A,B$ such that $\SpecterM{r}{q}{k}(B)$
contains $\SpecterM{r}{q}{k}(A)$,
for all $r$-local $\MSO$ formula $\theta(\vec{x})$ of quantifier rank
less than $q$,
$A \models \exists \vec{x}. \theta(x)$ implies $B \models \exists \vec{x}.
\theta(x)$.
\end{fact}

\begin{fact}
    \label{fact:pres:mso:specter2}
Let $A,B$ be structures such that 
for all $r$-local $\MSO$ formula $\theta(\vec{x})$ of quantifier rank
less than $q$,
$A \models \exists \vec{x}. \theta(x)$ implies $B \models \exists \vec{x}.
\theta(x)$. Then $\SpecterM{r}{q}{k}(B)$
contains $\SpecterM{r}{q}{k}(A)$.
\end{fact}

\begin{lemma}[Local normal form adapted to $\MSO$]
    \label{lem:pres:mso:existlocalnf}
    Let $\XClass$ be a class of finite structures.
    Let $\varphi$ be a sentence preserved over $\XClass$ under $\MSO$-local types with
    finite parameters $(r,k,q)$.
    Then $\varphi$ is equivalent over $\XClass$
    to an existential local $\MSO$ sentence $\psi = 
    \exists \vec{x}. \theta(\vec{x})$
    where $\theta(\vec{x})$ is an $r$-local $\MSO$-formula.
\end{lemma}
\begin{proof}
    It is the same proof as~\cref{lem:ccl:existlocalnf}
    when replacing $\Specter{r}{q}{k}$ with $\SpecterM{r}{q}{k}$,
    leveraging
    \cref{fact:pres:mso:specter1,fact:pres:mso:specter2}.
\end{proof}

An $\MSO$ $(q,r)$-type $t$ with a single free variable
is $R$-\emph{covered} by a subset $C$
if for all realisations $a$ of $t$ in $A$
the $r$-neighborhood $\Neighb{A}{a}{r}$
is included in
$\Neighb{A}{C}{R}$.
The following lemma is a uniform version
of the one given by \citet*[][Lemma 8]{dawar2006approximation}.
It can be thought as a generalisation of
the technique from~\cref{lem:ccl:extentedcover}
to describe the spatial repartition of points of interest
in a structure.

\begin{lemma}[Type covering]
    \label{lem:pres:mso:typecover}
    For all $r,q,k \geq 0$,
    there exists $K_m$ and $R_m$ such that
    for all structures $A \in \Mod(\upsigma)$
    there exists $r \leq R \leq R_m$,
    a subset $C^A \subseteq A$
    and a subset $G^A \subseteq A$
    satisfying the following properties
    \begin{enumerate}[(i)]
        \item $|C^A| \leq K_m$.
        \item $|G^A| \leq K_m$.
        \item Elements of $G^A$ are at pairwise distance greater than $2R$
            and at distance greater than $2R$ of $C^A$.
        \item For every $a \in A$,
            either $\NeighbM{A}{a}{r}{q}$ is $R$-covered
            by $C^A$
            or there exists $k$ elements $b_1, \dots, b_k$
            in $G^A$
            such that
            $\NeighbM{A}{a}{r}{q} = \NeighbM{A}{b_i}{r}{q}$.
    \end{enumerate}
\end{lemma}
\begin{proof}
    Let $r,q,k \geq 0$ be natural numbers and
    consider $Q$ the number of different $\MSO$-local types
    at radius $r$ and quantifier rank $q$ around $1$ variable.
    Define $K_m = Q \times Q \times k$ and $R_m = 3^{Q+1} r$.

    Let $A$ be a structure and consider
    $T \defined \SpecterM{r}{q}{1}(A)$.
    By definition, $|T| \leq Q$.
    We construct by induction sets $S_i \subseteq T$
    and $C_i \subseteq A$ such that
    every type of $S_i$ is $3^i r$-covered by $C_i$
    and
    $|C_i| \leq K_m$.
    Let $S_0 \defined \emptyset$, $C_0 \defined \emptyset$.
    Assume that $S_i$ and $C_i$
    have been defined and that the following holds:
    \begin{align*}
        \forall G \subseteq A,
        &\, \exists u, v \in G, d_A(u,v) \leq 2 \times 3^i r \\
        &\vee  \exists u \in G, d_A(u,C_i) \leq 2 \times 3^i r \\
        &\vee  \exists t \in T \setminus S_i,
        \neg \left(\exists u_1 \neq u_2 \neq \dots \neq u_k \in G, \bigwedge_{1 \leq i \leq k}
        \NeighbT{A}{u_i}{r}{q} = t\right)
    \end{align*}

    We enumerate types in $T \setminus S_i$ in a sequence
    $(t_p)_{1 \leq p \leq |T \setminus S_i|}$.
    Using this sequence, we construct
    iteratively a set $G_i^j$ of size at most $Q \times k$
    such that points of $G_i^j$ are at pairwise distance greater than
    $2 \times 3^i r$ and at distance greater than $2 \times 3^i r$  from $C_i$
    as follows. Let $G_i^0 \defined \emptyset$,
    and construct $G_i^{j+1}$ by selecting the first type $t_p$
    that has less than $k$ witnesses in $G_i^{j}$, this is possible
    by the assumption made on $S_i$ and $C_i$. 
    If there exists a point $a \in A$ at distance greater than $2 \times 3^i r$
    from $G_i^j$ and $C_i$ and of type $t_p$,
    we can add it to $G_i^j$ to build $G_i^{j+1}$.
    Otherwise, every choice of point $a \in A$ of type $t_p$
    is at distance at most $2 \times 3^i r$ from $C_i \cup G_i^j$,
    in this case let $C_{i+1} = C_i \cup G_i^j$ and $S_{i+1} = S_i \cup \{ t_p
    \}$, the hypothesis on $S_i$ is that
    every type in $S_i$ is $3^{i} r$-covered by $C_i$
    and we showed that $t_p$ was $3^{i+1} r$-covered by $C_i \cup G_i^j$.

    The set $S_i$ is strictly increasing and of size bounded by $Q$,
    as a consequence, there exists a step $i \leq Q$ such that one cannot
    build $S_{i+1}$ and $C_{i+1}$.
    By definition, this means that one can build a subset $G \subseteq A$
    such that 
    \begin{align*}
        &\forall u, v \in G, d_A(u,v) > 2 \times 3^i r \\
        &\wedge  \forall u \in G, d_A(u,C_i) > 2 \times 3^i r \\
        &\wedge  \forall t \in T \setminus S_i,
        \exists u_1 \neq u_2 \neq \dots \neq u_k \in G, \bigwedge_{1 \leq i \leq k}
        \NeighbT{A}{u_i}{r}{q} = t
    \end{align*}

    By extracting only $k$ elements per type in $T \setminus S_i$
    from such a set $G$,
    we can construct that $G^A$ of size at most $Q \times k \leq K_m$.
    Let $C^A \defined C_i$ and $R = 3^i r$. The only thing left to check
    is the size of $C^A$, which is below $Q \times Q \times k$ since
    we do at most $Q$ steps and each step adds at most $Q \times k$ elements
    to $C_i$.
\end{proof}


Following the terminology used
by \citeauthor*{dawar2006approximation}
given a subset $C^A$ obtained through
the preceding lemma, types covered by $C^A$
will be called \emph{rare}
and those obtained in $G^A$ will be called
\emph{frequent}.

\begin{figure}[ht]
    \centering
\def\svgwidth{1\columnwidth}
\import{./figures/}{mso_construction.pdf_tex}

    \caption{Construction of $\Neighb{A}{C^A}{5r} \uplus E'$ using $\MSO$-types.}
    \label{fig:mso_construction}
\end{figure}

\begin{lemma}[Preservation under $\MSO$ types]
    \label{lem:failure:presmso}
    Let $\XClass \subseteq \ModF(\upsigma)$
    that is hereditary and closed under
    disjoint unions
    and $\varphi \in \FO[\upsigma]$
    a sentence
    preserved under $\subseteq_i$ over $\XClass$.
    There exists $R,Q,K$ such that for all $A, B \in \XClass$,
    $A \models \varphi$ and $\SpecterM{R}{Q}{K}(A)
    \subseteq \SpecterM{R}{Q}{K}(B)$
    implies $B \models \varphi$.
\end{lemma}
\begin{proof}
    Consider a sentence $\varphi$
    preserved under $\subseteq_i$ over $\XClass$.
    As a first step, we 
    write $\varphi$ in Gaifman normal form
    and collect $\theta_1, \dots, \theta_l$
    the basic local sentences appearing in this normal form.
    The sentences $\theta_i$ are
    of the form $\exists^{\geq k_i}_{r_i} x. \psi_i(x)$
    where $\psi_i$ is an $r_i$-local formula around $x$
    of quantifier rank $q_i$.
    Recall that we use $\exists^{\geq k_i}_{r_i} x. \psi_i(x)$
    as a shorthand
    for
    $\exists x_1, \dots, x_{k_i}. \bigwedge_{1 \leq p \neq q \leq k_i} d(x_p, x_q) >
    2r_i \wedge \bigwedge_{1 \leq p \leq k} \psi_i(x_p)$.

    Let $r$ be the maximal locality radius of these
    sentences, $q$ their maximal quantifier rank
    and $k$ the maximal number of external existential
    quantifications.
    We use \cref{lem:pres:mso:typecover}
    over the tuple $(2r, 2 \times k \times l, q + 1)$ to obtain
    numbers $2r \leq R_m$ and $k \leq K_m$.
    Define $K = 2 K_m$, $R = 2R_m$ and $Q = 2R_m + k + q + 1 + \max \rk \theta_i$.
    Our goal is to prove that
    $\varphi$ is preserved under $\MSO$-local types with
    parameters $(R,Q,K)$.

    Let us now consider $A \models \varphi$.
    We build
    $C^A$ and $G^A$ as provided
    by \cref{lem:pres:mso:typecover}
    over the tuple $(2r,2k, q + 1)$. Without loss of generality
    assume that the radius given is $R_m$ (the largest possible one).
    The sets $C^A$ and $G^A$ are of size below $K_m$.
    We call $I_f$ the set of indices $1 \leq i \leq l$
    such that $\psi_i(x)$ is in a type represented by $G^A$.
    We call $I_m$ the set of indices $1 \leq i \leq l$
    such that $\exists X. \psi_i^X(x)$
    is in a type represented by $G^A$, where
    $\psi_i^X(x)$ is the relativisation to the set variable $X$
    of $\psi_i$.

    Let $B$ be a structure such that
    $\SpecterM{R}{Q}{K}(A)$
    is contained in
    $\SpecterM{R}{Q}{K}(B)$.
    By definition, there exists
    sets $C^B$ and $G^B$
    in $B$ such that
    \begin{equation}
        \NeighbM{A}{C^A G^A}{R}{Q}
        =
        \NeighbM{B}{C^B G^B}{R}{Q}
        \, .
    \end{equation}

    As $Q$ is large enough
    to check distances up to $R$,
    the distance between two elements in $G^B$
    is greater than $2R$,
    and the distance between one element of $G^B$
    and one element of $C^B$ is greater than $2R$.

    Let us define
    $E \defined B \setminus \Neighb{B}{C^B}{R}$.
    Notice that $\Neighb{B}{G^B}{r} \subseteq E$
    because $2r \leq R_m \leq 2R_m = R$ and the distance between $E$ and $C^B$
    is greater than $2R_m$.
    Let $i \in I_m$, one can choose $k$ elements in $G^B$
    such that $B, b \models \exists X. \psi_i^X(x)$,
    let us call this vector $\vec{b}_i^m$.
    Let $i \in I_f$, one can choose $k$ elements in $G^B$
    such that $B, b \models \psi_i^X(x)$,
    let us call this vector $\vec{b}_i^f$.
    Without loss of generality since types in $G^B$ have
    $2 \times k \times l$ witnesses, we can assume that
    none of these vectors share elements. 
    Let $i \in I_m$ and $b \in \vec{b}_i^m$,
    there exists $b \in F_b \subseteq \Neighb{B}{b}{r}$
    such that $\Neighb{B}{b}{r} \cap F_b, b \models \psi_i(x)$.
    Let us build $E'$ as the structure $E$ where the complementaries
    of the sets $F_b$ have been removed, i.e.
    $E' \defined E \setminus \bigcup_{i \in I_m} \bigcup_{b \in \vec{b}_i^m}
    \Neighb{B}{b}{r} \setminus F_b$.

    We assert that for every $1 \leq i \leq l$
    the following properties are equivalent
    \begin{enumerate}[(i)]
        \item $A \uplus E' \models \theta_i$,
        \item $\Neighb{A}{C^A}{R_m} \uplus E' \models \theta_i$,
        \item $\Neighb{B}{C^B}{R_m} \uplus E' \models \theta_i$.
    \end{enumerate}

    Since $\NeighbM{A}{C^A}{R}{Q} = \NeighbM{B}{C^B}{R}{Q}$
    and $r \leq R_m \leq R$, it is clear that 
    $\NeighbM{A}{C^A}{R_m}{Q} = \NeighbM{B}{C^B}{R_m}{Q}$.
    Therefore,
    $\Neighb{A}{C^A}{R_m} \uplus E'$ and
    $\Neighb{B}{C^B}{R_m} \uplus E'$ satisfy the same first-order
    sentences at quantifier rank less than $Q$.
    As $Q \geq \rk \theta_i$ for $1 \leq i \leq l$, this proves the equivalence
    between $(iii)$ and $(ii)$.
    Let us now prove that $(i)$ is equivalent to $(ii)$.
    \begin{itemize}
        \item Assume $A \uplus E' \models \theta_i$.
            \begin{itemize}
                \item If all witnesses $a$ of $\theta_i$
                    are such that $\Neighb{A}{a}{r} \subseteq
                    \Neighb{A}{C^A}{R_m}$
                    or $a \in E'$ then we conclude.
                \item If some witness $a$ of $\theta_i$
                    is in $A$ but $\Neighb{A}{a}{r}$ is not included in
                    $\Neighb{A}{C^A}{R_m}$. By definition of $C^A$ the type
                    of $a$ is found at least $2 kl$ times in $G^A$.
                    As a consequence, the vector $\vec{b}_i^f$
                    of size $k$ is found in $G^B$ such that
                    $\Neighb{B}{\vec{b}_i^f}{r} \subseteq E'$
                    and
                    we have proven that $E' \models \theta_i$.
                    In turn, this implies $\Neighb{A}{C^A}{r_1} \uplus E' \models
                    \theta_i$.
            \end{itemize}
        \item Assume that $\Neighb{A}{C^A}{r_1} \uplus E' \models \theta_i$.
            \begin{itemize}
                \item If all witnesses $a$ of $\theta_i$
                    are such that $\Neighb{A}{a}{r} \subseteq
                    \Neighb{A}{C^A}{R_m}$
                    or $a \in E'$ then we conclude.
                \item If some witness $a$
                    is found in $\Neighb{A}{C^A}{R_m}$
                    but $\Neighb{A}{a}{r}$ is not included in
                    $\Neighb{A}{C^A}{R_m}$,
                    then the type of $a$ in $A$ is found at least
                    $2kl$ times in $G^A$.
                    Notice that $A, a \models \exists X. \psi_i^X(x)$,
                    where $X$ can be chosen as $\Neighb{A}{a}{r} \cap
                    \Neighb{A}{C^A}{R_m}$.

                    Let $b \in \vec{b}_i^m$,
                    $E', b \models \psi_i(x)$ because
                    $\Neighb{B}{b}{r}\cap F_b \models \psi_i(x)$
                    and $\Neighb{E'}{b}{r} = \Neighb{B}{b}{r} \cap F_b$.
                    As the elements of $\vec{b}_i^m$ are at pairwise distance
                    greater than $2R_m \geq 2r$.

                    We have proven that $E' \models \theta_i$,
                    hence, that $A \uplus E' \models \theta_i$.
            \end{itemize}
    \end{itemize}

    The assumption that $A \models \varphi$, which is closed under
    disjoint unions, 
    directly leads to $A \uplus E' \models \varphi$. Moreover, the equivalences
    above
    assert that $\Neighb{B}{C^B}{R_m} \uplus E' \models \varphi$.
    As $\Neighb{B}{C^B}{R_m} \uplus E' \subseteq_i B$
    and since $\varphi$ is preserved under $\subseteq_i$
    we conclude that $B \models \varphi$.

\end{proof}
\begin{ifappendix}
    \inducedminimalmodels*
\end{ifappendix}
\begin{proof}
    Applying  \cref{lem:failure:presmso}
    provides $R,Q,K$
    such that $\varphi$ is preserved
    over $\XClass$ under $(R,Q,K)$ $\MSO$ types.
    Using \cref{lem:pres:mso:existlocalnf},
    $\varphi$ is equivalent over $\XClass$ to a sentence
    $\psi = \exists \vec{x}. \theta(\vec{x})$
    where $\theta(\vec{x})$ is a $r$-local $\MSO$-formula.

    Consider $A$ a minimal model of $\psi$ with respect to $\subseteq_i$,
    in $\XClass$. As $A \models \psi$, there exists a vector $\vec{a} \in
    A^{|\vec{x}|}$
    such that $A, \vec{a} \models \theta(\vec{x})$,
    and since $\theta$ is $r$-local, this
    proves that $\Neighb{A}{\vec{a}}{r}, \vec{a} \models \theta(\vec{x})$.
    As a consequence, $\Neighb{A}{\vec{a}}{r} \models \psi$
    and lies in $\XClass$ at is it hereditary.
    The minimality of $A$ proves that
    $A = \Neighb{A}{\vec{a}}{r}$, thus, that
    $A \in \Balls{\XClass}{r}{|\vec{x}|}$ which is independent of $A$.
\end{proof}
    \begin{ifappendix}%
        \exwithoutdisjointunions*%
    \end{ifappendix}%
    \begin{proof}
        This class is \emph{not}
        closed under disjoint unions.
        The sentence $\varphi$ detects the cycles in $\mathcal{C}$ which are
        all incomparable maximal elements for $\subseteq_i$.
        As a consequence, $\varphi$ is monotone for $\subseteq_i$ in $\mathcal{C}$,
        and for every pair $k,r$
        there exists a
        minimal models for $\varphi$
        that is not in $\Balls{\mathcal{C}}{r}{k}$.
\end{proof}
\subsection{Appendix to ``Preservation Under Extensions'' (Section \ref{sec:presthm:localdecstruct})}\label{app:sec:presthm:localdecstruct}

\begin{figure}
\begin{center}

\begin{tikzpicture}
    \newcommand{\tikztemplategraph}{
        \node[draw, circle,added] (A) at (0.5,0) {};
        \node[draw, circle,temporary] (B) at (2,0) {$b_1$};
        \node[draw, circle,noise] (C) at (3.5,0) {};
        \node[draw, circle,substruct] (D) at (5,0) {};
        \node[draw, circle,substruct] (E) at (0,1) {};
        \node[draw, circle,substruct] (F) at (1,1) {};
        \node[draw, circle,substruct] (G) at (2,1) {};
        \node[draw, circle,noise] (H) at (3,1) {};
        \node[draw, circle,temporary] (I) at (4,1) {$b_2$};
        \node[draw, circle,substruct] (J) at (5,1) {};

        \draw[added] (A) -- (E);
        \draw[added] (A) -- (F);
        \draw[noise] (G) -- (H) -- (C) -- (I) -- (H);
        \draw[substruct] (E) -- (F) -- (G);
        \draw[temporary] (B) -- (G);
        \draw[temporary] (B) -- (A);
        \draw[temporary] (I) -- (D);
        \draw[temporary] (I) -- (J);
    }
    \begin{scope}[yshift=0cm,
        noise/.style={color=black},
        temporary/.style={color=blue,font=\tiny,inner sep=0pt},
        substruct/.style={color=red},
        added/.style={color=teal},
        ]
        \draw (3,-1) node {The whole structure $B\models\varphi$
        is centered around $(b_1,b_2)$.};
        \tikztemplategraph
    \end{scope}
    \begin{scope}[yshift=3cm,
        noise/.style={color=white},
        temporary/.style={color=blue,font=\tiny,inner sep=0pt},
        substruct/.style={color=red},
        added/.style={color=teal},
        ]
        \draw (3,-1) node {The extended substructure $A'\in\Balls{\XClass}{2}{2}$.};
        \tikztemplategraph
    \end{scope}
    \begin{scope}[yshift=6cm,
        noise/.style={color=white},
        temporary/.style={color=white},
        added/.style={color=white},
        substruct/.style={color=red}]
        \draw (3,-1) node {The induced substructure $A\models\theta$.};
        \tikztemplategraph
    \end{scope}

\end{tikzpicture}

\end{center}
\caption{Extracting a structure in $\Balls{\XClass}{2}{2}$
from an induced substructure $A$ of a larger structure $B$.}
\label{fig:presthm:extracting}
\end{figure}

    \begin{ifappendix}
        \localtarski*
    \end{ifappendix}
\begin{proof}
    Assume that the {\puext} holds on the spaces 
    $\Balls{\XClass}{r}{k}$ for $r,k \geq 1$, 
    the class $\XClass$,
    \cref{lem:pres:locwbpres}
    proves that $\XClass$ satisfies preservation under extensions.

    Conversely, assume that the {\puext} holds over $\XClass$.
    Let us fix $r,k \in \mathbb{N}^2$
    and consider
    a sentence $\varphi$ preserved under
    extensions over $\Balls{\XClass}{r}{k}$.

    Let $\psi \defined \exists x_1, \dots, x_k.
    \varphi_{|\Neighb{}{\vec{x}}{r}}$,
    i.e.
    the relativisation of $\varphi$ to a neighborhood
    of size $r$ around some elements $x_1, \dots, x_k$.
    It is left as an exercice to prove that $\psi$ and $\varphi$
    are equivalent on $\Balls{\XClass}{r}{k}$.

    Assume that $A, B \in \XClass^2$ are such that
    $A \models \psi$ and $A \subseteq_i B$.
    This amounts to the existence of a strong injective
    morphism $h \colon A \to B$.
    Since $A \models \psi$, there exists $\vec{a} \in A^k$
    (possibly with repetition) 
    such that $\Neighb{A}{\vec{a}}{r} \models \varphi$.
    Moreover, $\Neighb{A}{\vec{a}}{r} \subseteq_i
    \Neighb{B}{h(\vec{a})}{r}$ through the morphism $h$.
    As both neighborhoods are built using radius $r$ and
    tuples of $k$ elements, they belong to $\Balls{\XClass}{r}{k}$.
    The assumption on $\varphi$ implies
    that $\Neighb{B}{h(\vec{a})}{r} \models \varphi$
    and in turn
    this proves
    $B \models \psi$.

    Using the {\puext} over $\XClass$ on $\psi$
    gives an existential sentence $\theta$
    equivalent to $\psi$ over $\XClass$.
    In particular, $\theta$ is equivalent to $\psi$ over
    $\Balls{\XClass}{r}{k}$ since it is a subset,
    thus $\theta$ is equivalent to $\psi$ on $\Balls{\XClass}{r}{k}$.

    We have proven that {\puext} holds over $\Balls{\XClass}{r}{k}$.
\end{proof}

    \begin{ifappendix}
        \finiteballs*
    \end{ifappendix}
\begin{proof}
Notice that whenever $A$
has a $(r,m)$-scattered set
then it cannot be a member of $\Balls{\XClass}{r}{k}$ with $k < r$.
Conversely, if $A$ does not contain any $(r,m)$-scattered set,
then it is a union of $m-1$ balls of radius $r$.
Assume that $\XClass$ is wide, the set
$\Balls{\XClass}{r}{k}$ contains only structures of size less than $\rho(r,k)$
and is therefore finite.
Assume that $\Balls{\XClass}{r}{k}$ is finite for every $r,k \in \mathbb{N}$,
then we let $\rho(r,m) \defined \max \setof{|A|}{A \in \Balls{\XClass}{r}{k}, k
< m}$,
whenever $|A| > \rho(r,m)$, the structure $A$ cannot be in $\Balls{\XClass}{r}{k}$
with $k < m$, hence it has a $(r,m)$-scattered set.
\end{proof}
\subsection{Appendix to ``Preservation under Extensions on Locally Well-Behaved Classes'' (Section \ref{sec:presthm:localwb})}\label{app:sec:presthm:localwb}

    \begin{ifappendix}
        \locboundedtd*
    \end{ifappendix}
    \begin{proof}
    Assuming that for all $k$ there is a bound
    on the tree-depth of
    elements of $\Balls{\XClass}{r}{k}$, we define
    $\rho(r)$ to be
    the maximum of $\td(A)$
    for $A$ in $\Balls{\XClass}{r}{1}$.

    Conversely, assume that there exists
    an increasing function $\rho$ such that
    $\td(A) \leq \rho(r)$ for every $A \in \Balls{\XClass}{r}{1}$
    a simple induction on $A$ shows that
    $\td(A)  \leq \rho(r \times (2k + 1))$ for every $A \in \Balls{\XClass}{r}{k}$.
    \begin{itemize}
        \item If $A \in \Balls{\XClass}{r}{k}$ can be written $A_1 \uplus A_2$
            then $\td(A) = \max(\td(A_1), \td(A_2))$ and
            we conclude by induction hypothesis since $\rho$ is increasing.
        \item If $A$ is totally connected and in $\Balls{\XClass}{r}{k}$
            then it is included in a ball of radius $r \times (2k + 1)$
            hence the tree-depth is bounded by
            $\rho(r \times (2k + 1))$.
            \qedhere
    \end{itemize}
    \end{proof}

    \begin{ifappendix}
        \pointedgraphs*
    \end{ifappendix}
\begin{proof}
    The class $P\Delta_2$ is not locally well-quasi-ordered
    because it contains the infinite antichain of \emph{wheels},
    that are cycles with one added point connected to the whole cycle.
    Les us prove that it satisfies preservation under extensions.

    The class $C\Delta_2$ of
    elements of $\Delta_2$ with one added colour $c$
    satisfies preservation under extensions
    because it is hereditary, closed under disjoint unions 
    and locally finite.

    We will use generic stability properties of preservation
    theorems~\cite[e.g.][Section 5]{lopez2020preservation}
    to transport the theorem over $C\Delta_2$ to
    $P\Delta_2$.
    The class $I\Delta_2$ defined
    as elements of $C\Delta_2$ satisfying
    $\exists x. c(x)$ and $\forall x,y. c(x) \wedge c(y) \implies x = y$
    satisfies preservation under extensions
    as it is a Boolean combination of definable upwards closed
    subsets of $C\Delta_2$.

    Let us define
    $I \colon I\Delta_2 \to P\Delta_2$
    the first-order interpretation
    via $I_E(x,y) \defined c(x) \vee c(y) \vee E(x,y)$.
    The map $I$ is surjective
    and monotone with respect to $\subseteq_i$,
    as a consequence, $P\Delta_2$
    satisfies preservation under extensions.
\end{proof}

\begin{figure}[t]
    \centering
    \begin{tikzpicture}[scale=0.5, circle/.style = {
            inner sep=2pt,
        }]
        \def \diamondsize {5}
        \def \diamondprop {360 / \diamondsize}
        \newcommand{\tpath}[3]{
            \pgfmathsetmacro{\moinsun}{int(#1 - 1)}
            \pgfmathsetmacro{\moinsdeux}{int(#1 - 2)}
            \foreach \i in {1,...,\moinsun} {
                \pgfmathsetmacro{\coefa}{\i / #1}
                \pgfmathsetmacro{\coefb}{1 - \i / #1}
                \node[draw,circle] (#2#3p{\i}) at
                    (barycentric cs:#2=\coefa,#3=\coefb)
                    {};
            }
            \draw (#3) -- (#2#3p{1});
            \draw (#2#3p{\moinsun}) --  (#2);
            \foreach \i in {1,...,\moinsdeux} {
                \pgfmathsetmacro{\plusun}{int(\i + 1)}
                \draw (#2#3p{\i}) --  (#2#3p{\plusun});
            }
        }
        \node[draw,circle] (A) at (-3,0) {};
        \node[draw,circle] (B) at (3,0)  {};
        \node[draw,circle] (C) at (0,2)  {};
        \node[draw,circle] (D) at (0,-2) {};

        \tpath{3}{A}{B};
        \draw (A)  .. controls (0,3.5) .. (B);
        \tpath{4}{C}{A};
        \tpath{4}{C}{B};
        \tpath{4}{D}{A};
        \tpath{4}{D}{B};

        \foreach \y in {1,...,2} {
            \tpath{4}{C}{ABp{\y}}
            \tpath{4}{D}{ABp{\y}}
        }

    \end{tikzpicture}
    \caption{An element of $\mathsf{\uplus D}$, $D_4$}
    \label{fig:diamond_class}
\end{figure}

\end{ifappendix}

\end{document}